\documentclass[11pt]{article} 
\usepackage{fullpage}

\usepackage{mathpazo} 

\usepackage{amsmath,amsfonts,amssymb}
\usepackage{amsthm}
\usepackage[colorlinks]{hyperref}
\usepackage[nameinlink]{cleveref}
\hypersetup{colorlinks={true},linkcolor={blue},citecolor=magenta}
\usepackage{ninecolors}
\usepackage{graphicx} 
\usepackage{diagbox}
\usepackage[T1]{fontenc}
\usepackage{physics}
\usepackage{todonotes}
\usepackage{lipsum}
\usepackage{multicol}
\usepackage{bbm}
\usepackage{longfbox}
\usepackage{enumitem}
\usepackage{float} 
\usepackage{nicefrac} 
\usepackage{dsfont}
\usepackage{longtable} 
\usepackage{circuitikz} 
\usepackage{caption}
\usepackage{subcaption}
\usepackage{mathtools} 
\allowdisplaybreaks
\usepackage{comment}
\usepackage[style=alphabetic,minalphanames=3,maxalphanames=4,maxnames=99,backref=true]{biblatex}
\newtheorem{theorem}{Theorem}[section]

\newtheorem{definition}[theorem]{Definition}
\newtheorem{lemma}[theorem]{Lemma}

\newtheorem{myclaim}[theorem]{Claim}
\newtheorem{remark}[theorem]{Remark}
\newtheorem{corollary}[theorem]{Corollary}

\newtheorem{construction}[theorem]{Construction}
\newtheorem{fact}[theorem]{Fact}
\newtheorem{notation}[theorem]{Notation}

\Crefname{construction}{Construction}{Constructions}
\Crefname{fact}{Fact}{Facts}
\Crefname{myclaim}{Claim}{Claims}
\Crefname{notation}{Notation}{Notations}

\newcommand{\myparagraph}[1]{\vspace{.5em} \noindent \textbf{#1.}\,}

\newcommand{\resp}{resp.,\ }
\newcommand{\ie}{i.e.,\ }

\newcommand{\wrt} {with respect to\ }

\newcommand{\st}{\ \text{s.t.}\ }

\newcommand{\N}{\mathbb{N}}

\renewcommand{\bra}[1]{\langle#1\rvert}
\renewcommand{\braket}[2]{\langle #1 \mid #2 \rangle}
\renewcommand{\ket}[1]{\lvert#1\rangle}
\newcommand{\set}[1]{\{ #1 \}}
\newcommand{\bit}{\{0,1\}}

\newcommand{\cB}{{\mathcal B}}

\newcommand{\cD}{{\mathcal D}}

\newcommand{\cF}{{\mathcal F}}
\newcommand{\cG}{{\mathcal G}}
\newcommand{\cH}{{\mathcal H}}
\newcommand{\cI}{{\mathcal I}}

\newcommand{\cK}{{\mathcal K}}

\newcommand{\cO}{{\mathcal O}}
\newcommand{\cP}{{\mathcal P}}

\newcommand{\cR}{{\mathcal R}}
\newcommand{\cS}{{\mathcal S}}

\newcommand{\cU}{{\mathcal U}}

\newcommand{\bfk}{\mathbf{k}}

\newcommand{\bfz}{\mathbf{z}}

\newcommand{\sfG}{\mathsf{G}}

\newcommand{\secp}{{\lambda}}
\newcommand{\poly}{\mathsf{poly}}

\newcommand{\Exp}{\operatorname*{\mathbb{E}}}
\newcommand{\Ex}{\Exp}

\newcommand{\Enc}{\operatorname{Enc}}

\newcommand{\Dec}{\operatorname{Dec}}

\newcommand{\negl}{\mathsf{negl}}
\newcommand{\Supp}{\operatorname{Supp}}

\newcommand{\TD}{\mathsf{TD}}

\usepackage[normalem]{ulem}

\newcommand{\E}{\mathop{\mathbb{E}}}

\newcommand{\Haar}{\mathcal{H}}

\renewcommand{\ketbra}[2]{\ket{#1}\bra{#2}}
\renewcommand{\braket}[2]{\langle #1 | #2 \rangle}

\newcommand{\wt}{\widetilde}

\renewcommand{\bra}[1]{\langle#1\rvert}
\renewcommand{\braket}[2]{\langle #1 | #2 \rangle}
\renewcommand{\ket}[1]{\lvert#1\rangle}
\renewcommand{\ketbra}[2]{\ket{#1}\!\bra{#2}}
\newcommand{\proj}[1]{\ket{#1}\!\bra{#1}}
\newcommand{\reg}[1]{\mathsf{#1}}

\newcommand{\id}{\mathrm{id}}

\newcommand{\haarstates}{\Haar}
\newcommand{\haarunitaries}{\mu}
\newcommand{\Symgp}{\mathsf{Sym}}

\newcommand{\Adversary}{{\cal A}}

\newcommand{\rel}{\mathcal{R}}

\newcommand{\dist}{\mathrm{dist}}

\newcommand{\indic}{\mathds{1}}

\newcommand{\Rinj}{\cR^{\cI \text{-} \mathrm{dist}}}
\newcommand{\RIdist}{\Rinj}
\newcommand{\RDdist}{\cR^{\cD \text{-} \mathrm{dist}}}

\newcommand{\Sspru}{\sfG}
\newcommand{\Ospru}{\cS}
\newcommand{\Osprus}{\cD}

\newcommand{\isolate}{\sf isolate}
\newcommand{\source}{\sf source}
\newcommand{\target}{\sf target}
\newcommand{\pair}{\sf pair}

\newcommand{\extract}{\mathsf{extract}}
\newcommand{\hyb}{\mathbf{H}}

\renewcommand{\op}{\mathrm{op}}

\newcommand{\Dom}{\operatorname{Dom}}

\newcommand{\rcolor}[1]{{\color{red4} #1}}
\newcommand{\bcolor}[1]{{\color{blue4} #1}}
\newcommand{\gcolor}[1]{{\color{green4} #1}}

\title{Pseudorandom Unitaries in the Haar Random Oracle Model\footnote{A preliminary version of a merge of this paper and~\cite{ABGL25a} appears in the proceedings of \textit{the 45th Annual International Cryptology Conference} (\textsc{CRYPTO~2025})~\cite{CRYPTO:ABGL25}. This is Part~I of the full version.}}

\author{Prabhanjan Ananth\thanks{\texttt{prabhanjan@cs.ucsb.edu}}\\ \small{UCSB} \and John Bostanci\thanks{\texttt{johnb@cs.columbia.edu}}\\ \small{Columbia} \and Aditya Gulati\thanks{\texttt{adityagulati@ucsb.edu}}\\ \small{UCSB} \and Yao-Ting Lin\thanks{\texttt{yao-ting\_lin@ucsb.edu}}\\ \small{UCSB}}

\date{}

\begin{document}

\maketitle

\begin{abstract}
    \noindent The quantum Haar random oracle model is an idealized model where every party has access to a single Haar random unitary and its inverse.  We construct strong pseudorandom unitaries in the quantum Haar random oracle model. This strictly improves upon prior works who either only prove the existence of pseudorandom unitaries in the inverseless quantum Haar random oracle model [Ananth, Bostanci, Gulati, Lin, EUROCRYPT 2025] or prove the existence of a weaker notion (implied by strong pseudorandom unitaries) in the quantum Haar random oracle model [Hhan, Yamada, 2024]. Our results also present a viable approach for building quantum pseudorandomness from random quantum circuits and analyzing pseudorandom objects in nature.
\end{abstract}

\newpage
\tableofcontents
\newpage

\section{Introduction}
\noindent Pseudorandom unitaries~\cite{JLS18} (PRUs) are efficiently computable unitaries that satisfy the following property: any quantum polynomial time adversary cannot distinguish whether it has oracle access to a PRU or a unitary sampled from the Haar measure. They generalize the notion of $t$-unitary designs, wherein the number of oracle adversarial queries is upper bounded by $t$ although there are no restrictions on the computational power of the adversary. Both $t$-designs and pseudorandom unitaries have played a vital role in quantum information science and its interplay with other areas. Notably, they have applications to black-hole physics, quantum benchmarking, lower bounds on space complexity in learning theory and finally, in cryptography. 
\par Exploring the design of quantum pseudorandom primitives has been an ongoing active research direction. 
Kretschmer's~\cite{Kretschmer21,KQST23} work suggests one-way functions might plausibly be a stronger assumption than PRUs. This opens the door of basing pseudorandom unitaries on weaker assumptions than one-way functions. 
Despite this oracle separation, actually building PRUs from one-way functions has proven a challenging problem.  It was only last year that a series of works~\cite{LQSYZ24,AGKL,MPSY24,ABFGVZZ24}, culminating in the work of Ma and Huang~\cite{MH25}, were able to finally establish the feasibility of PRUs from classical cryptographic assumptions. In particular, the concurrent works by~\cite{MPSY24,ABFGVZZ24} designed selectively secure PRUs from one-way functions and the subsequent work by~\cite{MH25} achieved the stronger adaptive security under the same assumption. 
\par In addition to studying the relation between classical cryptography and quantum pseudo-randomness, there have been two other lines of works attempting to study the properties of pseudo-random unitaries. The first line of work~\cite{BHHP25} posits a concrete quantum assumption that gives rise to PRUs, but is plausibly weaker than one-way functions.\footnote{We note that in a later version of~\cite{BHHP25}, the claim regarding the existence of PRUs under their proposed assumption was retracted; at present, the existence of PRUs under their assumption is currently only conjectured.} However, these kinds assumptions are yet to be tested and not considered standard by the community at the moment. Another line of work studies properties of pseudorandom unitaries in idealized models. One such idealized model that has recently been garnering some interest~\cite{BFV19,CM24,ABGL24,HY24} is the quantum Haar random oracle model (QHROM). In this model, which is a quantum analogue of the classical random oracle model, all the parties have oracle access to $U,U^{\dagger}$, where $U$ is drawn from the Haar measure. This model can especially come in handy when analyzing cryptographic constructions from random circuits. Our work is geared towards understanding the feasibility of pseudorandom unitaries in the quantum Haar random oracle model.

 \paragraph{Haar Random Oracle Model: Prior Work.} The question of investigating the possibility of quantum pseudorandomness in QHROM was first initiated by Bouland, Fefferman and Vazirani~\cite{BFV19}, who proposed the construction of pseudorandom state generators\footnote{Informally speaking, pseudorandom states (PRS)~\cite{JLS18} are efficiently computable states that are computationally indistinguishable from Haar random states. Importantly, the computational indistinguishability should hold even if the adversary receives many copies of the state. The existence of PRS is implied by the existence of pseudorandom unitaries.} without proof in QHROM. Recently, two independent and concurrent works~\cite{HY24,ABGL24} made further progress and presented provably secure constructions of quantum pseudorandom primitives in QHROM. Hhan and Yamada~\cite{HY24} showed that pseudorandom function-like state generators\footnote{Pseudorandom function-like states (PRFS)~\cite{AGKL}, a generalization of pseudorandom states, allows for generation of many pseudorandom states, each indexed by a binary string, using the same key. The existence of PRFS is implied by the existence of pseudorandom unitaries.}~\cite{AQY22} exist in QHROM. Ananth, Bostanci, Gulati and, Lin~\cite{ABGL24} showed that pseudorandom unitaries existence in a weaker variant of QHROM, referred to as inverseless quantum Haar random oracle model (iQHROM). In this variant, all the parties receive oracle access to only the Haar unitary $U$ (but not its inverse). These works have left a big gap in our understanding on the existence of pseudorandom unitaries in QHROM. A priori it should not even be clear whether pseudorandom unitaries exist in QHROM (i.e., with inverses).~\cite{ABGL24} showed that in the inverseless QHROM setting, any PRU construction making one parallel query to the Haar unitary is insecure. They also provided a matching upper bound and showed that two sequential queries suffice. In contrast, the minimum number of parallel queries needed for the existence of PRU in QHROM has not been established so far.   

 \paragraph{Why study the QHROM?} Every cryptographic model should be subjected to scrutiny and the QHROM is no different. A potential criticism of the QHROM is that constructions proven secure in the QHROM often offer little to no security guarantees when these constructions are realized in the real world. It is important to note that this line of skepticism is not new and is often used to attack the classical random oracle model. While this does suggest that we need to often exercise care when using the QHROM, or its predecessor the classical random oracle model, to justify the security of cryptographic constructions, these models still offer useful insights into properties of heuristic constructions which we otherwise do not have the tools to analyze. Let us take some concrete examples. Random circuits, which are circuits composed of 1-qubit and 2-qubit Haar unitaries, are popularly used in quantum benchmarking and for quantum advantage experiments~\cite{boixo2018characterizing,bouland2019complexity,movassagh2019quantum}. While for deep enough circuits, they are commonly supposed to be indistinguishable from Haar random unitaries, it is often not clear how to reduce the security of cryptographic constructions using random circuits to concrete cryptographic assumptions. Another example is the modeling of physical processes such as black-hole dynamics.  A long line of works have posited that black-holes possess information scrambling and thermalization properties similar to Haar random unitaries, but an exact formulation of black-hole dynamics has yet to be found. These kinds of situations, where scientists suspect that objects posses random-like features, but lack a complete model, are well suited to analysis in the QHROM.
 \par Another reason to study the QHROM model is that, perhaps surprisingly, of its implications to the plain model.~\cite{ABGL24}, leveraged the result of PRU in the inverseless QHROM, to show that any PRU can be transformed into one where the key length is much shorter than the output length. To put this result in context, there have been a few works in the past that have explored the tradeoff between the output length and the key size of quantum pseudorandom primitives. Gunn, Ju, Ma and Zhandry~\cite{GJMZ23} showed that the existence of any (multi-copy) pseudorandom state generator implies the existence of a pseudorandom state generator where the output length is strictly longer than the key size, as long as the adversary receives only one copy of the state. Extending this result to achieve a transformation that preserves the number of copies is an interesting open question. Recently, Levy and Vidick~\cite{levy2024prs} achieved some limited results in this direction but fell short of resolving the question. On the other hand,~\cite{ABGL24} showed that in the context of forward-only pseudorandom unitaries, such a transformation -- that is, generically transforming a multi-query PRU into a PRU with short keys -- is indeed possible.  

  \paragraph{Our Work.} The overarching theme of our work and related works is to address the following question: {\em What are the cryptographic implications of the quantum Haar random oracle model?} This includes constructing useful cryptography, such as PRUs, in the QHROM, but also adopting insights from studying the QRHOM to get novel results in the plain model. 

In this work, we construct strong pseudo-random unitaries in the QHROM.

\begin{theorem}[PRUs in the QHROM]
     Strong pseudo-random unitaries exist in the quantum Haar random oracle model.
\end{theorem}

\noindent We note that our construction of strong PRUs, presented in \Cref{sec:spru:construction}, is simple to describe, only makes two queries to the Haar random oracle, and only requires sampling $O(\log^{1+\epsilon}(n))$ bits of randomness, for any constant $\epsilon$, to get security against all adversaries making $\poly(n)$ queries to the strong PRU and Haar random oracle (and their inverses).

\subsection{Related Works}

\paragraph{Quantum Pseudorandomness.}\ 
Work on quantum pseudo-randomness began with the paper of \cite{JLS18}, which first defined pseudorandom states and unitaries.  They also presented the first constructions of pseudo-random states from one-way functions, and presented candidate constructions of pseudo-random unitaries from one-way functions without any security proof.  In the years after the paper of \cite{JLS18}, several works made progress towards building pseudo-random unitaries by considering related pseudo-random objects with security against adversaries with restricted queries~\cite{LQSYZ24,AGKL,BM24}.  

\cite{MPSY24} presented the first construction of non-adaptively secure pseudorandom unitaries (that is, secure against adversaries who only make a single parallel query to the PRU) via the so-called $\mathsf{PFC}$ ensemble, i.e. ``(random) Permutation-(random) Function-(random) Clifford''.  Simultaneously, \cite{CDXBBH24} presented an alternate construction of non-adaptively secure PRUs using random permutations.  In a breakthrough paper, \cite{MH25} proved that the $\mathsf{PFC}$ ensemble and related $\mathsf{C}^{\dagger}\mathsf{PFC}$ ensemble yielded a construction of PRUs and strong PRUs from one-way function.  The paper extended the compressed oracle technique for random functions from \cite{zhandry2019record} to the path-recording formalism for Haar random unitaries.  Since then the path-recording formalism has been used to show that the repeated $\mathsf{FHFHF}\ldots$ ensemble, originally conjectured to be a PRU in \cite{JLS18}, is secure~\cite{BHHP25}.  Using the security of the $\mathsf{PFC}$ ensemble, combined with a novel gluing lemma, \cite{schuster2025random} showed that inverseless PRUs can be formed in surprisingly low depth, assuming sub-exponential LWE.

\paragraph{Idealized Models in the Quantum World.}\ 
The first works on idealized models in the quantum world studied the quantum random oracle model (QROM), in attempt to prove the post-quantum security of a number of classical cryptographic primitives~\cite{boneh2011random,zhandry2015note,zhandry2015secure,targhi2016post, eaton2017leighton, zhandry2019record}.  In this model, parties can make superposition queries to a \emph{classical} random function.  

The quantum auxiliary state model~\cite{MNY24,Qian24}, and related common Haar random state (CHRS) model~\cite{AGL24,chen2024power} are models of computation in which all parties have access to an arbitrary polynomial many copies of a common quantum state.  This is meant to be the quantum equivalent to the common reference string model.  \cite{MNY24,Qian24} show that quantum bit commitments exist in the quantum auxiliary state model, and both \cite{AGL24,chen2024power} show that bounded copy pseudo-random states with short keys exist in the CHRS (which in turn imply the existence of quantum bit commitments).  Furthermore, \cite{AGL24} rules out quantum cryptography with classical communication in the CHRS, and \cite{chen2024power} rule out unbounded copy pseudo-random states.  More recently, these models have been used to provide oracle separations between one-way puzzles and efficiently verifiable one-way state generators and other quantum cryptographic primitives~\cite{BCN25,BMMMY25}.  While this idealized model provides interesting constructions of and black-box separations between primitives, the model can be problematic to instantiate in a realistic setting.  For example, instantiating the model in the real world might require a complicated multi-party computation to compute a shared quantum state, or a trusted third party who can distribute copies of the state.  

The quantum Haar random oracle model (QHROM) was first formally introduced in \cite{CM24}, who provided a construction of succinct commitments in the QHROM. However, \cite{CM24} was not able to analyze the security of their construction in the QHROM.  Separately, \cite{BFV19} considered the QHROM as an idealized model of black hole scrambling, and provided a construction of pseudo-random states in the QHROM.  Similar to \cite{CM24}, \cite{BFV19} present their construction without a security proof, although they sketch how a proof might proceed.  \cite{HY24} later showed that pseudo-random states and the related pseudo-random function-like state generators exist in the QHROM.  Their analysis involves heavy use of the Weingarten calculus and the approximate orthogonality of permutations.  Separately, \cite{ABGL24} consider a modification to the QHROM called the inverseless QHROM (iQHROM).  They construct forward-only PRUs in this model and provide a security proof using the path-recording formalism of \cite{MH25}.  Finally, \cite{Kretschmer21} considers another variant of the QHROM where all parties have access to an exponential number of Haar random unitaries.  They show that in this model, pseudo-random unitaries exist but one-way functions can be broken with $\mathsf{PSPACE}$ computation.  
\section{Technical Overview}

\paragraph{Ma-Huang's Path Recording Framework.} Before we recall the isometries described by~\cite{MH25}, we first set up some notation. A relation $R$ is defined as a \emph{multiset} $R = \set{(x_1,y_1),\ldots,(x_t,y_t)}$ of ordered pairs $(x_i,y_i) \in [N] \times [N]$, for some $N \in \mathbb{N}$. For any relation $R = \set{(x_1,y_1),\ldots,(x_t,y_t)}$, we say that $R$ is \emph{$\cD$-distinct} if the first coordinates of all elements are distinct, and \emph{injective} or \emph{$\cI$-distinct} if the second coordinates are distinct. For a relation $R$, we use $\Dom(R)$ to denote the \emph{set} $\Dom(R) \coloneqq \set{x \colon x \in [N], \, \exists \, y \st (x,y) \in R}$ and $\Im(R)$ to denote the \emph{set} $\Im(R) \coloneqq \set{y \colon y \in [N], \, \exists \, x \st (x,y) \in R}$. 

We define the following two operators (which are partial isometries) such that for any relations $L, R$,\footnote{For an $\cI$-distinct or $\cD$-distinct relation $L = \{(x_1,y_1), \dots, (x_t, y_t) \}$, the corresponding \emph{relation state} $\ket{L}$ is defined to be 
\[
\ket{L} := \frac{1}{\sqrt{t!}} \sum_{\pi \in \Symgp_t} \ket{x_{\pi^{-1}(1)}} \ket{y_{\pi^{-1}(1)}} \dots \ket{x_{\pi^{-1}(t)}} \ket{y_{\pi^{-1}(t)}}.
\] 
In~\cite{MH25}, relation states are defined for arbitrary relations, whereas we will not require them in this work.}
\[
V^L \colon \ket{x}_{\reg{A}}\ket{L}_{\reg{S}}\ket{R}_{\reg{T}} \mapsto \frac{1}{\sqrt{N - |\Im(L\cup R)|}} \sum_{y \notin \Im(L\cup R)} \ket{y}_{\reg{A}} \ket{L \cup \{(x,y)\}}_{\reg{S}} \ket{R}_{\reg{T}},
\]
\[
V^R \colon \ket{y}_{\reg{A}} \ket{L}_{\reg{S}}\ket{R}_{\reg{T}} \mapsto \frac{1}{\sqrt{N - |\Dom(L\cup R)|}} \sum_{x \notin \Dom(L\cup R)} \ket{x}_{\reg{A}} \ket{L}_{\reg{S}} \ket{R \cup \{(x,y)\}}_{\reg{T}}.
\]
\noindent Using $V^L$ and $V^R$, they define the following partial isometry:
\[
V = V^L \cdot (\id - V^R \cdot V^{R,\dagger}) + (\id - V^L \cdot V^{L,\dagger}) \cdot V^{R,\dagger}.
\]
They then showed that oracle access to a Haar random unitary $U$ and its inverse $U^\dagger$ can be simulated by $V$ and $V^\dagger$, respectively. In more detail, consider any oracle algorithm $\Adversary$ described by a sequence of unitaries $\left( A_1,A_2,\dots,A_{2t}\right)$ such that $\Adversary$ alternatively makes $t$ forward queries and $t$ inverse queries. The final state of $\Adversary$ with oracle access to (fixed) $U,U^{\dagger}$ is denoted by
\[
\ket{\Adversary_t^{U, U^{\dagger}}}_{\reg{AB}} \coloneqq \prod_{i=1}^{t} \left( U^{\dagger} A_{2i} U A_{2i-1} \right) \ket{0}_{\reg{A}} \ket{0}_{\reg{B}},
\]
where $\reg{A}$ is the adversary's query register, $\reg{B}$ is the adversary's auxiliary register, and each $A_i$ acts on $\reg{AB}$. They then consider the final joint state of $\Adversary$ and the purification after interacting with $V,V^{\dagger}$:
\[
\ket{\Adversary_t^{V,V^{\dagger}}}_{\reg{ABST}} \coloneqq \prod_{i=1}^{t} \left( V^{\dagger} A_{2i} V A_{2i-1} \right) \ket{0}_{\reg{A}}\ket{0}_{\reg{B}} \ket{\varnothing}_{\reg{S}} \ket{\varnothing}_{\reg{T}}.
\]
\cite{MH25} showed that $\rho_{{\sf Haar}}$ is $O(t^2/N^{1/8})$-close to $\rho_{{\sf MH}}$ in trace distance, where 
\[
\rho_{{\sf Haar}} \coloneqq \E_{U \sim \haarunitaries_n} \left[\ketbra{\Adversary_t^{U,U^{\dagger}}}{\Adversary_t^{U,U^{\dagger}}}_{\reg{AB}}\right]
\quad \text{and} \quad
\rho_{{\sf MH}} \coloneqq \Tr_{\reg{ST}} \left( \ketbra{\Adversary_t^{V,V^{\dagger}}}{\Adversary_t^{V,V^{\dagger}}}_{\reg{ABST}}\right),
\]
$\haarunitaries_n$ denotes the Haar measure over $n$-qubit unitaries and $N = 2^n$.

\subsection{Strong PRUs in the QHROM}

We begin with our construction of strong PRU in the QHROM. On input the security $\secp \in \N$ and key $k \in \bit^{3\secp}$, our construction is described as follows:
\begin{align*}
    G^U(1^\secp,k) \coloneqq X^{k_3} \cdot U_\secp \cdot X^{k_2} \cdot U_\secp \cdot X^{k_1},
\end{align*}
where $U_\secp$ is the $\secp$-qubit Haar random oracle and the key is written as $k \coloneqq k_1||k_2||k_3$, \ie the concatenation of three $\secp$-bit strings. 

To analyze our construction, we consider an adversary $\Adversary$ that has oracle access to $\cO_1, \cO_2$ and their respective inverses $\cO^{\dagger}_1, \cO^{\dagger}_2$. The adversary operates on the registers $\reg{A}\reg{B}$ and is parameterized as a sequence of unitaries $\Adversary = (A_1, A_2, \dots, A_{4t})$ each acting on $\reg{A}\reg{B}$. Without loss of generality, we assume that all oracle queries are made on the register $\reg{A}$, following a fixed sequence of interactions: first querying $\cO_1$, then $\cO_2$, followed by their inverses $\cO^{\dagger}_1$ and $\cO^{\dagger}_2$. The final state of the adversary after making $t$ queries to each oracle is given by:
\[
\ket{\Adversary_t^{\cO_1, \cO_2, \cO^{\dagger}_1, \cO^{\dagger}_2}}_{\reg{AB}} = \prod_{i=1}^t \left( \cO^{\dagger}_2 \cdot A_{4i} \cdot \cO^{\dagger}_1 \cdot A_{4i-1} \cdot \cO_2 \cdot A_{4i-2} \cdot \cO_1 \cdot A_{4i-3} \right) \ket{0}_{\reg{A}} \ket{0}_{\reg{B}}.
\]
In the ideal experiment, $\cO_1 = U_1$ and $\cO_2 = U_2$ where $U_1, U_2$ are independently sampled from the Haar distribution. In the real experiment, $\cO_1 = U$ and $\cO_2 = G^U(1^\secp,k)$ where $U$ is sampled from the Haar distribution and $k$ are sampled uniformly from $\bit^{3\secp}$.

\subsubsection{Using the Path Recording Framework}
The first step in our analysis is to replace the Haar-random unitaries with \textit{path-recording isometries}.

\paragraph{Ideal Experiment:}  
In the ideal experiment, we replace the two independent Haar-random unitaries $U_1$ and $U_2$ with two corresponding \textit{path-recording isometries} $V_1$ and $V_2$. These path-recording isometries operate on \textit{independent databases} stored in the purification registers, denoted as $(L_1, R_1)$ and $(L_2, R_2)$, respectively. 

\paragraph{Real Experiment:}  
In the real experiment, we replace the Haar-random oracle $U$ with a \textit{single path-recording isometry} $V$. Unlike the idealized case, here the purification register maintains a \textit{combined database} $(L, R)$, along with the \textit{keys} $(k_1, k_2, k_3)$ for strong PRU construction.\\

\noindent Note that at the end of the computation, the purification register is traced out. Consequently, if we can construct an isometry that acts on the purification register and that maps the output of one case close to the other, then the adversary's view remains statistically indistinguishable between the two cases.\footnote{This is because applying an isometry on the traced out registers does not change the final state.}

\subsubsection{Defining the Isometry $\Ospru$}
We define an approximate isometry $\Ospru$ such that it acts on the purification register of the adversary in the ideal experiment and maps it close to the state in the real experiment. Specifically, $\Ospru$ serves as a \textit{merge operator}, combining two independent databases into a single unified database while simultaneously simulating keys. 

\noindent Formally, $\Ospru$ acts on the auxiliary registers as follows:
\[
\Ospru: ((L_1, R_1), (L_2, R_2)) \mapsto (L, R, k_1, k_2, k_3)
\]
where:
\begin{itemize}
    \item  $(L, R)$ is the merged database containing all recorded queries
    \item $k_1, k_2, k_3$ are simulated keys compatible with the original PRU construction.
\end{itemize}

\noindent The key challenge in defining $\Ospru$ such that when it is applied to the ideal experiment, it closely mimics the real experiment, while still being an isometry (\ie reversible). We talk more about how we define $\Ospru$ such that it is reversible in~\Cref{sec:spru:def}. The rest of proof focuses on showing that $\Ospru$ maps the state in the real experiment to the state in the ideal experiment. 

\subsubsection{Progress Measure}

The main challenge in demonstrating that \(\Ospru\) approximately maps the state in the real experiment close to the one in the ideal experiment is the difficulty of obtaining a simple closed-form expression, as was possible in the inverseless setting (see~\cite{ABGL24}). Instead, we draw inspiration from the query-by-query analysis approach in the literature of the quantum random oracle model~\cite{zhandry2019record,CMS19,DFMS22}. Specifically, we do query-by-query analysis via defining the \emph{progress measure} as the adversary's distinguishing advantage after each query.

\noindent A key step in our analysis is to show that, for any state \(\ket{\psi}\) (generated using the ideal oracles), the process of first simulating the keys and then making a query to a real oracle (e.g., \(V\)) is close to making a query to a corresponding ideal oracle (e.g., \(V_1\)) first and then simulating the keys. Formally, we show that the following two states are close:
\begin{equation*}
    V \Ospru \ket{\psi}  \quad \text{and} \quad  \Ospru V_1 \ket{\psi}\,,
\end{equation*}
which we establish by proving that the operator norm bound 
\[
\|(V \Ospru - \Ospru V_1) \Pi_{\leq t}\|_{\op} = \negl(n)\,,
\]
where $\Pi_{\leq t}$ denotes the projector acting on the Hilbert spaces labeled by $\reg{S}_1\reg{T}_1\reg{S}_2\reg{T}_2$ that projects onto the space spanned by $\ket{L_1}_{\reg{S}_1}\ket{R_1}_{\reg{T}_1}\ket{L_2}_{\reg{S}_2}\ket{R_2}_{\reg{T}_2}$ such that $L_1, L_2 \in \Rinj_{\leq t}$ and $R_1, R_2 \in \RDdist_{\leq t}$.
\noindent Similarly, we extend this argument to show all the following quantities
\begin{enumerate}
    \item \( \|(V^\dagger \Ospru - \Ospru V_1^\dagger) \Pi_{\leq t}\|_{\op} \)
    \item \( \|(X^{k_3} V X^{k_2} V X^{k_1} \Ospru - \Ospru V_2) \Pi_{\leq t}\|_{\op} \)
    \item \( \|(X^{k_1} V^\dagger X^{k_2} V^\dagger X^{k_3} \Ospru - \Ospru V_2^\dagger) \Pi_{\leq t}\|_{\op} \)
\end{enumerate}
are negligible when $t(n) = \poly(n)$. As the main technical contribution of this work, the details can be found in~\Cref{sec:spru:1st_oracle} and~\Cref{sec:spru:2nd_oracle}. By establishing these bounds, we can inductively analyze the adversary’s distinguishing advantage after each query (for details, see~\Cref{sec:spru:induction}). Hence, we show that \(\Ospru\) approximately maps the state in the real experiment to the one in the ideal experiment.

\subsubsection{Simplifications for the Analysis}

To streamline our analysis, we introduce several simplifications that make computations more manageable. We outline some key steps below:

\paragraph{Leveraging Unitary Invariance of the Haar Measure:}  
In the real experiment, we have access to two oracles: \( U \) and \( X^{k_3} U X^{k_2} U X^{k_1} \), along with their inverses. These oracles are difficult to work with because the second oracle involves two calls to \( U \) and depends on all three keys, whereas the first oracle is comparatively simple. To balance the difficulty in analysis across both oracles, we use the unitary invariance of the Haar measure. Specifically, we apply the transformation 
\[
U \mapsto X^{k_3} U X^{k_1}
\]
and redefine the key \( k_2 \) as  
\[
k_2 \mapsto k_1 \oplus k_2 \oplus k_3\,.
\]
This effectively changes our oracle pair to \( X^{k_3} U X^{k_1} \) and \( U X^{k_2} U \) (along with their inverses), making the analysis easier.

\paragraph{Working with Non-Norm-Preserving Operators Instead of Isometries:}  
Explicitly maintaining normalization coefficients throughout the analysis can lead to unnecessary complications, especially when we only care about asymptotic behavior. To simplify calculations, we work with non-norm-preserving (\ie unnormalized) operators, and then establish that these operators remain close to isometries. For more details, see~\Cref{sec:spru:iso}.

\paragraph{Decoupling \( L \) and \( R \):}  
To further simplify our framework, we modify the operators so that \( L \) and \( R \) become completely independent. Instead of using the partial isometry
\[
V^L: \ket{x}_{\reg{A}} \ket{L}_{\reg{S}} \ket{R}_{\reg{T}} \mapsto \frac{1}{\sqrt{N - |{\rm Im}(L\cup R)|}} \sum_{y \notin \Im(L\cup R)} \ket{y}_{\reg{A}} \ket{L\cup \{(x,y)\}}_{\reg{S}} \ket{R}_{\reg{T}},
\]
we switch to
\[
F^L: \ket{x}_{\reg{A}} \ket{L}_{\reg{S}} \ket{R}_{\reg{T}} \mapsto \frac{1}{\sqrt{N}} \sum_{y \notin \rcolor{\Im(L)}} \ket{y}_{\reg{A}} \ket{L \cup \{(x,y)\}}_{\reg{S}} \ket{R}_{\reg{T}}.
\]
A similar transformation is applied to \( V^R \), replacing it with \( F^R \). The operator $F$ is defined analogously to $V$. Using techniques analogous to those in~\cite{MH25}, we show that these modified operators remain negligibly close to the original ones while significantly simplifying calculations (see~\Cref{sec:FvsV}).

\subsubsection{Overview of Hybrids} 
To prove that our strong pseudorandom unitary (PRU) construction is secure, we go through the following stages from the real experiment (an adversary querying the PRU and Haar random oracle) to the ideal experiment (an adversary querying two Haar random unitaries): 

\begin{itemize}
    \item \textbf{Real Experiment:} The adversary has oracle access to the Haar oracle and the PRU construction and their inverses:
    \[
    \cO_1 = U, \quad \cO_2 = X^{k_3} U X^{k_2} U X^{k_1}.
    \]
    \item \textbf{$H_1$:} By leveraging the unitary invariance of the Haar measure, we equivalently define the following oracles to balance complexity:
    \[
    \cO_1 = X^{k_3} U X^{k_1}, \quad \cO_2 = U X^{k_2} U.
    \]
    \item \textbf{$H_2$:} We replace the Haar-random unitary \( U \) with the path-recording isometry \( V \), allowing us to track queries explicitly:
    \[
    \cO_1 = X^{k_3} V X^{k_1}, \quad \cO_2 = V X^{k_2}V.
    \]
    \item \textbf{$H_3$:} We modify the path-recording isometry to ensure that the registers \( L \) and \( R \) are independent, making calculations simpler:
    \[
    \cO_1 = X^{k_3} F X^{k_1}, \quad \cO_2 = F X^{k_2} F.
    \]
    \item \textbf{$H_4$:} This is where most of our technical contributions lie. We introduce \( \Ospru \) to transition from separate databases to a merged structure while simulating keys. We use query-by-query analysis to show closeness of the $H_3$ and $H_4$:
    \[
    \cO_1 = F_1, \quad \cO_2 = F_2.
    \]
    \item \textbf{$H_5$:} We transition back from \( F_1, F_2 \) to standard path-recording isometries:
    \[
    \cO_1 = V_1, \quad \cO_2 = V_2.
    \]
    \item \textbf{Ideal Experiment:} Finally, we switch from path-recording isometries back to independent Haar-random unitaries:
    \[
    \cO_1 = U_1, \quad \cO_2 = U_2.
    \]
\end{itemize}

\section{Preliminaries} \label{sec:prelim}

We denote the security parameter by $\secp$. We assume that the reader is familiar with fundamentals of quantum computing, otherwise readers can refer to \cite{nielsen_chuang_2010}. We refer to $\negl(\cdot)$ to be a negligible function. 

\subsection{Notation}

\myparagraph{Sets and vectors} For $N \in \N$, we use the notation $[N]$ to refer to the set $\set{1, 2, \dots, N}$. For two binary strings $a, b$ of equal length, we define $a \oplus b$ as their bitwise XOR. For a set of binary strings $A \subseteq \bit^n$ and a binary string $b \in \bit^n$, we define
\[
A \oplus b \coloneqq \set{\, a \oplus b \colon a \in A \,}.
\]
For two sets $A,B \subseteq \bit^n$ of binary strings, we define
\[
A \oplus B \coloneqq \set{\, x \colon \exists a \in A, b \in B \, \st x = a \oplus b \,}.
\]
Given a set $A$ and $t \in \N$, we use the notation $A^t$ to denote the $t$-fold Cartesian product of $A$, and the notation $A^t_{\mathrm{dist}}$ to denote distinct subspace of $A^t$, i.e. the vectors in $A^t$, $\vec{y} = (y_1, \ldots, y_t)$, such that for all $i \neq j$, $y_i \neq y_j$. For any vector $\vec{x}$, we also define the set $\{\vec{x}\} := \bigcup_{i\in[t]} \{x_i\}$. We denote the $i$-th coordinate of $\vec{x}$ by $x_i$. For an ordered set $A$ and an element $x \in A$, we denote by $x \in_i A$ to mean $x$ is the $i$-th largest element in $A$. For any vector $\vec{x} \in A^t$, index $i \in [t]$, and element $y \in A$, let $\vec{x}^{\,(i \gets y)}$ denote the vector obtained by inserting $y$ into the $i$-th coordinate of $\vec{x}$ and shifting all subsequent coordinates one position to the right. For any vector $\vec{x} \in A^t$ and index $i \in [t]$, let $\vec{x}_{-i}$ denote the vector obtained by deleting its $i$-th coordinate and shifting all subsequent coordinates one position to the left.

\myparagraph{Quantum states and distances}
A register $\reg{R}$ is a named finite-dimensional Hilbert space.  If $\reg{A}$ and $\reg{B}$ are registers, then $\reg{AB}$ denotes the tensor product of the two associated Hilbert spaces.  We denote by $\mathcal{D}(\reg{R})$ the density matrices over register $\reg{R}$.  For $\rho_{\reg{AB}} \in \mathcal{D}(\reg{AB})$, we let $\Tr_\reg{B}(\rho_{\reg{AB}}) \in \mathcal{D}(\reg{A})$ denote the reduced density matrix that results from taking the partial trace over $\reg{B}$.  We denote by $\TD(\rho, \rho') = \frac{1}{2} \norm{\rho - \rho'}_1$ the trace distance between $\rho$ and $\rho'$, where $\norm{X}_1 = \Tr(\sqrt{X^{\dagger} X})$ is the trace norm. We use $\|\ket{\psi}\|_2 = \sqrt{\braket{\psi}{\psi}}$ to denote the Euclidean norm. For two pure (and possibly subnormalized) states $\ket{\psi}$ and $\ket{\phi}$, we use  $\TD\qty(\ket{\psi},\ket{\phi})$ as a shorthand for $\TD\qty(\proj{\psi},\proj{\phi})$. We also say that $A \preceq B$ if $B - A$ is a positive semi-definite matrix. We denote by $\haarstates_n$ the Haar distribution over $n$-qubit states, and $\haarunitaries_n$ the Haar measure over $n$-qubit unitaries (i.e. the unique left and right invariant measure).

\subsection{Cryptographic Primitives}

In this section, we define strong pseudo-random unitaries (strong PRU)~\cite{JLS18}, which are the quantum equivalent of a pseudorandom function, in that an adversary can not distinguish the strong PRU from a truly Haar random unitary, even with inverse access to both.

\begin{definition}[Adversaries with Forward and Inverse Access to Two Oracles]
An adversary $\Adversary$ with oracle access to two $n$-qubit unitaries $\cO_1, \cO_2$ and their inverses $\cO^{\dagger}_1, \cO^{\dagger}_2$ is defined as follows. $\Adversary$ has $n$-qubit query register $\reg{A}$ and a finite-size ancilla register $\reg{B}$, and always queries the oracles on the register $\reg{A}$. By padding with dummy queries, we assume that the adversary queries the oracles in the order $(\cO_1, \cO_2, \cO_1^{\dagger}, \cO_2^{\dagger})$. An adversary $\Adversary$ making $t$ queries to each oracle is parameterized by a sequence of unitaries $(A_1, A_2, \dots, A_{4t})$ acting on $\reg{AB}$. We denote the final state of the adversary as 
\[
\ket{\Adversary_t^{\cO_1, \cO_2, \cO^{\dagger}_1, \cO^{\dagger}_2}}_{\reg{AB}} 
:= \prod_{i=1}^{t} 
\left(\cO^{\dagger}_2 \cdot A_{4i} \cdot \cO^{\dagger}_1 \cdot A_{4i-1} \cdot \cO_2 \cdot A_{4i-2} \cdot \cO_1 \cdot A_{4i-3}\right)
\ket{0}_{\reg{A}}\ket{0}_{\reg{B}}.
\]
\end{definition}

\begin{definition}[Strong Pseudorandom Unitaries]\label{def:pru}
We say that $\set{\cG_\secp}_{\secp \in \N}$ is a strong pseudorandom unitary if, for all $\secp \in \N$, $\cG_\secp = \set{G_k}_{k \in \cK_\secp}$ is a set of $m(\secp)$-qubit unitaries where $\cK_\secp$ denotes the key space, satisfying the following:
\begin{enumerate}
    \item {\bf Efficient Computation:} There exists a $\poly(\secp)$-time quantum algorithm that implements $G_k$ for all $k \in \cK_\secp$.
    \item {\bf Indistinguishability from Haar:} For any quantum polynomial-time oracle adversary $\Adversary$, there exists a negligible function $\negl$ such that for all $\secp \in \N$,
    \begin{equation*}
        \left| \Pr_{k \gets \cK_\secp} \left[ 1 \gets \Adversary^{G_k, {G_k^\dagger}}(1^\secp) \right] - \Pr_{U \gets \mu_{m(\secp)}} \left[1 \gets \Adversary^{U, U^{\dagger}}(1^\secp) \right] \right|  
        \leq \negl(\secp)\,.
    \end{equation*} 
\end{enumerate}
In the $\mathsf{QHROM}$, both $G$ and $\Adversary$ have oracle access to an additional family of unitaries $\set{U_{\ell}}_{\ell \in \N}$ sampled independently from the Haar measure on $\ell$ qubits, and their inverses.
\end{definition}

\subsection{Useful Lemmas}
Here we present some standard lemmas.

\begin{lemma}   \label{lem:op_norm}
For any operator $A$ and vector $\ket{\psi}$, $\norm{A \ket{\psi}}_2 \leq \norm{A}_{\op} \norm{\ket{\psi}}_2$.
\end{lemma}

\begin{lemma}   \label{lem:op_norm_orthogonal}
Let $A$ be an operator and let $\cB$ be an orthonormal basis of the domain of $A$. 
If $A\ket{i}$ is orthogonal to $A\ket{j}$ for all distinct $\ket{i}, \ket{j} \in \cB$, 
then
\[
  \|A\|_{\op} = \max_{\ket{i} \in \cB} \|A \ket{i}\|_2.
\]
\end{lemma}
\begin{proof}
    For any normalized $\ket{\psi} = \sum_{\ket{i} \in \cB} \alpha_{i} \ket{i}$, we have
    \[
        \|A \ket{\psi}\|^2_2 
        = \big\|\sum_{\ket{i} \in \cB} \alpha_{i} \cdot A \ket{i}\big\|^2_2
        = \sum_{\ket{i} \in \cB} |\alpha_{i}|^2 \cdot \|A \ket{i}\|^2_2
        \leq \max_{\ket{i} \in \cB} \|A \ket{i}\|^2_2.  \qedhere
    \]
\end{proof}

\subsection{Path-Recording Framework} \label{sec:path_recording}

We recall the path-recording framework. The following definitions are taken from~\cite{MH25} with modest changes for our purposes. 

\myparagraph{Relations} Relations are an important part of the path recording framework, here we define relations between sets, as well as what it means to be injective and to take the inverse of a relation. A relation $R$ is defined as a \emph{multiset} $R = \set{(x_1,y_1), \dots, (x_t,y_t)}$ of ordered pairs $(x_i,y_i) \in [N] \times [N]$, for some $N \in \mathbb{N}$. For any relation $R = \set{(x_1,y_1),\ldots,(x_t,y_t)}$, we say that $R$ is \emph{$\cD$-distinct} if the first coordinates of all elements are distinct, and \emph{injective} or \emph{$\cI$-distinct} if the second coordinates are distinct. For a relation $R$, we use $\Dom(R)$ to denote the \emph{set} $\Dom(R) \coloneqq \set{x \colon x \in [N], \, \exists \, y \st (x,y) \in R}$ and $\Im(R)$ to denote the \emph{set} $\Im(R) \coloneqq \set{y \colon y \in [N], \, \exists \, x \st (x,y) \in R}$. For any $t \ge 0$, let $\rel_t$ denote the set of all relations of size~$t$. Let $\rel \coloneqq \bigcup_{i=0}^\infty \rel_t$. The size of a relation refers to the number of ordered pairs in the relation, including multiplicities. We denote this by $|R|$, as the size corresponds to the cardinality of $R$ viewed as a multiset. Let $\RIdist_t$ be the set of all $\cI$-distinct relations of size~$t$. Let $\RDdist_t$ be the set of all $\cD$-distinct relations of size~$t$. Let $\RIdist \coloneqq \bigcup_{i=0}^\infty \RIdist_t$ and $\RDdist \coloneqq \bigcup_{i=0}^\infty \RDdist_t$. Let $\RIdist_{\le t} \coloneqq \bigcup_{i=0}^t \RIdist_t$ and $\RDdist_{\le t} \coloneqq \bigcup_{i=0}^t \RDdist_t$.

\myparagraph{Variable-length registers}
For every integer $t \ge 0$, let $\reg{S}^{(t)}$ be a register associated with the Hilbert space
\[
\cH_\reg{S}^{(t)} \coloneqq \bigl( \mathbb{C}^N \otimes \mathbb{C}^N \bigr)^{\otimes t} \, .
\]
Let $\reg{S}$ be a register corresponding to the infinite-dimensional Hilbert space
\[
\cH_\reg{S} \coloneqq \bigoplus_{t=0}^\infty \cH_\reg{S}^{(t)}
= \bigoplus_{t=0}^\infty \bigl( \mathbb{C}^N \otimes \mathbb{C}^N \bigr)^{\otimes t}.
\]
When $t=0$, the space $\bigl( \mathbb{C}^N \otimes \mathbb{C}^N \bigr)^{\otimes 0} = \mathbb{C}$ is a one-dimensional Hilbert space. Thus, $\cH_\reg{S}^{(t)}$ is spanned by the states
\[
\ket{x_1, y_1, \dots, x_t, y_t} \qquad \text{where } x_i, y_i \in [N] .
\]
Note that the relation states $\ket{R}$ for $R \in \cR_t$ span the symmetric subspace of
$\cH_\reg{S}^{(t)}$.
We will sometimes divide up the register $\reg{S}^{(t)}$ as
\[
  \reg{S}^{(t)} := \bigl( \reg{S}^{(t)}_X, \reg{S}^{(t)}_Y \bigr),
\]
where $\reg{S}^{(t)}_X$ refers to the registers containing $\ket{x_1, \dots, x_t}$ and $\reg{S}^{(t)}_Y$ refers to the registers containing $\ket{y_1, \dots, y_t}$. We denote $\reg{S}^{(t)}_{X,i}$ as the register containing $\ket{x_i}$ and $\reg{S}^{(t)}_{Y,i}$ as the register containing $\ket{y_i}$.

Following our convention for defining the length/size of a relation $R$, we say that a state $\ket{x_1, y_1, \dots, x_t, y_t}$ has length/size $t$. Two states of different lengths are orthogonal by definition, since $\cH_\reg{S}$ is a direct sum $\bigoplus_{t=0}^\infty \cH_\reg{S}^{(t)}$.

\begin{notation}
For any $L \in \RIdist \cup \RDdist$, define the relation state
\[
\ket{L} := \frac{1}{\sqrt{t!}} \sum_{\pi \in \Symgp_t} \ket{x_{\pi^{-1}(1)}} \ket{y_{\pi^{-1}(1)}} \dots \ket{x_{\pi^{-1}(t)}} \ket{y_{\pi^{-1}(t)}},
\] 
where $t \coloneqq |L|$.\footnote{In~\cite{MH25}, relation states are defined for arbitrary relations, whereas we will not require them in this work.} For any integer $t \ge 0$, let $\Pi_{\le t}$ denote the projector\footnote{We note that our definition of $\Pi_{\le t}$ differs from that in~\cite{MH25}.}
\[
\bigoplus_{ \substack{
    L \in \RIdist, R \in \RDdist \colon \\
    |L|+|R| \le t
}   }
\proj{L}_\reg{S} \otimes \proj{R}_\reg{T} \, .
\]
\end{notation}

\begin{definition}[Path-Recording Oracle,~{\cite[Definitions~25 and~26]{MH25}}] \label{def:path_recording}
Define the following two operators (which are partial isometries). For any $x \in [N]$ and relations $L,R \in \rel$ such that $|L| + |R| < N$,
\begin{align*}
    & V^L \colon \ket{x}_{\reg{A}}\ket{L}_{\reg{S}}\ket{R}_{\reg{T}} \mapsto \frac{1}{\sqrt{N - |\Im(L\cup R)|}} \sum_{y \notin \Im(L\cup R)} \ket{y}_{\reg{A}} \ket{L \cup \{(x,y)\}}_{\reg{S}} \ket{R}_{\reg{T}}.
\end{align*}
For any $y \in [N]$ and relations $L,R \in \rel$ such that $|L| + |R| < N$,
\begin{align*}
    & V^R \colon \ket{y}_{\reg{A}}\ket{L}_{\reg{S}}\ket{R}_{\reg{T}} \mapsto \frac{1}{\sqrt{N - |\Dom(L\cup R)|}} \sum_{x \notin \Dom(L\cup R)} \ket{x}_{\reg{A}} \ket{L}_{\reg{S}} \ket{R \cup \{(x,y)\}}_{\reg{T}}.
\end{align*}
Define the following operator (which is a partial isometry).
\begin{align*}
    V \coloneqq V^L \cdot (\id - V^R \cdot V^{R,\dagger}) + (\id - V^L \cdot V^{L,\dagger}) \cdot V^{R,\dagger}.
\end{align*}
\end{definition}

\begin{theorem}[{\cite[Theorem~8]{MH25}}]
\label{thm:MH:path}
For any integer $0 \leq t < N$ and adversary $\Adversary = \left( A_1, \dots, A_{2t} \right)$, 
\[
\TD \qty( \Ex_{U \sim \haarunitaries_n}\ketbra{\Adversary_t^{U,U^{\dagger}}}{\Adversary_t^{U,U^{\dagger}}}, \Tr_{\reg{ST}} \left( \ketbra{\Adversary_t^{V,V^{\dagger}}}{\Adversary_t^{V,V^{\dagger}}}\right) ) 
\leq O(t^2/N^{1/8}),
\]
where 
\[
\ket{\Adversary_t^{U,U^{\dagger}}} \coloneqq \prod_{i=1}^{t} \left( U^{\dagger} \cdot A_{2i} \cdot U \cdot A_{2i-1} \right) \ket{0}_{\reg{A}} \ket{0}_{\reg{B}}, \quad \text{and}
\]
\[\ket{\Adversary_t^{V,V^{\dagger}}} \coloneqq \prod_{i=1}^{t} \left( V^{\dagger} \cdot A_{2i} \cdot V \cdot A_{2i-1} \right) \ket{0}_{\reg{A}} \ket{0}_{\reg{B}} \ket{\varnothing}_{\reg{S}} \ket{\varnothing}_{\reg{T}}.
\]
\end{theorem}

We will work with the following variants of path-recording oracles throughout this work.
\begin{definition}[Operators $F^L$, $F^R$, and $F$] \label{def:FL_FR_F}
For any $x \in [N]$, $L \in \RIdist$ and $R \in \RDdist$ such that $|L| + |R| < N$,
\begin{align}
F^L \colon \ket{x}_{\reg{A}} \ket{L}_{\reg{S}} \ket{R}_{\reg{T}} 
\mapsto \frac{1}{\sqrt{N}} \sum_{y \notin \Im(L)} \ket{y}_{\reg{A}} \ket{L \cup \set{(x,y)}}_{\reg{S}} \ket{R}_{\reg{T}}.
\label{eq:def:F_L}
\end{align}
For any $y \in [N]$, $L \in \RIdist$ and $R \in \RDdist$ such that $|L| + |R| < N$,
\begin{align}
F^R \colon \ket{y}_{\reg{A}} \ket{L}_{\reg{S}} \ket{R}_{\reg{T}} 
\mapsto \frac{1}{\sqrt{N}} \sum_{x \notin \Dom(R)} \ket{x}_{\reg{A}} \ket{L}_{\reg{S}} \ket{R \cup \set{(x,y)}}_{\reg{T}}.
\label{eq:def:F_R}
\end{align}
Define the operator
\begin{align}
F \coloneqq F^L \cdot (\id - F^R \cdot F^{R,\dagger}) + (\id - F^L \cdot F^{L,\dagger}) \cdot F^{R,\dagger}.
\label{eq:def:F}
\end{align}
\end{definition}

\noindent When $N = 2^\secp$ and $t = \poly(\secp)$, we show that $F$ is negligibly close to $V$ in operator norm. The formal statements and their proofs can be found in~\Cref{sec:FvsV}. Notice that $F^L$ and $F^R$ are \emph{not} partial isometries. In fact, they are contractions; that is, the operator norm of $F^L, F^R, F^{L,\dagger}, F^{R,\dagger}$ are all bounded by $1$. Nevertheless, they preserve orthogonality between the standard basis vectors of the domain. Formally, we have the following lemma.

\begin{fact}   \label{fact:image_FL}
For any distinct triples $(x,L,R) \neq (x',L',R')$, the states $F^L \ket{x}_{\reg{A}} \ket{L}_{\reg{S}} \ket{R}_{\reg{T}}$ and $F^L \ket{x'}_{\reg{A}} \ket{L'}_{\reg{S}} \ket{R'}_{\reg{T}}$ are orthogonal. Therefore, the set of subnormalized vectors
\[
\set{F^L \ket{x}_{\reg{A}} \ket{L}_{\reg{S}} \ket{R}_{\reg{T}}}_{(x,L,R)}
\]
form an orthogonal basis for the image of $F^L$ where $(x,L,R)$ ranges over $x \in [N], L \in \RIdist, R \in \RDdist$ such that $|L|+|R| < N$. Similar conditions hold for $F^R$ as well.
\end{fact}

\noindent Their adjoint operators $F^{L,\dagger}$ and $F^{R,\dagger}$ acts as follow: \\

\noindent For any $y \in [N], L \in \Rinj, R \in  \RDdist$,
\begin{equation}    \label{eq:FL_dagger}
F^{L,\dagger} \cdot \ket{y}_{\reg{A}} \ket{L}_{\reg{S}} \ket{R}_{\reg{T}} 
=
\begin{cases}
\frac{1}{\sqrt{N}} \ket{x}_{\reg{A}} \ket{L \setminus \set{(x,y)}}_{\reg{S}} \ket{R}_{\reg{T}} & \text{if } \exists x \in [N] \st (x,y) \in L \\
0 & \text{otherwise}.
\end{cases}    
\end{equation}
For any $x \in [N], L \in \Rinj, R \in  \RDdist$,
\begin{equation}      \label{eq:FR_dagger}
F^{R,\dagger} \cdot \ket{x}_{\reg{A}} \ket{L}_{\reg{S}} \ket{R}_{\reg{T}} 
=
\begin{cases}
\frac{1}{\sqrt{N}} \ket{y}_{\reg{A}} \ket{L}_{\reg{S}} \ket{R  \setminus \set{(x,y)}}_{\reg{T}} & \text{if } \exists y \in [N] \st (x,y) \in R, \\
0 & \text{otherwise}.
\end{cases}    
\end{equation}

\noindent Let $T$ be a partial isometry. It is well known that the operator $T^\dagger T$ is the orthogonal projection onto the domain of $T$. The domains of $V^L$ and $V^R$ are given by the span of all relation states. Although $F^L$ and $F^R$ are not partial isometries, they satisfy the following properties.

\begin{lemma}   \label{lem:domain_FL_FR}
For any integer $t \ge 0$,
\begin{align*}
    \|(F^{L,\dagger} F^L - \id) \Pi_{\le t}\| \le t/N
    \quad \text{and} \quad
    \|(F^{R,\dagger} F^R - \id) \Pi_{\le t}\| \le t/N.
\end{align*}
\end{lemma}
\noindent The proof can be found in~\Cref{sec:lemmas_F}. The following operators will be used extensively in~\Cref{sec:spru:1st_oracle,sec:spru:2nd_oracle}.

\begin{definition}[Operators $F^L_{\extract}$ and $F^R_{\extract}$]
Define the partial isometry $F_L^{\extract}$ such that for any $L \in \RIdist$ and $y \notin \Im(L)$,
\begin{equation} \label{eq:FL_extract}
F^L_{\extract} \colon \ket{y}_{\reg{A}} \ket{L \cup \set{(x,y)}}_{\reg{S}}
\mapsto \ket{y}_{\reg{A}'} \ket{x}_{\reg{A}} \ket{L}_{\reg{S}},
\end{equation}
where register $\reg{A}'$ labels a Hilbert space with the same dimension as $\reg{A}$. Similarly, define the partial isometry $F_R^{\extract}$ such that for any $R \in \RDdist$ and $x \notin \Dom(R)$,
\begin{equation} \label{eq:FR_extract}
F^R_{\extract} \colon \ket{x}_{\reg{A}} \ket{R \cup \set{(x,y)}}_{\reg{T}} 
\mapsto \ket{x}_{\reg{A}'} \ket{y}_{\reg{A}} \ket{R}_{\reg{T}}.
\end{equation}
\end{definition}
\noindent We will use the following lemma in~\Cref{sec:spru} to bound error terms. It can be viewed as a consequence of the ``monogamy of entanglement''. Intuitively, after applying $F^L$ to a state, the registers $\reg{A}$ and $\reg{S_Y}$ become ``maximally entangled'' (see~\Cref{fact:operator_E}). The monogamy of entanglement then implies that $\reg{A}$ and $\reg{S}$ must be nearly disentangled from $\reg{T}$—that is, they lie almost entirely outside the image of $F^R$, and vice versa.
\begin{lemma}   \label{lem:FLdagger_U_FR:zero}
For any integer $t \ge 0$ and any unitary $U$ acting non-trivially on the register $\reg{A}$,
\begin{align*}
\|F^{L,\dagger} U F^R \Pi_{\leq t}\|_{\op} \leq 3\sqrt{t(t+2)/N}
\quad \text{and} \quad
\|F^{R,\dagger} U F^L \Pi_{\leq t}\|_{\op} \leq 3\sqrt{t(t+2)/N} \, .
\end{align*}
\end{lemma}
\noindent The proof of~\Cref{lem:FLdagger_U_FR:zero} can be found in~\Cref{sec:lemmas_F}.

\section{Strong Pseudorandom Unitaries in the QHROM}
\label{sec:spru}

\subsection{Construction}
\label{sec:spru:construction}

\begin{construction}\label{construction:SPRU}
For every $\secp \in \N$, let $\cK_\secp \coloneqq \bit^{3\secp}$ and $\cG_\secp = \set{G^\cU_k}_{k \in \cK_\secp}$ denote the family of unitaries with access to the Haar random oracle $\cU = \set{U_\ell}_{\ell \in \N}$, defined as follows. For every $\secp \in \N$ and $k \in \cK_\secp$, define the $\secp$-qubit unitary:
\begin{align*}
    G^{\cU}_k \coloneqq X^{k_3} \cdot U_\secp \cdot X^{k_2} \cdot U_\secp \cdot X^{k_1},
\end{align*}
where $k = k_1 \,\|\, k_2 \,\|\, k_3$ with $k_1,k_2,k_3 \in \{0,1\}^{\secp}$, and for a bitstring $s = s_1 \cdots s_\secp$ we set $X^{s} \coloneqq \bigotimes_{i=1}^{\secp} X^{s_i}$ (so $X^0 = \id$, $X^1 = X$).
\end{construction}

\begin{remark}
We observe that~\Cref{construction:SPRU} does not require any ancilla qubits. Moreover, it is optimal in terms of the number of sequential queries to the Haar random oracle. In particular, \cite{ABGL24} constructs a polynomial-query attack that breaks every non-adaptive PRU constructions in the inverseless QHROM, namely constructions that are allowed to make a single parallel query to the Haar random oracle $U$ of the form $U^{\otimes q}$ for an arbitrary polynomial $q(\secp)$. We observe that the same attack also applies in the QHROM.
\end{remark}

\begin{theorem}\label{thm:strong:pru:qhrom}
The family of unitaries defined in~\Cref{construction:SPRU} is a strong pseudorandom unitary in the QHROM. 
\end{theorem}
We will prove~\Cref{thm:strong:pru:qhrom} in~\Cref{sec:strong:proof:pru:qhrom}. Looking ahead, inspecting the proof shows that if we shorten the key blocks to $n(\secp)=\omega(\log \secp)$—that is, $k_1,k_2,k_3 \in \{0,1\}^{n(\secp)}$—and restrict $X$ to act on only $n(\secp)$ qubits, the modified construction remains a strong PRU in the QHROM. This yields the following corollary.

\begin{corollary}\label{cor:stretch_pru}
    Let $n(\secp) = \omega(\log\secp), \cK_\secp := \bit^{3n(\secp)}$ for every $\secp \in \N$ and $\cF_\secp = \set{F^\cU_k}_{k \in \cK_\secp}$ denote the family of unitaries with access to the Haar random oracle $\cU = \set{U_\ell}_{\ell \in \N}$, defined as follows. For every $\secp \in \N$ and $k \in \cK_\secp$, define the $\secp$-qubit unitary:
\begin{align*}
    F^{\cU}_k \coloneqq (X^{k_3} \otimes \id_{\secp - n}) \cdot U_\secp \cdot (X^{k_2} \otimes \id_{\secp - n}) \cdot U_\secp \cdot (X^{k_1} \otimes \id_{\secp - n}),
\end{align*}
where $k = k_1||k_2||k_3$ with $k_1, k_2, k_3 \in \bit^{n(\secp)}$. Then $\set{\cF_\secp}_{\secp \in \N}$ is a strong PRU in the QHROM.
\end{corollary}

\subsection{Security Proof: Proving~\Cref{thm:strong:pru:qhrom}}
\label{sec:strong:proof:pru:qhrom}

\noindent Fix $\secp$ and let $N = 2^\secp$. Consider an adversary $\Adversary = (A_1,A_2, \dots, A_{4t})$ in the strong PRU security experiment. We define the following hybrid states on registers $(\reg{A},\reg{B})$. Although $\Adversary$ has access to the Haar random oracle of all lengths, \Cref{construction:SPRU} make queries only to $U_{\secp}$, which is independent of the oracles of other lengths. We may, without loss of generality, assume that $\Adversary$ queries only the Haar random oracle on $\secp$ qubits.

\paragraph{Hybrid $\hyb_1$:} This is the ideal experiment. Namely, the adversary is interacting with two independent Haar random unitaries $(U_1,U_2)$. The final state of the adversary is the following:
\[
\rho_1 \coloneqq \Ex_{U_1,U_2 \sim \haarunitaries_\secp} \left[\ketbra{\Adversary_t^{U_1, U_2, U_1^{\dagger}, U_2^{\dagger}}}{\Adversary_t^{U_1, U_2, U_1^{\dagger}, U_2^{\dagger}}}\right].
\]

\paragraph{Hybrid $\hyb_2$:} Same as Hybrid $\hyb_1$ except Haar unitaries $(U_1, U_2)$ are simulated by path-recording oracles $(V_1, V_2)$ defined in~\Cref{def:path_recording}. Define the following state: 
\[
\ket{\hyb_2}_{\reg{ABS_1T_1S_2T_2}} 
\coloneqq \prod_{i=1}^{t} \left(V_2^{\dagger} \cdot A_{4i} \cdot V_1^{\dagger} \cdot A_{4i-1} \cdot V_2 \cdot A_{4i-2} \cdot V_1 \cdot A_{4i-3}\right) 
\ket{0}_{\reg{A}} \ket{0}_{\reg{B}} 
\ket{\varnothing}_{\reg{S_1}} \ket{\varnothing}_{\reg{T_1}} 
\ket{\varnothing}_{\reg{S_2}} \ket{\varnothing}_{\reg{T_2}},
\]
where $V_1$ acts on the registers $\reg{A}, \reg{S}_1, \reg{T}_1$, and $V_2$ acts on the registers $\reg{A}, \reg{S}_2, \reg{T}_2$. Define 
\[
\rho_2 \coloneqq \Tr_{\reg{S_1T_1S_2T_2}}(\ketbra{\hyb_2}{\hyb_2}).
\]
\paragraph{Hybrid $\hyb_3$:} Same as Hybrid $\hyb_2$ except $(V_1, V_2)$ are replaced by $(F_1, F_2)$ defined in~\Cref{def:FL_FR_F}. Define the following state:
\[
\ket{\hyb_3}_{\reg{ABS_1T_1S_2T_2}} 
\coloneqq \prod_{i=1}^{t} \left( F_2^{\dagger} \cdot A_{4i} \cdot F_1^{\dagger} \cdot A_{4i-1} \cdot F_2 \cdot A_{4i-2} \cdot F_1 \cdot A_{4i-3} \right)
\ket{0}_{\reg{A}}\ket{0}_{\reg{B}}
\ket{\varnothing}_{\reg{S_1}}\ket{\varnothing}_{\reg{T_1}}
\ket{\varnothing}_{\reg{S_2}}\ket{\varnothing}_{\reg{T_2}},
\]
where $F_1$ acts on the registers $\reg{A}, \reg{S}_1, \reg{T}_1$, and $F_2$ acts on the registers $\reg{A}, \reg{S}_2, \reg{T}_2$. Define
\[
\rho_3 \coloneqq \Tr_{\reg{S_1T_1S_2T_2}}(\ketbra{\hyb_3}{\hyb_3}).
\]
\paragraph{Hybrid $\hyb_4$:} Define the following state:
\begin{multline*}
\ket{\hyb_4}_{\reg{ABSTK_1K_2K_3}} \coloneqq \\
\frac{1}{\sqrt{N^3}} \sum_{k_1,k_2,k_3 \in \bit^\secp} 
\prod_{i=1}^{t} \left(F^{\dagger} X^{k_2} F^{\dagger} \cdot A_{4i} \cdot X^{k_1} F^{\dagger} X^{k_3} \cdot A_{4i-1} \cdot F X^{k_2} F \cdot A_{4i-2} \cdot X^{k_3} F X^{k_1} \cdot A_{4i-3} \right) \\
\cdot \ket{0}_{\reg{A}} \ket{0}_{\reg{B}} \ket{\varnothing}_{\reg{S}} \ket{\varnothing}_{\reg{T}} 
\ket{k_1}_{\reg{K_1}} \ket{k_2}_{\reg{K_2}} \ket{k_3}_{\reg{K_3}},
\end{multline*}
where $F$ acts on the registers $\reg{A}, \reg{S}, \reg{T}$, and $X^{k_j}$ acting on register $\reg{A}$ for $j = 1, 2, 3$. Define
\[
\rho_4 \coloneqq \Tr_{\reg{STK_1K_2K_3}}(\ketbra{\hyb_4}{\hyb_4}).
\]
\paragraph{Hybrid $\hyb_5$:} Same as Hybrid $\hyb_4$ except $F$ is replaced by $V$. Define the following state:
\begin{multline*}
\ket{\hyb_5}_{\reg{ABSTK_1K_2K_3}} \coloneqq \\
\frac{1}{\sqrt{N^3}} \sum_{k_1,k_2,k_3 \in \bit^\secp} 
\prod_{i=1}^{t} \left(V^{\dagger} X^{k_2} V^{\dagger} \cdot A_{4i} \cdot X^{k_1} V^{\dagger} X^{k_3} \cdot A_{4i-1} \cdot V X^{k_2} V \cdot A_{4i-2} \cdot X^{k_3} V X^{k_1} \cdot A_{4i-3} \right) \\
\cdot \ket{0}_{\reg{A}} \ket{0}_{\reg{B}} \ket{\varnothing}_{\reg{S}} \ket{\varnothing}_{\reg{T}} 
\ket{k_1}_{\reg{K_1}} \ket{k_2}_{\reg{K_2}} \ket{k_3}_{\reg{K_3}},
\end{multline*}
where $V$ acts on the registers $\reg{A}, \reg{S}, \reg{T}$, and $X^{k_j}$ acting on register $\reg{A}$ for $j = 1, 2, 3$. Define
\[
\rho_5 \coloneqq \Tr_{\reg{STK_1K_2K_3}}(\ketbra{\hyb_5}{\hyb_5}).
\]
\paragraph{Hybrid $\hyb_6$:} Same as Hybrid $\hyb_5$ except $V$ is replaced by a $\secp$-qubit Haar random unitary $U$, and no purifications are introduced. Define
\[
\rho_6 \coloneqq \Ex_{\substack{k_1,k_2,k_3 \gets \bit^\secp,\, U \sim \haarunitaries_\secp \\ \cO_1 \equiv X^{k_3} U X^{k_1},\, \cO_2 \equiv U X^{k_2} U}}
\left[ \ketbra{\Adversary_t^{\cO_1, \cO_2, \cO_1^\dagger, \cO_2^\dagger}}{\Adversary_t^{\cO_1, \cO_2, \cO_1^\dagger, \cO_2^\dagger}} \right].
\]

\paragraph{Hybrid $\hyb_7$:} This is the real experiment. Namely, the adversary is interacting with the Haar random oracle and the strong PRU construction $G^{\cU}_k$ defined in~\Cref{construction:SPRU}. The final state of the adversary is
\[
\rho_7 \coloneqq \Ex_{\substack{k_1,k_2,k_3 \gets \bit^\secp,\, U \sim \haarunitaries_\secp \\ \cO_1 \equiv U,\, \cO_2 \equiv X^{k_3} U X^{k_2} U X^{k_1}}}
\left[ \ketbra{\Adversary_t^{\cO_1, \cO_2, \cO_1^\dagger, \cO_2^\dagger}}{\Adversary_t^{\cO_1, \cO_2, \cO_1^\dagger, \cO_2^\dagger}} \right].
\]

\paragraph{Statistical Indistinguishability of Hybrids.}

\noindent We prove the closeness as follows:

\begin{myclaim} \label{claim:H1H2}
$\TD(\rho_1,\rho_2) = O\left(\frac{t^2}{N^{1/8}}\right)$ and $\TD(\rho_5,\rho_6) = O\left(\frac{t^2}{N^{1/8}}\right)$.
\end{myclaim}
\begin{proof}
It immediately follows from~\Cref{thm:MH:path}.
\end{proof}

\begin{myclaim} \label{claim:H2H3}
$\TD(\rho_2,\rho_3) = O\left(\frac{t^2}{N^{1/2}}\right)$ and $\TD(\rho_4,\rho_5) = O\left(\frac{t^2}{N^{1/2}}\right)$.
\end{myclaim}
\begin{proof}
It immediately follows from~\Cref{lem:path:FV}.
\end{proof}

\begin{myclaim} \label{claim:H6H7}
$\rho_6 = \rho_7$. 
\end{myclaim}
\begin{proof}
We prove a stronger result by showing that the oracles in both hybrids are identically distributed. For any fixed choice of $k_1, k_2, k_3 \in [N]$, $(U, X^{k_3} U X^{k_2} U X^{k_1})$ is identically distributed to $(X^{k_3} U X^{k_1}, X^{k_3} \cdot X^{k_3} U X^{k_1} \cdot X^{k_2} \cdot X^{k_3} U X^{k_1} \cdot X^{k_1}) = (X^{k_3} U X^{k_1}, U X^{k_1 \oplus k_2 \oplus k_3} U)$ by the unitary invariance of the Haar measure. Next, after averaging over $k_2$, $(X^{k_3} U X^{k_1}, U X^{k_1 \oplus k_2 \oplus k_3} U)$ is identically distributed to $(X^{k_3} U X^{k_1}, U X^{k_2} U)$ since $k_2$ is uniformly random and independent of $U$, $k_1$ and $k_3$. Finally, averaging over $k_1$ and $k_3$ completes the proof.
\end{proof}

\begin{lemma} \label{lem:spru:induction}
$\TD(\rho_3, \rho_4) = O\left(\frac{t^2}{N^{1/2}}\right)$. 
\end{lemma}

\noindent Proving~\Cref{lem:spru:induction} is the main technical step of proving~\Cref{thm:strong:pru:qhrom}. Toward the proof, we begin by defining an approximate isometry $\Ospru$ in~\Cref{sec:spru:def}, which we then use to prove~\Cref{lem:spru:induction} in~\Cref{sec:spru:induction}.

\begin{proof}[Proof of~\Cref{thm:strong:pru:qhrom}]
It immediately follows from a standard hybrid argument, \Cref{claim:H1H2,claim:H2H3,claim:H6H7} and~\Cref{lem:spru:induction}.
\end{proof}

\subsection{Auxiliary Definitions}
\label{sec:spru:def}

\subsubsection{Overview} Our goal is to show the closeness between the adversary's final states in hybrids $\hyb_3$ and $\hyb_4$. We start by noting that in hybrid $\hyb_3$, the purification register contains two pairs of ``databases'' on registers $(\reg{S_1}, \reg{T_1})$ and $(\reg{S_2}, \reg{T_2})$, whereas in hybrid $\hyb_4$, the purification register contains a single pair of databases on registers $(\reg{S},\reg{T})$ along with the key registers $(\reg{K}_1, \reg{K}_2, \reg{K}_3)$. Indeed, two quantum states are equal if their purifications are related by an isometry acting solely on the traced-out registers. Thus, it suffices to find an (approximate) isometry $\Ospru$ that maps the purification registers $(\reg{S}_1, \reg{T}_1, \reg{S}_2, \reg{T}_2)$ in hybrid $\hyb_3$ to the purification registers $(\reg{S}, \reg{T}, \reg{K}_1, \reg{K}_2, \reg{K}_3)$ in hybrid $\hyb_4$.

Before defining $\Ospru$, we present a ``classical'' analogue which, while not exact, serves to motivate the upcoming definitions. Suppose $\Adversary$ in hybrid $\hyb_3$ makes one query $x$ to $F_1$. Informally, the action of $F_1$ can be viewed first ``sampling'' a uniform $y \notin \Im(L_1)$, then ``adding'' the pair $(x,y)$ to $L_1$, and finally returning $y$ back to $\Adversary$, all in superposition. On the other hand, suppose $\Adversary$ in hybrid $\hyb_4$ makes one query $x$ to $X^{k_3} F X^{k_1}$. Similarly, $(x \oplus k_1, y)$ is added to $L$, and $y \oplus k_3$ is returned to $\Adversary$. We can relabel $y  \oplus k_3 \mapsto y$ to have that $(x \oplus k_1, y \oplus k_3)$ is added to $L$, and $y$ is returned to $\Adversary$. Now, suppose $\Adversary$ in hybrid $\hyb_3$ makes $q$ queries to $F_1$, so that $L_1$ becomes $\set{(x_1,y_1),\dots,(x_q,y_q)}$. By inspection, the corresponding $L$ is $\set{(x_1 \oplus k_1, y_1 \oplus k_3),\dots,(x_q \oplus k_1, y_q \oplus k_3)}$, which is identical to that of $L_1$ except that each element in the domain is XOR-ed with $k_1$, and each element in the range is XOR-ed with $k_3$. We refer to it as the \emph{augmented relation of $L_1$} parameterized by $(k_1, k_3)$. We denote it by $L_1^{\ell,(k_1,k_3)}$.

Next, consider a query $x$ to $F_2$ in hybrid $\hyb_3$. In this case, a uniform $y \notin \Im(L_2)$ is sampled, the pair $(x,y)$ is added to $L_2$, and $y$ it returned to $\Adversary$, all in superposition. By contrast, in hybrid $\hyb_4$, a query $x$ to $F X^{k_2} F$ proceeds as follows: $F$ samples a uniform $z \notin \Im(L)$, adds the pair $(x,z)$ to $L$, XORs $z$ by $k_2$, samples a uniform $y \notin \Im(L)$, adds the pair $(z \oplus k_2, y)$ to $L$, and finally returns to $\Adversary$. Similarly, suppose $\Adversary$ in hybrid $\hyb_3$ makes $q$ queries to $F_2$, so that $L_2$ becomes $\set{(x_1,y_1),\dots,(x_q,y_q)}$. The corresponding $L$ is $\set{(x_1, z_1), (z_1 \oplus k_2, y_1), \dots, (x_q, z_q), (z_q \oplus k_2, y_q)}$, where $\vec{z} = (z_1, \dots,z_q)$ is a vector of distinct coordinates. An important observation, also noted in~\cite{ABGL24}, is that the elements are ``paired'' by $k_2$ and are referred to as \emph{$k_2$-correlated pairs}. We refer to it as the \emph{augmented relation of $L_2$} parameterized by $(\vec{z}, k_2)$. We denote it by $L_2^{\ell,(\vec{z},k_2)}$.

Suppose $\Adversary$ in hybrid $\hyb_3$ makes $q$ queries to $F_1$ and $F_2$ respectively, leading to $L_1 = \set{(x_1,y_1), \ldots, \allowbreak (x_q,y_q)}$ and $L_2 = \set{(x'_1,y'_1),\ldots,(x'_q,y'_q)}$. Now, in order to map $(L_1,L_2)$ to $(L,k_1,k_2,k_3)$\footnote{Since $R_1$ and $R_2$ are empty, we ignore them in the exposition here.} almost injectively, we approach is as follows. We first sample $(\vec{z},k_1,k_2,k_3)$ and let $L = L_1^{\ell,(k_1,k_3)} \cup L_2^{\ell,(\vec{z},k_2)}$ such that one can uniquely recover $(L_1,L_2)$ given the information of $(L,k_1,k_2,k_3)$. Following the idea in~\cite{ABGL24}, as long as the number of $k_2$-correlated pairs in $L$ is exactly $q$, we know that $L_2^{\ell,(\vec{z},k_2)}$ consists of the $k_2$-correlated pairs in $L$. Suppose $q$ is polynomial. An elementary combinatorial argument shows that the fraction of ``bad'' keys for which unique recoverability fails is negligible. 

Careful readers may notice that the previous argument does not rely on keys $(k_1,k_3)$. Indeed, the PRU construction in the inverseless QHROM of~\cite{ABGL24} does not use $(k_1,k_3)$; on key $k$, their construction is simply $U X^k U$. However, when we move to the strong PRU setting, this construction is insecure--one can learn the key by sequentially querying (1) $U^\dagger$ (2) $U X^k U$ and (3) $U^\dagger$. More generally, queries to $U^\dagger$ may lead to unintended cancellation in the databases. The role of $(k_1,k_3)$ is precisely to prevent such event from happening. Looking ahead, when defining $\Ospru$, the condition on $(k_1,k_3)$ ensures the image and domain of $L_1^{\ell,(k_1,k_3)}$ and $L_2^{\ell,(\vec{z},k_2)}$ are distinct.

\subsubsection{Graph-theoretic definitions}

In this subsection, we let $N = 2^\secp$, and use $[N]$ and $\bit^\secp$ interchangeably. For the purposes of the proofs in later sections, it is convenient—and also intuitive—to introduce graph-theoretic languages. In particular, the following type of \emph{directed graph} will play an important role.

\begin{definition}[Decomposable Graphs] \label{def:decomp_graphs}
A directed graph $G$, possibly with self-loops, is \emph{decomposable} if it contains no self-loops and no two edges share a common vertex. In addition, we can uniquely partition its vertex set into three disjoint subsets: 
\begin{itemize}
    \item $V_{\isolate}(G)$: the set of vertices with no incoming or outgoing edges
    \item $V_{\source}(G)$: the set of vertices with an outgoing edge
    \item $V_{\target}(G)$: the set of vertices with an incoming edge
\end{itemize}
Moreover, we have $|V_{\source}(G)| = |V_{\target}(G)|$. 
\end{definition}

\noindent The following definition connects relation–key pairs to directed graphs.
\begin{definition}[Relation-Key-Induced Graphs] \label{def:rel_key_induced_graph}
For any relation $L$ and $k \in [N]$, define the directed graphs, possibly with self-loops, $G^{\ell}_{L,k}$ as follows. The vertex set of $G^{\ell}_{L,k}$ is equal to $L$.\footnote{Note that $L$ may contain repeated elements. Equivalently, one can regard the vertex set as having size $|L|$, with each vertex labeled by an element of $L$.} For any two vertices $(x, y)$ and $(x', y')$, there is a directed edge from $(x, y)$ to $(x', y')$ if and only if $x' = y \oplus k$.

Similarly, for any relation $R$ and $k \in [N]$, define the directed graphs, possibly with self-loops, $G^{r}_{R,k}$ as follows. The vertex set of $G^{r}_{R,k}$ is equal to $R$. For any two vertices $(x, y)$ and $(x', y')$, there is a directed edge from $(x, y)$ to $(x', y')$ if and only if $y' = x \oplus k$.\footnote{Notice that this is the opposite of defining edges in $G^{\ell}_{L,k}$.}
\end{definition}

\subsection{Augmented Relations and (Robust) Decodability}

We define augmented relations mentioned in the overview below.

\begin{definition}[Augmented Relations] \label{def:aug_relations}
For any $L_1, L_2 \in \Rinj$, $R_1, R_2 \in \RDdist$, $k_1, k_2, k_3 \in [N]$, $\vec{z}_L \in [N]^{|L_2|}_\dist$, and $\vec{z}_R \in [N]^{|R_2|}_\dist$, define the corresponding \emph{augmented relations}:\footnote{Similar to how we define relations, augmented relations are also defined as multisets.}
\begin{align}
L_1^{\ell,(k_1,k_3)} & \coloneqq \set{ (x \oplus k_1, y \oplus k_3): {(x, y) \in L_1} },
\label{eq:aug_L1} \\
L_2^{\ell,(k_2,\vec{z}_L)} & \coloneqq \set{ (x_i, z_{L,i}), (z_{L,i} \oplus k_2,y_i): {(x_i, y_i) \in L_2} }, 
\label{eq:aug_L2} \\
R_1^{r,(k_1,k_3)} & \coloneqq \set{(x \oplus k_1, y \oplus k_3): {(x, y) \in R_1} }, 
\label{eq:aug_R1} \\
R_2^{r,(k_2,\vec{z}_R)} & \coloneqq \set{ (x_i, z_{R,i} \oplus k_2),(z_{R,i},y_i): {(x_i, y_i) \in R_2} },
\label{eq:aug_R2}
\end{align}
where we impose an ordering on the elements of $L_2$ (\resp $R_2$) by the canonical ordering on $\Im(L_2)$ (\resp $\Dom(R_2)$). That is, we use $(x_i, y_i) \in L_2$ to indicate the element in $L_2$ such that $y_i$ is the $i$-th largest element in $\Im(L_2)$, and $(x_i, y_i) \in R_2$ indicate the element in $R_2$ such that $x_i$ is the $i$-th largest element in $\Dom(R_2)$.
\end{definition}

\begin{remark}
As expressions like $L_1^{r,(k_1,k_3)}$ or $L_1^{\ell,(k_2,\vec{z}_L)}$ never appear in this work, we omit the superscript $\ell,r$ in the augmented relations when it is clear from the context. 
\end{remark}

\noindent Now, given any $L_1, L_2 \in \Rinj$, $R_1, R_2 \in \RDdist$, which can be viewed as a view in hybrid~$\hyb_3$. We aim to map it to a view in hybrid~$\hyb_4$ by:
\begin{enumerate}
    \item sampling $(k_1,k_2,k_3,\vec{z}_L,\vec{z}_R)$ from an appropriate distribution;
    \item outputting $(L_1^{(k_1,k_3)} \cup L_2^{(k_2,\vec{z}_L)}, R_1^{(k_1,k_3)} \cup R_2^{(k_2,\vec{z}_R)}, k_1, k_2, k_3)$.
\end{enumerate}
To define an (approximate) isometry $\Ospru$ that performs this map coherently, we must ensure that the mapping is (almost) reversible. To this end, we proceed in two steps: 
\begin{enumerate}
    \item Define a ``decoder'' which, on input $(L,R,k_1,k_2,k_3)$, outputs $(L_1,R_1,L_2,R_2,\vec{z}_L,\vec{z}_R,k_1,k_2,k_3)$.
    \item For any $L_1, L_2 \in \Rinj$, $R_1, R_2 \in \RDdist$, identify conditions on $(k_1,k_2,k_3,\vec{z}_L,\vec{z}_R)$ such that applying the decoder to $(L_1^{(k_1,k_3)} \cup L_2^{(k_2,\vec{z}_L)}, R_1^{(k_1,k_3)} \cup R_2^{(k_2,\vec{z}_R)}, k_1, k_2, k_3)$ yields the correct $(L_1,R_1,L_2,R_2,\vec{z}_L,\vec{z}_R,$ $k_1,k_2,k_3)$.
\end{enumerate}

\begin{definition}[Function $\Dec$ and Operator $\cD$] \label{def:Ospru_split}
The deterministic function (algorithm) $\Dec$ is defined as follows:  \\
{\bf Input:} Two relations $(L,R)$ and $(k_1, k_2, k_3) \in [N]^3$
\begin{enumerate}
    \item If $G^{\ell}_{L,k_2}$ is not decomposable {\bf or} $G^{r}_{R,k_2}$ is not decomposable, then output $\bot$ indicating fail.
    \item If $V_{\target}(G^{\ell}_{L,k_2}) \notin \RIdist$ {\bf or} $V_{\target}(G^{r}_{R,k_2}) \notin \RDdist$, then output $\bot$.
    \item Suppose
    \begin{itemize}
        \item $V_{\isolate}(G^{\ell}_{L,k_2}) = \set{(a_1, b_1), \dots, (a_s, b_s)}$
        \item $V_{\source}(G^{\ell}_{L,k_2}) = \set{(x_1, e_1), \dots, (x_\ell, e_\ell)}$
        \item $V_{\target}(G^{\ell}_{L,k_2}) = \set{(e_1 \oplus k_2, y_1), \dots, (e_\ell \oplus k_2, y_\ell)}$ such that $y_1 < y_2 < \dots < y_\ell$\footnote{Since it passes Step~2, there exists an ordering in the image of $V_{\target}(G^{\ell}_{L,k_2})$.}
        \item $V_{\isolate}(G^{r}_{R,k_2}) = \set{(c_1, d_1), \dots, (c_t, d_t)}$
        \item $V_{\source}(G^{r}_{R,k_2}) = \set{(f_1, v_1), \dots, (f_r, v_r)}$
        \item $V_{\target}(G^{r}_{R,k_2}) = \set{(u_1, f_1 \oplus k_2), \dots, (u_r, f_r \oplus k_2)}$ such that $u_1 < u_2 < \dots < u_r$
    \end{itemize}
    \item Let
    \begin{itemize}
        \item $L_{\isolate} \coloneqq \set{(a_1 \oplus k_1, b_1 \oplus k_3), \ldots, (a_s \oplus k_1, b_s \oplus k_3)}$
        \item $L_{\pair} \coloneqq \set{(x_1, y_1), \ldots, (x_\ell, y_\ell)}$
        \item $\vec{m}_L \coloneqq (e_1, \ldots, e_\ell)$
        \item $R_{\isolate} \coloneqq \set{(c_1 \oplus k_1, d_1 \oplus k_3), \ldots, (c_t \oplus k_1, d_t \oplus k_3)}$
        \item $R_{\pair} \coloneqq \set{(u_1, v_1), \ldots, (u_r, v_r)}$
        \item $\vec{m}_R \coloneqq (f_1, \ldots, f_\ell)$
    \end{itemize}
    \item If $L_{\isolate} \notin \RIdist$ {\bf or} $L_{\pair} \notin \RIdist$ {\bf or} $R_{\isolate} \notin \RDdist$ {\bf or} $R_{\pair} \notin \RDdist$ {\bf or} $\vec{m}_L$ has repeated coordinates {\bf or} $\vec{m}_R$ has repeated coordinates, then output $\bot$.
    \item If $\Im(L_{\pair}) \cap \set{\vec{m}_L} \neq \varnothing$ {\bf or} $\Dom(R_{\pair}) \cap \set{\vec{m}_R} \neq \varnothing$, then output $\bot$.\footnote{Recall that for a vector $\vec{v} = (v_1, v_2, \ldots)$, we denote the set $\bigcup_i \set{v_i}$ by $\set{\vec{v}}$.}
    \item Output $(L_{\isolate},R_{\isolate},L_{\pair},R_{\pair},\vec{m}_L,\vec{m}_R,k_1, k_2, k_3)$
\end{enumerate}
We denote by $\Osprus$ the operator that performs $\Dec$ coherently,\footnote{Since $(k_1,k_2,k_3)$ is classical information, we require $\Dec$ to also output $(k_1,k_2,k_3)$.} and maps an input to the zero vector if $\Dec$ evaluates to $\bot$ on that input.
\end{definition}

\begin{lemma}[$\Osprus$ is a partial isometry]  \label{lem:D_is_partial_isometry}
Let $\Supp(\Dec)$ denote the set of $(L,R,k_1,k_2,k_3)$ such that $\Dec(L,R,$ $k_1,k_2,k_3) \neq \bot$. Then $\Dec$ is injective when restricted to $\Supp(\Dec)$. As a corollary, $\Osprus$ is a partial isometry.
\end{lemma}

\begin{proof}
We prove the lemma by constructing the inverse of $\Dec$ on $\Supp(\Dec)$. We denote this inverse by $\Enc$. On input $(L_{\isolate}, R_{\isolate}, L_{\pair}, R_{\pair}, \vec{m}_L, \vec{m}_R, k_1, k_2, k_3)$, $\Enc$ outputs 
\[
(L_{\isolate}^{(k_1,k_3)} \cup L_{\pair}^{(k_2,\vec{m}_L)}, R_{\isolate}^{(k_1,k_3)} \cup R_{\pair}^{(k_2,\vec{m}_R)}, k_1, k_2, k_3).
\]
It suffices to show that for any $(L,R,k_1,k_2,k_3) \in \Supp(\Dec)$, the following holds:
\begin{align*}
    (L,R,k_1,k_2,k_3) = \Enc(\Dec((L,R,k_1,k_2,k_3))).
\end{align*}
Suppose in Step~3 of $\Dec(L,R,k_1,k_2,k_3)$, we have
\begin{itemize}
    \item $V_{\isolate}(G^{\ell}_{L,k_2}) = \set{(a_1, b_1), \dots, (a_s, b_s)}$
    \item $V_{\source}(G^{\ell}_{L,k_2}) = \set{(x_1, e_1), \dots, (x_\ell, e_\ell)}$
    \item $V_{\target}(G^{\ell}_{L,k_2}) = \set{(e_1 \oplus k_2, y_1), \dots, (e_\ell \oplus k_2, y_\ell)}$ such that $y_1 < y_2 < \dots < y_\ell$
    \item $V_{\isolate}(G^{r}_{R,k_2}) = \set{(c_1, d_1), \dots, (c_t, d_t)}$
    \item $V_{\source}(G^{r}_{R,k_2}) = \set{(f_1, v_1), \dots, (f_r, v_r)}$
    \item $V_{\target}(G^{r}_{R,k_2}) = \set{(u_1, f_1 \oplus k_2), \dots, (u_r, f_r \oplus k_2)}$ such that $u_1 < u_2 < \dots < u_r$,
\end{itemize}
and the input satisfies 
\begin{align*}
& L = \set{(a_1, b_1), \dots, (a_s, b_s)} \cup \set{(x_1, e_1), \dots, (x_\ell, e_\ell)} \cup \set{(e_1 \oplus k_2, y_1), \dots, (e_\ell \oplus k_2, y_\ell)}, \\
& R = \set{(c_1, d_1), \dots, (c_t, d_t)} \cup \set{(f_1, v_1), \dots, (f_r, v_r)} \cup \set{(u_1, f_1 \oplus k_2), \dots, (u_r, f_r \oplus k_2)}.
\end{align*}
In Step~4 of $\Dec((L,R,k_1,k_2,k_3))$, we have
\begin{itemize}
    \item $L_{\isolate} \coloneqq \set{(a_1 \oplus k_1, b_1 \oplus k_3), \ldots, (a_s \oplus k_1, b_s \oplus k_3)}$
    \item $L_{\pair} \coloneqq \set{(x_1, y_1), \ldots, (x_\ell, y_\ell)}$
    \item $\vec{m}_L \coloneqq (e_1, \ldots, e_\ell)$
    \item $R_{\isolate} \coloneqq \set{(c_1 \oplus k_1, d_1 \oplus k_3), \ldots, (c_t \oplus k_1, d_t \oplus k_3)}$
    \item $R_{\pair} \coloneqq \set{(u_1, v_1), \ldots, (u_r, v_r)}$
    \item $\vec{m}_R \coloneqq (f_1, \ldots, f_\ell)$.
\end{itemize}
Recall~\Cref{def:aug_relations}. It is straightforward to verify that applying $\Enc$ to the above input yields $L$ and $R$.
\end{proof}

\begin{definition}[Decodability]  \label{def:decodable_tuples} 
For any $L_1, L_2 \in \Rinj$, $R_1, R_2 \in \RDdist$, a tuple $(k_1, k_2, k_3, \vec{z}_L, \vec{z}_R)$ said to be \emph{decodable} (\wrt $(L_1, L_2, R_1, R_2)$) if
\begin{align}   \label{eq:dec_correct}
\Dec(L_1^{(k_1,k_3)} \cup L_2^{(k_2,\vec{z}_L)}, R_1^{(k_1,k_3)} \cup R_2^{(k_2,\vec{z}_R)}, k_1, k_2, k_3)
= (L_1,R_1,L_2,R_2,\vec{z}_L,\vec{z}_R,k_1,k_2,k_3).
\end{align}
\end{definition}

\begin{remark}
There are cases that $\Dec(L_1^{(k_1,k_3)} \cup L_2^{(k_2,\vec{z}_L)}, R_1^{(k_1,k_3)} \cup R_2^{(k_2,\vec{z}_R)}, k_1, k_2, k_3) \neq \bot$ while does not match the correct $(L_1, R_1, L_2, R_2, \vec{z}_L, \vec{z}_R, k_1, k_2, k_3)$. For example, consider the case where $L_1 = \set{(a,b),(b \oplus k_2, c)}$ and $L_2 = R_1 = R_2 = \varnothing$, then $\Dec$ will output $L_{\isolate} = \varnothing$ and $L_{\pair} = \set{(a,c)}$.
\end{remark}

\noindent For technical reasons, we require a stronger condition than decodability. Given a decodable tuple $(L_1, L_2, \allowbreak R_1, R_2, \vec{z}_L, \vec{z}_R, k_1, k_2, k_3)$, we use $L$ and $R$ as shorthand for $L_1^{(k_1,k_3)} \cup L_2^{(k_2,\vec{z}_L)}$ and $R_1^{(k_1,k_3)} \cup R_2^{(k_2,\vec{z}_R)}$, respectively. In particular, we are concerned with the \emph{robustness} of these tuples. Suppose we delete an element $v$ from $L$. We are interested in whether $\Dec(L \setminus \set{v}, R, k_1, k_2, k_3)$ remains non-$\bot$. 

\begin{definition}[Robust Decodability]  \label{def:robust_decodable_tuples} 
For any $L_1, L_2 \in \Rinj$, $R_1, R_2 \in \RDdist$, a tuple $(k_1, k_2, k_3, \vec{z}_L, \vec{z}_R)$ said to be \emph{robustly decodable} (\wrt $(L_1, L_2, R_1, R_2)$) if it is decodable and satisfies
\begin{enumerate}
    \item $\Dec\left( L_1^{(k_1,k_3)} \cup L_2^{(k_2,\vec{z}_L)} \setminus \set{v}, R_1^{(k_1,k_3)} \cup R_2^{(k_2,\vec{z}_R)}, k_1, k_2, k_3 \right) \neq \bot$\, for all $v \in L_1^{(k_1,k_3)} \cup L_2^{(k_2,\vec{z}_L)}$, and
    \item $\Dec\left(L_1^{(k_1,k_3)} \cup L_2^{(k_2,\vec{z}_L)}, R_1^{(k_1,k_3)} \cup R_2^{(k_2,\vec{z}_R)} \setminus \set{u}, k_1, k_2, k_3 \right) \neq \bot$\, for all $u \in R_1^{(k_1,k_3)} \cup R_2^{(k_2,\vec{z}_R)}$.
\end{enumerate}
\end{definition}

\noindent In the following, we list the sufficient conditions for robustly decodable tuples and introduce the shorthands $\bfz \coloneqq (\vec{z}_L, \vec{z}_R)$ and $\bfk \coloneqq (k_1, k_2, k_3)$.

\begin{lemma}[Sufficient conditions for robust decodability] \label{lem:conditions_robust_decodability}
Let $L_1, L_2 \in \Rinj$, $R_1, R_2 \in \RDdist$. Every $5$-tuple $(k_1, k_2, k_3, \vec{z}_L, \vec{z}_R) \in [N] \times [N] \times [N] \times [N]^{|L_2|} \times [N]^{|R_2|}$ that satisfies all the following conditions is decodable:
\begin{enumerate}
    \item \textbf{Distinctness:} $L_1^{(k_1,k_3)}, L_2^{(k_2,\vec{z}_L)} \in \Rinj$ and $R_1^{(k_1,k_3)}, R_2^{(k_2,\vec{z}_R)} \in \RDdist$
    \item \textbf{Disjointness:} $\Im(L_1^{(k_1,k_3)}) \cap \Im(L_2^{(k_2,\vec{z}_L)}) = \varnothing$ and $\Dom(R_1^{(k_1,k_3)}) \cap \Dom(R_2^{(k_2,\vec{z}_R)}) = \varnothing$
    \item \textbf{No extra $k_2$-correlated pairs:} There are exactly $|L_2|$ number of pairs $((x, y), (x', y'))$ with $(x, y),$ $(x', y') \in L_1^{(k_1,k_3)} \cup L_2^{(k_2,\vec{z}_L)}$ such that $x' = y \oplus k_2$, and there are exactly $|R_2|$ number of pairs $((x, y), (x', y'))$ with $(x, y), (x', y') \in R_1^{(k_1,k_3)} \cup R_2^{(k_2,\vec{z}_L)}$ such that $y' = x \oplus k_2$
\end{enumerate}
Furthermore, for all such tuples, it holds that
\begin{align}      \label{eq:Osplit_on_augmented_relation}
\Osprus \ket{L^{(k_1,k_3)}_1 \cup L_2^{(k_2,\vec{z}_L)}}_{\reg{S}} 
\ket{R^{(k_1,k_3)}_1 \cup R_2^{(k_2,\vec{z}_R)}}_{\reg{T}} 
\ket{\bfk}_{\reg{K}} = 
\ket{L_1}_{\reg{S}_1} \ket{R_1}_{\reg{T}_1} 
\ket{L_2}_{\reg{S}_2} \ket{R_2}_{\reg{T}_2} 
\ket{\bfz}_{\reg{Z}}
\ket{\bfk}_{\reg{K}}.
\end{align}
-- For $(x,y) \in L_1^{(k_1,k_3)}$,
\begin{align}
& \Osprus \ket{L^{(k_1,k_3)}_1 \cup L_2^{(k_2,\vec{z}_L)} \setminus \set{(x,y)}}_{\reg{S}} 
\ket{R^{(k_1,k_3)}_1 \cup R_2^{(k_2,\vec{z}_R)}}_{\reg{T}} 
\ket{\bfk}_{\reg{K}} 
\notag\\
&\hspace{.3\textwidth} = \ket{L_1 \setminus \set{(x \oplus k_1, y \oplus k_3)}}_{\reg{S}_1} \ket{R_1}_{\reg{T}_1} 
\ket{L_2}_{\reg{S}_2} \ket{R_2}_{\reg{T}_2} \ket{\bfz}_{\reg{Z}} \ket{\bfk}_{\reg{K}}.
\label{eq:robust_isolate}
\end{align}
-- For $(x,y) \in L_{2,\source}^{(k_2,\vec{z}_L)}$, suppose $i \in |L_2|$ is the index such that $y = z_{L,i}$. Let $y_i$ denote the $i$-th largest element in $\Im(L_2)$ and $x_i$ is the unique element such that $(x_i,y_i) \in L_2$. Then $x = x_i$ and
\begin{align}
& \Osprus \ket{L^{(k_1,k_3)}_1 \cup L_2^{(k_2,\vec{z}_L)} \setminus \set{(x,y)}}_{\reg{S}} 
\ket{R^{(k_1,k_3)}_1 \cup R_2^{(k_2,\vec{z}_R)}}_{\reg{T}} 
\ket{\bfk}_{\reg{K}} 
\notag\\
&\hspace{.15\textwidth} = \ket{L_1 \cup \set{(z_{L,i} \oplus k_2 \oplus k_1, y_i \oplus k_3)}}_{\reg{S}_1} \ket{R_1}_{\reg{T}_1} 
\ket{L_2 \setminus \set{(x_i, y_i)}}_{\reg{S}_2} \ket{R_2}_{\reg{T}_2} \ket{\bfz}_{\reg{Z}} \ket{\bfk}_{\reg{K}}.
\label{eq:robust_source}
\end{align}
-- For $(x,y) \in L_{2,\target}^{(k_2,\vec{z}_L)}$, let $i \in [|L_2|]$ be the index such that $y$ is the $i$-th largest element in $\Im(L_2)$, denoted by $y_i$, and let $x_i$ be the unique value satisfying $(x_i,y_i) \in L_2$. Then $x = z_{L,i} \oplus k_2$ and
\begin{align}
& \Osprus \ket{L^{(k_1,k_3)}_1 \cup L_2^{(k_2,\vec{z}_L)} \setminus \set{(x,y)}}_{\reg{S}} 
\ket{R^{(k_1,k_3)}_1 \cup R_2^{(k_2,\vec{z}_R)}}_{\reg{T}} 
\ket{\bfk}_{\reg{K}}
\notag\\
&\hspace{.15\textwidth} = \ket{L_1 \cup \set{(x_i \oplus k_1, z_{L,i} \oplus k_3)}}_{\reg{S}_1} \ket{R_1}_{\reg{T}_1} 
\ket{L_2 \setminus \set{(x_i, y_i)}}_{\reg{S}_2} \ket{R_2}_{\reg{T}_2} \ket{\bfz}_{\reg{Z}} \ket{\bfk}_{\reg{K}}.
\label{eq:robust_target}
\end{align}
\end{lemma}

\begin{proof}
The proof consists of two steps. First, we establish decodability. Then, by reusing the same argument, we prove the remaining part of the lemma. \\

\noindent {\bf Step~1:\;} Fix the tuple. Let $L$ be shorthand for $L^{(k_1,k_3)}_1 \cup L_2^{(k_2,\vec{z}_L)}$. Recall~\Cref{def:decomp_graphs,def:rel_key_induced_graph}. We first show that $G^{\ell}_{L,k_2}$ is decomposable. The first and second conditions together imply that each vertex in $G^{\ell}_{L,k_2}$ has a distinct label. Next, note that from the definition of augmented relation $L_2^{(k_2,\vec{z}_L)}$ (see~\Cref{def:aug_relations}), there are at least $|L_2|$ edges in $G^{\ell}_{L,k_2}$. Together with the third condition, they imply that these $|L_2|$ edges form a perfect matching of vertices in $L_2^{(k_2,\vec{z}_L)}$. Thus, $G^{\ell}_{L,k_2}$ is decomposable. In particular, we have
\begin{equation}    \label{eq:three_parts_good_tuples}
\begin{aligned}
    & V_{\isolate}(G^{\ell}_{L,k_2}) = L_1^{(k_1,k_3)} \\
    & V_{\source}(G^{\ell}_{L,k_2}) = \set{(x,y) \in L_2^{(k_2,\vec{z}_L)} \colon y \in \set{z_{L,1}, \ldots, z_{L,|L_2|}}} \\
    & V_{\target}(G^{\ell}_{L,k_2}) = L_2^{(k_2,\vec{z}_L)} \setminus V_{\source}(G^{\ell}_{L,k_2})
\end{aligned}
\end{equation}
By symmetry, we can show that $G^{r}_{R,k_2}$ is also decomposable, so the input will pass Step~1 of~$\Dec$ in~\Cref{def:Ospru_split}. Next, the first condition ensures the image of $V_{\target}(G^{\ell}_{L,k_2})$ is $\cI$-distinct and $V_{\target}(G^{r}_{L,k_2})$ is $\cD$-distinct, so the input will pass Step~2. By inspection, after Step~4, we obtain $L_{\isolate} = L_1, R_{\isolate} = R_1, L_{\pair} = L_2, R_{\pair} = R_2, \vec{m}_L = \vec{z}_L, \vec{m}_L = \vec{z}_R$. Finally, the first condition ensures that the input passes Steps~5 and 6. This shows that the tuple is decodable and~\Cref{eq:Osplit_on_augmented_relation}. \\

\noindent{\bf Step~2:\;} In the view of relation-key induced graphs, deleting $v = (x,y)$ from $L$ is equivalently as deleting $v$ from $G^{\ell}_{k_2,L}$, together with all edges incident to $v$, that is, all edges either entering or leaving $v$. We denote this resulting graph by $G^{\ell}_{k_2,L}\!-\!v$. Suppose $v \in V_{\isolate}(G^{\ell}_{L,k_2})$ and $L_1$ satisfies $L_1 = L'_1 \cup \set{(x \oplus k_1, y \oplus k_3)}$ for some $L'_1$. From~\Cref{eq:three_parts_good_tuples}, it is not hard to verify that
\begin{align*}
    & V_{\isolate}(G^{\ell}_{L,k_2}\!-\!v) = (L'_1)^{(k_1,k_3)} \\
    & V_{\source}(G^{\ell}_{L,k_2}\!-\!v) = \set{(x,y) \in L_2^{(k_2,\vec{z}_L)} \colon y \in \set{z_{L,1}, \ldots, z_{L,|L_2|}}} \\
    & V_{\target}(G^{\ell}_{L,k_2}\!-\!v) = L_2^{(k_2,\vec{z}_L)} \setminus V_{\source}(G^{\ell}_{L,k_2})
\end{align*}
and they pass Steps~5 and~6 of~$\Dec$. Thus, we proved~\Cref{eq:robust_isolate}.

When $v \in V_{\source}(G^{\ell}_{L,k_2})$. Suppose $L_2 = \set{(x_i,y_i)}_{i \in [|L_2|]}$ and $L_2^{(k_2,\vec{z}_L)} = \set{(x_i,z_i),(z_i \oplus k_2, y_i)}_{i \in [|L_2|]}$. In the former case, $(x,y) = (x_{i^*},z_{i^*})$ for some $i^* \in [|L_2|]$. As a result, the vertex $(z_{i^*} \oplus k_2, y_{i^*})$ in the graph $G^{\ell}_{L,k_2}\!-\!v$ becomes ``unpaired''. We may view $(z_{i^*} \oplus k_2, y_{i^*})$ being assigned to the isolated vertices. Namely,
\begin{align*}
    & V_{\isolate}(G^{\ell}_{L,k_2}\!-\!v) = (L_1 \cup \set{(z_{i^*} \oplus k_2 \oplus k_1, y_{i^*} \oplus k_3)})^{(k_1,k_3)} \\
    & V_{\source}(G^{\ell}_{L,k_2}\!-\!v) = \set{(x_i,z_i)}_{i \neq i^*} \\
    & V_{\target}(G^{\ell}_{L,k_2}\!-\!v) = \set{(z_i \oplus k_2, y_i)}_{i \neq i^*}
\end{align*}
At first sight, the resulting graph remains decomposable. However, a caveat is that the newly added element $(z_{i^*} \oplus k_2 \oplus k_1,\, y_{i^*} \oplus k_3)$ may coincide with an element of $L_1$, which would cause Step~5 of the decoder $\Dec$ (\Cref{def:Ospru_split}) to fail. Fortunately, Item~2 prevent this event from happening. This proves~\Cref{eq:robust_source}. 

Finally, if $v \in V_{\target}(G^{\ell}_{L,k_2})$, we have $(x,y) = (z_{i^*} \oplus k_2, y_{i^*})$ for some $i^* \in [|L_2|]$ and
\begin{align*}
    & V_{\isolate}(G^{\ell}_{L,k_2}\!-\!v) = (L_1 \cup \set{(x_{i^*} \oplus k_1, z_{i^*} \oplus k_3)})^{(k_1,k_3)} \\
    & V_{\source}(G^{\ell}_{L,k_2}\!-\!v) = \set{(x_i,z_i)}_{i \neq i^*} \\
    & V_{\target}(G^{\ell}_{L,k_2}\!-\!v) = \set{(z_i \oplus k_2, y_i)}_{i \neq i^*}.
\end{align*}
Similarly, Item~2 prevent this event from happening. This proves~\Cref{eq:robust_target}. 
\end{proof}

\subsection{Good Tuples and Their Combinatorial Properties}

\noindent In the following, we define good tuples.

\begin{definition}[Good tuples] \label{def:good_tuples}
For any $L_1, L_2 \in \Rinj$, $R_1, R_2 \in \RDdist$, let $\Sspru\left(\substack{L_1,L_2 \\ R_1,R_2}\right)$ denote the set of $5$-tuples $(k_1, k_2, k_3, \vec{z}_L, \vec{z}_R) \in [N] \times [N] \times [N] \times [N]^{|L_2|} \times [N]^{|R_2|}$ satisfying all the following conditions:\footnote{Recall that for a set $A \subseteq \bit^n$ and a string $x \in \bit^n$, we denote $A \oplus x \coloneqq \set{a \oplus x : a \in A}$.}
\begin{enumerate}
    \item \label{item:k1} 
    $k_1 \notin \Big( \Dom(L_1) \oplus \Dom(L_2) \Big) \cup \Big( \Dom(R_1) \oplus \Dom(R_2) \Big)$
    
    \item \label{item:k2} \(
        k_2 \notin 
        \Bigg( \bigg( \Big( \Dom(L_1) \oplus k_1 \Big) \cup \Dom(L_2) \bigg) 
        \oplus 
        \bigg( \Big( \Im(L_1) \oplus k_3 \Big) \cup \Im(L_2) \bigg) \Bigg) \\
        \cup 
        \Bigg( \bigg( \Big( \Dom(R_1) \oplus k_1 \Big) \cup \Dom(R_2) \bigg) 
        \oplus 
        \bigg( \Big( \Im(R_1) \oplus k_3 \Big) \cup \Im(R_2) \bigg) \Bigg)
    \)
    
    \item \label{item:k3} 
    $k_3 \notin \Big( \Im(L_1) \oplus \Im(L_2) \Big) \cup \Big( \Im(R_1) \oplus \Im(R_2) \Big)$
    
    \item \label{item:zL_dist} 
    $\vec{z}_L \in [N]^{|L_2|}_{\dist}$ \; and \; $\vec{z}_R \in [N]^{|R_2|}_{\dist}$
    
    \item \label{item:zL}   
    $\set{\vec{z}_L}$ and $\Big( \Im(L_1) \oplus k_3 \Big) \cup \Im(L_2) \cup \Bigg( \bigg( \Big( \Dom(L_1) \oplus k_1 \Big) \cup \Dom(L_2) \bigg) \oplus k_2 \Bigg)$ are disjoint.
    
    \item \label{item:zR} 
    $\set{\vec{z}_R}$ and $\Big( \Im(R_1) \oplus k_3 \Big) \cup \Im(R_2) \cup \Bigg( \bigg( \Big( \Dom(R_1) \oplus k_1 \Big) \cup \Dom(R_2) \bigg) \oplus k_2 \Bigg)$  are disjoint.
\end{enumerate}
\end{definition}

\begin{lemma}  \label{lem:good_tuples_satisfy_conditions}
For any $L_1, L_2 \in \Rinj$, $R_1, R_2 \in \RDdist$, every tuple in $\Sspru\left(\substack{L_1,L_2 \\ R_1,R_2}\right)$ satisfies all conditions in~\Cref{lem:conditions_robust_decodability} and is therefore robustly decodable.
\end{lemma}

\noindent The proof proceeds by a direct verification of all conditions and is given in~\Cref{sec:missing_defs_proofs}. \\

\begin{remark}
Indeed, one could have alternative definitions for good tuples. In particular, we only require them to be robustly decodable and satisfy all properties introduced in the following.
\end{remark}

\myparagraph{Enumerating good tuples} Here we describe how to enumerate tuples from $\Sspru\left(\substack{L_1,L_2 \\ R_1,R_2}\right)$. First, for any $L_1, L_2 \in \Rinj$ and $R_1, R_2 \in \RDdist$, we define the sets $\cB_1(L_1, L_2, R_1, R_2)$ and $\cB_3(L_1, L_2, R_1, R_2)$, which can be viewed as ``bad $k_1$ and $k_3$'':
\begin{align*}
& \cB_1(L_1, L_2, R_1, R_2) 
\coloneqq \big( \Dom(L_1)\oplus\Dom(L_2) \big) \cup \big( \Dom(R_1) \oplus\Dom(R_2) \big), \\
& \cB_3(L_1, L_2, R_1, R_2) 
\coloneqq \big( \Im(L_1)\oplus \Im(L_2) \big) \cup \big( \Im(R_1) \oplus \Im(R_2) \big).
\end{align*}
Next, for any $L_1, L_2 \in \Rinj, R_1, R_2 \in \RDdist$, and ``good $(k_1, k_3)$'' pair, namely, $k_1 \notin \cB_1(L_1, L_2, R_1, R_2)$ and $k_3 \notin \cB_3(L_1, L_2, R_1, R_2)$, define the set $\cB_2(L_1, L_2, R_1, R_2, k_1, k_3)$ consisting of ``bad $k_2$'':
\begin{align*}
\cB_2(L_1, L_2, R_1, R_2, k_1, k_3) 
& \coloneqq \Bigg( \bigg( \Big( \Dom(L_1) \oplus k_1 \Big) \cup \Dom(L_2) \bigg) \oplus \bigg( \Big( \Im(L_1) \oplus k_3 \Big) \cup \Im(L_2) \bigg) \Bigg) 
\notag \\
& \qquad \cup \Bigg( \bigg( \Big( \Dom(R_1) \oplus k_1 \Big) \cup \Dom(R_2) \bigg) \oplus \bigg( \Big( \Im(R_1) \oplus k_3 \Big) \cup \Im(R_2) \bigg) \Bigg).
\end{align*}
Finally, for any $L_1, L_2 \in \Rinj, R_1, R_2 \in \RDdist$ and ``good $(k_1, k_2, k_3)$'' such that $k_1 \notin \cB_1(L_1, L_2, R_1, R_2)$, $ k_3 \notin \cB_3(L_1, L_2, R_1, R_2)$, and $k_2 \notin \cB_2(L_1, L_2, R_1, R_2, k_1, k_3)$, define the sets consisting of ``bad $\vec{z}_L$ and $\vec{z}_R$'' as
\begin{align*}
& \cB_L \left(\substack{L_1,L_2,\\ R_1,R_2,\\ k_1,k_2,k_3}\right) \coloneqq 
\Bigg\{ \vec{z}_L \in [N]^{|L_2|}: \vec{z}_L \notin [N]^{|L_2|}_\dist
\notag \\
& \qquad \lor \ \exists i \in [|L_2|] \st z_{L,i} \in 
\Big( \Im(L_1) \oplus k_3 \Big) 
\cup \Im(L_2) 
\cup \Bigg( \bigg( \Big( \Dom(L_1) \oplus k_1 \Big) \cup \Dom(L_2) \bigg) \oplus k_2 \Bigg) 
\Bigg\},
\end{align*}
\begin{align*}
& \cB_R \left(\substack{L_1,L_2,\\ R_1,R_2,\\ k_1,k_2,k_3}\right) \coloneqq
\Bigg\{ \vec{z}_R \in [N]^{|R_2|}: \vec{z}_R \notin [N]^{|R_2|}_{\dist}
\notag \\ 
& \qquad \lor \ \exists i \in [|R_2|] \st z_{R,i} \in 
\Big( \Im(R_1) \oplus k_3 \Big)
\cup \Im(R_2) 
\cup \Bigg( \bigg( \Big( \Dom(R_1) \oplus k_1 \Big) \cup \Dom(R_2) \bigg) \oplus k_2 \Bigg)
\Bigg\}.
\end{align*}
Therefore, one can enumerate elements in $\Sspru\left(\substack{L_1,L_2 \\ R_1,R_2}\right)$ by first choosing $(k_1,k_3)$, then $k_2$ second, and finally $(\vec{z}_L,\vec{z}_R)$, ensuring that each choice avoids being bad. In addition, if all $L_1, L_2, R_1, R_2$ are all of polynomial size, then most choices of $k_1, k_2, k_3, \vec{z}_L$ and $\vec{z}_R$ are good. This is simply because the number of good keys they can eliminate is at most quadratic in the size of the relations.

\begin{lemma}   \label{lem:number_of_good_tuples}
For any integer $t \geq 0$, $L_1, L_2 \in \Rinj_{\leq t}$, $R_1, R_2 \in \RDdist_{\leq t}$, the following holds. 
\begin{enumerate}
    \item $\cB_1(L_1,L_2,R_1,R_2)$ and $\cB_3(L_1,L_2,R_1,R_2)$ each occupy at most a $2t^2/N$ fraction of the universe $[N]$.
    \item For any $k_1 \notin \cB_1(L_1,L_2,R_1,R_2)$ and $k_3 \notin \cB_3(L_1,L_2,R_1,R_2)$, the set $\cB_2(L_1,L_2,R_1,R_2,k_1,k_3)$ occupies at most a $8t^2/N$ fraction of the universe $[N]$.
    \item For any $k_1 \notin \cB_1(L_1,L_2,R_1,R_2)$, $k_3 \notin \cB_3(L_1,L_2,R_1,R_2)$, and $k_2 \notin \cB_2(L_1,L_2,R_1,R_2,k_1,k_3)$, the sets $\cB_L\!\left(\substack{L_1,L_2\\ R_1,R_2\\ k_1,k_2,k_3}\right)$ and $\cB_R\!\left(\substack{L_1,L_2\\ R_1,R_2\\ k_1,k_2,k_3}\right)$ each occupy at most a $5t^2/N$ fraction of their respective universes, $[N]^{|L_2|}$ and $[N]^{|R_2|}$.
\end{enumerate}
As a corollary, all but a $22t^2/N$ fraction of the elements in $[N]^3 \times [N]^{|L_2|} \times [N]^{|R_2|}$ are good tuples.
\end{lemma}    
\begin{proof}
Notice that for any sets $A,B$, the size of the set $A \oplus B$ is at most $|A| \cdot |B|$. Thus, Item~1 follows. Items~2 and~3 can be proved in a similar manner.
\end{proof}

Good tuples satisfy the following property.
\begin{lemma}[Monotonicity]  \label{lem:good_tuples_monotonicity}
For any $x \in [N], L_1, L_2 \in \Rinj$, $R_1, R_2 \in \RDdist$, the following holds:
\begin{enumerate}
    \item There do not exist $(k_1, k_2, k_3, \vec{z}_L, \vec{z}_R) \notin \Sspru\left(\substack{L_1,L_2 \\ R_1,R_2}\right)$ and $y \notin \Im(L_1)$ such that $(k_1, k_2, k_3, \vec{z}_L, \vec{z}_R) \in \Sspru\left(\substack{L_1 \cup \set{(x,y)},L_2 \\ R_1,R_2}\right)$.
    \item There do not exist $(k_1, k_2, k_3, \vec{z}_L, \vec{z}_R) \notin \Sspru\left(\substack{L_1,L_2 \\ R_1,R_2}\right), z \in [N]$, and $y \notin \Im(L_2)$ such that $(k_1, k_2, k_3, \allowbreak \vec{z}_L^{\,(i \gets z)}, \vec{z}_R) \in \Sspru\left(\substack{L_1, L_2  \cup \set{(x,y)} \\ R_1,R_2}\right)$, where $i$ is the index such that $y$ is the $i$-th largest element in $\Im(L_2) \cup \set{y}$ and $\vec{z}_L^{\,(i \gets z)}$ denotes the vector obtained by inserting $z$ into the $i$-th coordinate of $\vec{z}_L$ and shifting all subsequent coordinates one position to the right.
\end{enumerate}
As a corollary, for any $x \in [N], L_1, L_2 \in \Rinj$, $R_1, R_2 \in \RDdist$ and $(k_1, k_2, k_3, \vec{z}_L, \vec{z}_R) \notin \Sspru\left(\substack{L_1,L_2 \\ R_1,R_2}\right)$, the following holds:
\begin{enumerate}
    \item For every $y \notin \Im(L_1)$, $\Sspru\left(\substack{L_1 \cup \set{(x,y)},L_2 \\ R_1,R_2}\right) \subseteq \Sspru\left(\substack{L_1,L_2 \\ R_1,R_2}\right)$.
    \item For every $y \notin \Im(L_2)$, $(k_1, k_2, k_3, \vec{z}_L, \vec{z}_R) \in \Sspru\left(\substack{L_1, L_2  \cup \set{(x,y)} \\ R_1,R_2}\right)$, it holds that $(k_1, k_2, k_3, \vec{z}_{L,-i}, \vec{z}_R) \in \Sspru\left(\substack{L_1, L_2 \\ R_1,R_2}\right)$, where $i$ is the index such that $y$ is the $i$-th largest element in $\Im(L_2) \cup \set{y}$ and $\vec{z}_{L,- i}$ denotes the vector obtained by deleting its $i$-th coordinate and shifting all subsequent coordinates one position to the left.
\end{enumerate}
\end{lemma}
\begin{proof}
Since $(k_1, k_2, k_3, \vec{z}_L, \vec{z}_R) \notin \Sspru\left(\substack{L_1,L_2 \\ R_1,R_2}\right)$, some condition in~\Cref{def:good_tuples} must be violated. To prove Item~1, simply note that the same condition remains violated after adding $(x,y)$ to $L_1$. To prove Item~2, observe that the conditions in~\Cref{def:good_tuples} are insensitive to the ordering of $\vec{z}_L$ and $\vec{z}_R$. Thus, although inserting $z$ into $\vec{z}_L$ changes its size and ordering, the same condition remains violated.
\end{proof}

The following lemmas will be used in~\Cref{sec:spru:1st_oracle}.

\begin{lemma} \label{lem:count_x}
For any $t \ge 0, x \in [N], L_1, L_2 \in \Rinj_{\leq t}$, $R_1, R_2 \in \RDdist_{\leq t}$, there are at most a $t(22t+4)/N$ fraction of the elements $(\bfk,\bfz) \in [N]^3 \times [N]^{|L_2|} \times [N]^{|R_2|}$ for which no $y \notin \Im(L_1)$ satisfies $(\bfk,\bfz) \in \Sspru\!\left(\substack{L_1 \cup \{(x,y)\},L_2 \\ R_1,R_2}\right)$.
\end{lemma}

\begin{proof}
Let us sample $(\bfk,\bfz)$ uniformly at random from $[N]^3 \times [N]^{|L_2|} \times [N]^{|R_2|}$. Define the following events
\begin{itemize}
    \item $\mathsf{Good} \colon (\bfk,\bfz) \in \Sspru\!\left(\substack{L_1,L_2 \\ R_1,R_2}\right)$
    \item $\mathsf{K}_1 \colon k_1 \in x \oplus \Dom(L_2)$
    \item $\mathsf{K}_2 \colon k_2 \in x \oplus k_1 \oplus \bigg( \Big( \Im(L_1) \oplus k_3 \Big) \cup \Im(L_2) \bigg)$
    \item $\mathsf{Z}_L \colon \set{\vec{z}_L} \cap \set{x \oplus k_1 \oplus k_2} \neq \varnothing$
\end{itemize}
Suppose $\mathsf{Good}$ occurs, and consider the items in~\Cref{def:good_tuples}. 
For there to be no $y$ such that 
$(\bfk,\bfz) \in \Sspru\!\left(\substack{L_1 \cup \{(x,y)\},L_2 \\ R_1,R_2}\right)$, 
one of the conditions imposed by the addition of $x$ to $\Dom(L_1)$ must be violated. 
These correspond to the events $\mathsf{K}_1$, $\mathsf{K}_2$, and $\mathsf{Z}_L$. 
Otherwise, there always exists some $y$ such that 
$(\bfk,\bfz) \in \Sspru\!\left(\substack{L_1 \cup \{(x,y)\},L_2 \\ R_1,R_2}\right)$. 
Therefore, by the union bound, the fraction of such $(\bfk,\bfz)$ is at most
\[
\Pr[\neg \mathsf{Good}] + \Pr[\mathsf{K}_1] + \Pr[\mathsf{K}_2] + \Pr[\mathsf{Z}_L]
   \;\le\; \frac{22t^2}{N} + \frac{t}{N} + \frac{2t}{N} + \frac{t}{N} 
   \;=\; \frac{t(22t+4)}{N},
\]
where $\Pr[\mathsf{K}_1]$, $\Pr[\mathsf{K}_2]$, and $\Pr[\mathsf{Z}_L]$ are bounded 
by sampling $k_1$, $k_2$, and $\vec{z}_L$, respectively, at the end.
\end{proof}

\begin{lemma} \label{lem:comb_F1Ldagger_y_difference}
For any $t \ge 0, x \in [N]$, $L_1, L_2 \in \Rinj_{\leq t}$, 
$R_1, R_2 \in \RDdist_{\leq t}$, $\bfk \in [N]^3$, $\bfz \in [N]^{|L_2|} \times [N]^{|R_2|}$, if there exists $y \notin \Im(L_1)$ such that $(\bfk, \bfz) \in \Sspru\!\left(\substack{L_1 \cup \set{(x,y)},L_2 \\ R_1,R_2}\right)$, then there are at most $4t$ values of $y \notin \Im(L_1)$ such that $(\bfk, \bfz) \notin \Sspru\!\left(\substack{L_1 \cup \set{(x,y)},L_2 \\ R_1,R_2}\right)$.
\end{lemma}

\begin{proof}
Suppose there exists some $y \notin \Im(L_1)$ such that $(\bfk, \bfz) \in \Sspru\!\left(\substack{L_1 \cup \set{(x,y)},L_2 \\ R_1,R_2}\right)$. According to the definition of good tuples, $(\bfk, \bfz)$ satisfies all conditions in~\Cref{def:good_tuples}. 
Now, suppose we attempt to vary $y$ in such a way that $(\bfk, \bfz)$ violates one of the conditions. Since we are only modifying $\Im(L_1 \cup \set{(x,y)})$, Items~1, 4, and~6 continue to hold. We now examine Item~2. If we replace $y$ with a different value $y'$, the additional imposed constraint is
\begin{align*}
    k_2 \notin \bigg( \big(\Dom(L_1) \oplus k_1 \big) \cup \Dom(L_2) \bigg) \oplus y' \oplus k_3.
\end{align*}
Thus, in order for this condition to be violated, we must have
\begin{align*}
    y' \in \bigg( \big(\Dom(L_1) \oplus k_1 \big) \cup \Dom(L_2) \bigg) \oplus k_2 \oplus k_3,
\end{align*}
and there are at most $2t$ such values of $y'$. Similarly, there are at most $t$ values of $y'$ that violate Item~3 and at most $t$ values of $y'$ that violate Item~5. 
Hence, the total number of such $y'$ is bounded by $4t$.
\end{proof}

\begin{lemma} \label{lem:comb_F1L_tuple_difference}
For any $t \ge 0, x \in [N]$, $L_1, L_2 \in \Rinj_{\leq t}$, 
$R_1, R_2 \in \RDdist_{\leq t}$, define the following two sets:
\begin{align*}
    & \Psi_{\phantom{x,}L_1,R_1,L_2,R_2} \coloneqq 
    \set{
    (y, \bfk, \bfz) \colon 
        (\bfk, \bfz) \in \Sspru\left(\substack{L_1,L_2, \\ R_1,R_2}\right), 
        y \notin \Im(L_1)  \cup \qty( \Im(L_2^{(k_2,\vec{z}_L)}) \oplus k_3 )
    } \\
    & \Phi_{x,L_1,R_1,L_2,R_2} \coloneqq
    \set{
    (y, \bfk, \bfz) \colon 
        y \notin \Im(L_1),
        (\bfk, \bfz) \in \Sspru\left(\substack{L_1 \cup \set{(x,y)},L_2, \\ R_1,R_2}\right)
    }.
\end{align*}
Then $\Psi_{\phantom{x,}L_1,R_1,L_2,R_2} \supseteq \Phi_{x,L_1,R_1,L_2,R_2}$ and the size of their difference is at most $t(22t+8)N^{|L_2|+|R_2|+3}$.
\end{lemma}

\begin{proof}
Consider elements in $\Phi_{x,L_1,R_1,L_2,R_2}$. First, from~\Cref{lem:good_tuples_satisfy_conditions} and Items~1 and~2 in~\Cref{lem:conditions_robust_decodability}, we have 
\[
(L_1 \cup \set{(x,y)})^{(k_1,k_3)} \in \RIdist
\quad \text{and} \quad
\Im\left((L_1 \cup \set{(x,y)})^{(k_1,k_3)}\right) \cap \Im(L_2^{(k_2,\vec{z}_L)}) = \varnothing,
\]
which implies $y \notin \Im(L_1) \cup \left( \Im(L_2^{(k_2,\vec{z}_L)}) \oplus k_3 \right)$. Then by monotonicity (Item~1 in~\Cref{lem:good_tuples_monotonicity}), we have 
\[
\Sspru\left(\substack{L_1 \cup \set{(x,y)},L_2, \\ R_1,R_2}\right)
\subseteq \Sspru\left(\substack{L_1, L_2, \\ R_1, R_2}\right).
\]
This proves $\Psi_{\phantom{x,}L_1,R_1,L_2,R_2} \supseteq \Phi_{x,L_1,R_1,L_2,R_2}$.

Next, we bound the size of their difference. We denote by $\mathsf{BAD}$ the set of 
$(\bfk,\bfz) \in \Sspru\!\left(\substack{L_1,L_2 \\ R_1,R_2}\right)$ for which no 
$y \notin \Im(L_1)$ satisfies 
$(\bfk,\bfz) \in \Sspru\!\left(\substack{L_1 \cup \{(x,y)\},L_2 \\ R_1,R_2}\right)$. 
By~\Cref{lem:count_x}, we have 
\[
|\mathsf{BAD}| \le t(22t+4)N^{|L_2|+|R_2|+2}.
\]
Now, if we enumerate $(y,\bfk,\bfz)$ in $\Psi_{\phantom{x,}L_1,R_1,L_2,R_2}$, then one of the following holds:
\begin{enumerate}
    \item If $(\bfk,\bfz) \in \mathsf{BAD}$, then every $y \notin \Im(L_1)$ satisfies 
    $(\bfk,\bfz) \notin \Sspru\!\left(\substack{L_1 \cup \{(x,y)\},L_2 \\ R_1,R_2}\right)$, which in turn implies $(y,\bfk,\bfz) \notin \Phi_{x,L_1,R_1,L_2,R_2}$.
    \item If $(\bfk,\bfz) \notin \mathsf{BAD}$, then by~\Cref{lem:comb_F1Ldagger_y_difference} 
    there are at most $4t$ values of $y \notin \Im(L_1)$ such that 
    $(\bfk,\bfz) \notin \Sspru\!\left(\substack{L_1 \cup \{(x,y)\},L_2 \\ R_1,R_2}\right)$. All other $y$ satisfy $(y,\bfk,\bfz) \in \Phi_{x,L_1,R_1,L_2,R_2}$.
\end{enumerate}
Therefore, the size of their difference is at most
\begin{align*}
    |\mathsf{BAD}| \cdot N + N^{|L_2|+|R_2|+3} \cdot 4t
    \le t(22t+8)N^{|L_2|+|R_2|+3}.
    \tag*{\qedhere}
\end{align*}
\end{proof}

The following lemmas will be used in~\Cref{sec:spru:2nd_oracle}.

\begin{lemma} \label{lem:count_x_z}
For any $t \ge 0, x \in [N], L_1, L_2 \in \Rinj_{\leq t}$, $R_1, R_2 \in \RDdist_{\leq t}$, there are at most a $(22t^2+10t+1)/N$ fraction of elements $(z,\bfk,\vec{z}_L,\vec{z}_R) \in [N] \times [N]^3 \times [N]^{|L_2|} \times [N]^{|R_2|}$ for which no $y \notin \Im(L_2)$ satisfies $(\bfk,\vec{z}^{\,(i \gets z)}_L,\vec{z}_R) \in \Sspru\!\left(\substack{L_1,L_2 \cup \{(x,y)\} \\ R_1,R_2}\right)$ where $i$ is the index such that $y \in_i \Im(L_2) \cup \set{y}$.
\end{lemma}

\begin{proof}
Let us sample $(z,\bfk,\vec{z}_L,\vec{z}_R)$ uniformly at random from $[N] \times [N]^3 \times [N]^{|L_2|} \times [N]^{|R_2|}$. Define the following events
\begin{itemize}
    \item $\mathsf{Good} \colon (\bfk,\vec{z}_L,\vec{z}_R) \in \Sspru\!\left(\substack{L_1,L_2 \\ R_1,R_2}\right)$
    \item $\mathsf{K}_1 \colon k_1 \in \Dom(L_1) \oplus x$
    \item $\mathsf{K}_2 \colon k_2 \in x \oplus \bigg( \Big( \Im(L_1) \oplus k_3 \Big) \cup \Im(L_2) \bigg)$
    \item $\mathsf{Z}_L \colon z \in \set{\vec{z}_L}$ {\bf or} $x \oplus k_2 \in \set{\vec{z}_L}$ {\bf or} $x \in \set{\vec{z}_L}$
    \item $\mathsf{Z} \colon z \in \Big( \Im(L_1) \oplus k_3 \Big) \cup \Im(L_2) \cup \Bigg( \bigg( \Big( \Dom(L_1) \oplus k_1 \Big) \cup \Dom(L_2) \bigg) \oplus k_2 \Bigg)$ {\bf or} $z = x$
\end{itemize}
Suppose $\mathsf{Good}$ occurs, and consider the items in~\Cref{def:good_tuples}. For there to be no $y$ such that $(\bfk,\vec{z}^{\,(i \gets z)}_L,\vec{z}_R) \in \Sspru\!\left(\substack{L_1,L_2 \cup \{(x,y)\} \\ R_1,R_2}\right)$, one of the conditions imposed by the addition of (i) $x$ to $\Dom(L_1)$ or (ii) $z$ to $\vec{z}_L$ must be violated. These correspond to the events $\mathsf{K}_1$, $\mathsf{K}_2$, $\mathsf{Z}_L$, and $\mathsf{Z}$. Otherwise, there always exists some $y$ such that $(\bfk,\vec{z}^{\,(i \gets z)}_L,\vec{z}_R) \in \Sspru\!\left(\substack{L_1,L_2 \cup \{(x,y)\} \\ R_1,R_2}\right)$. Therefore, by the union bound, the fraction of such $(z,\bfk,\vec{z}_L,\vec{z}_R)$ is at most
\[
\Pr[\neg \mathsf{Good}] + \Pr[\mathsf{K}_1] + \Pr[\mathsf{K}_2] + \Pr[\mathsf{Z}_L] + \Pr[\mathsf{Z}]
   \;\le\; \frac{22t^2}{N} + \frac{t}{N} + \frac{2t}{N} + \frac{3t}{N} + \frac{4t+1}{N} 
   \;=\; \frac{22t^2+10t+1}{N},
\]
where $\Pr[\mathsf{K}_1]$, $\Pr[\mathsf{K}_2]$, $\Pr[\mathsf{Z}_L]$, and $\Pr[\mathsf{Z}]$ are bounded by sampling $k_1$, $k_2$, $\vec{z}_L$, and $z$, respectively, at the end.
\end{proof}

\begin{lemma} \label{lem:comb_F2Ldagger_y_difference}
For any $t \ge 0, x \in [N]$, $L_1, L_2 \in \Rinj_{\leq t}$, 
$R_1, R_2 \in \RDdist_{\leq t}$, $\bfk \in [N]^3$, $\vec{z}_L \in [N]^{|L_2|-1}$, $z \in [N]$ and $\vec{z}_R \in [N]^{|R_2|}$, if there exists $y \notin \Im(L_2)$ such that $(\bfk, \vec{z}^{\,(i \gets z)}_L, \vec{z}_R) \in \Sspru\left(\substack{L_1, L_2 \cup \set{(x,y)} \\ R_1, R_2}\right)$ where $i$ is the index such that $y \in_i \Im(L_2) \cup \set{y}$, then there are at most $4t+1$ values of $y \notin \Im(L_2)$ such that $(\bfk, \vec{z}^{\,(i \gets z)}_L, \vec{z}_R) \notin \Sspru\left(\substack{L_1, L_2 \cup \set{(x,y)} \\ R_1, R_2}\right)$ where $i$ is the index such that $y \in_i \Im(L_2) \cup \set{y}$.
\end{lemma}

\begin{proof}
Suppose there exists some $y \notin \Im(L_2)$ such that $(\bfk, \vec{z}^{\,(i \gets z)}_L, \vec{z}_R) \in \Sspru\left(\substack{L_1, L_2 \cup \set{(x,y)} \\ R_1, R_2}\right)$. According to the definition of good tuples, $(\bfk, \vec{z}^{\,(i \gets z)}_L, \vec{z}_R)$ satisfies all conditions in~\Cref{def:good_tuples}. 
Now, suppose we attempt to vary $y$ in such a way that $(\bfk, \vec{z}^{\,(i \gets z)}_L, \vec{z}_R)$ violates one of the conditions. Notice that the index $i$ might vary with the value of $y$ because it is defined to be the index such that $y \in_i \Im(L_2) \cup \set{y}$. Since we are only modifying $\Im(L_2 \cup \{(x,y)\})$, and the conditions in~\Cref{def:good_tuples} depend only on the set $\{\vec{z}^{\,(i \gets z)}_L\}$ rather than on the ordering of $\vec{z}^{\,(i \gets z)}_L$, Items~1, 4, and~6 continue to hold. We can apply the same argument as in the proof of~\Cref{lem:comb_F1Ldagger_y_difference} to show that the total number of such $y'$ is $4t+1$. This completes the proof.
\end{proof}

\begin{lemma} \label{lem:comb_F2L_tuple_difference}
For any $x \in [N]$, $L_1, L_2 \in \Rinj_{\leq t}$, 
$R_1, R_2 \in \RDdist_{\leq t}$, define the following two sets:
\begin{align*}
    & \Psi_{\phantom{x,}L_1,R_1,L_2,R_2} \coloneqq 
    \set{
    (y, z, \bfk, \bfz) \colon 
        (\bfk,\bfz) \in \Sspru\left(\substack{L_1,L_2 \\ R_1,R_2}\right),
        z, y \notin \Im(L_1^{(k_1,k_3)} \cup L_2^{(k_2,\vec{z}_L)}),
        y \neq z'
    } \\
    & \Phi_{x,L_1,R_1,L_2,R_2} \coloneqq
    \set{
    (y, z, \bfk, \bfz) \colon 
        y \notin \Im(L_2),
        i \st y\, \in_i\, \Im(L_2) \cup \set{y},
        (\bfk,\vec{z}_L^{\,(i \gets z)},\vec{z}_R) \in \Sspru\left(\substack{L_1,L_2 \cup \set{(x,y)} \\ R_1,R_2}\right)
    }.
\end{align*}
Then $\Psi_{\phantom{x,}L_1,R_1,L_2,R_2} \supseteq \Phi_{x,L_1,R_1,L_2,R_2}$ and the size of their difference is at most $(22t^2+14t+2)N^{|L_2|+|R_2|+4}$.
\end{lemma}
\begin{proof}
Consider elements in $\Phi_{x,L_1,R_1,L_2,R_2}$. First, from~\Cref{lem:good_tuples_satisfy_conditions} and Items~1 and~2 in~\Cref{lem:conditions_robust_decodability}, we have 
\[
(L_2 \cup \set{(x,y)})^{(k_2,\vec{z}_L^{\,(i \gets z)})} \in \RIdist
\quad \text{and} \quad
\Im(L_1^{(k_1,k_3)}) \cap \Im\left((L_2 \cup \set{(x,y)})^{(k_2,\vec{z}_L^{\,(i \gets z)})}\right) 
= \varnothing,
\]
which implies $z, y \notin \Im(L_1^{(k_1,k_3)} \cup L_2^{(k_2,\vec{z}_L)}) \, \land \, y \neq z'$. Then by monotonicity (Item~2 in~\Cref{lem:good_tuples_monotonicity}), we have 
\[
(\bfk,\vec{z}_L^{\,(i \gets z)},\vec{z}_R) \in \Sspru\left(\substack{L_1,L_2 \cup \set{(x,y)} \\ R_1,R_2}\right)
\implies (\bfk,\vec{z}_L,\vec{z}_R) \in \Sspru\left(\substack{L_1,L_2 \\ R_1,R_2}\right).
\]
This proves $\Psi_{\phantom{x,}L_1,R_1,L_2,R_2} \supseteq \Phi_{x,L_1,R_1,L_2,R_2}$.

Next, we bound the size of their difference. We denote by $\mathsf{BAD}$ the set of 
$(z,\bfk,\bfz) \in [N] \times \Sspru\!\left(\substack{L_1,L_2 \\ R_1,R_2}\right)$ for which no $y \notin \Im(L_2)$ satisfies $(\bfk,\vec{z}_L^{\,(i \gets z)},\vec{z}_R) \in \Sspru\left(\substack{L_1,L_2 \cup \set{(x,y)} \\ R_1,R_2}\right)$. By~\Cref{lem:count_x_z}, we have 
\[
|\mathsf{BAD}| = (22t^2+10t+1)N^{|L_2|+|R_2|+3}.
\]
Now, if we enumerate $(y,z,\bfk,\bfz)$ in $\Psi_{\phantom{x,}L_1,R_1,L_2,R_2}$, then one of the following holds:
\begin{enumerate}
    \item If $(z,\bfk,\bfz) \in \mathsf{BAD}$, then every $y \notin \Im(L_2)$ satisfies 
    $(\bfk,\bfz) \notin \Sspru\left(\substack{L_1,L_2 \cup \set{(x,y)} \\ R_1,R_2}\right)$, which in turn implies $(y,\bfk,\bfz) \notin \Phi_{x,L_1,R_1,L_2,R_2}$.
    \item If $(z,\bfk,\bfz) \notin \mathsf{BAD}$, then by~\Cref{lem:comb_F2Ldagger_y_difference} there are at most $4t+1$ values of $y \notin \Im(L_2)$ for which $(\bfk,\bfz) \notin \Sspru\left(\substack{L_1,L_2 \cup \set{(x,y)} \\ R_1,R_2}\right)$. All other $y$ satisfy $(y,\bfk,\bfz) \in \Phi_{x,L_1,R_1,L_2,R_2}$.
\end{enumerate}
Therefore, the size of their difference is at most
\begin{align*}
    |\mathsf{BAD}| \cdot N + N^{|L_2|+|R_2|+4} \cdot (4t+1)
    \le (22t^2+14t+2)N^{|L_2|+|R_2|+4}.
    \tag*{\qedhere}
\end{align*}
\end{proof}

\subsection{Defining the Approximate Isometry $\Ospru$}

Before defining $\Ospru$, which intuitively maps a view in hybrid $\hyb_3$ to a uniform superposition of consistent views in hybrid $\hyb_4$, we first define the following operators, which mimic each step of enumerating good tuples in the previous subsection. 

\begin{definition}[Operator $\cS_{k_1,k_3}$]    \label{def:S_k1_k3}
Define the operator $\cS_{k_1,k_3}$ such that for any $L_1, L_2 \in \Rinj, R_1, R_2 \in \RDdist$,
\begin{align}
\cS_{k_1,k_3} \colon \ket{L_1}_{\reg{S}_1} \ket{R_1}_{\reg{T}_1} \ket{L_2}_{\reg{S}_2} \ket{R_2}_{\reg{T}_2} 
\mapsto \frac{1}{\sqrt{N^2}} \sum_{ \substack{ 
    k_1 \notin \cB_1(L_1, L_2, R_1, R_2) \\ 
    k_3 \notin \cB_3(L_1, L_2, R_1, R_2)
} } \ket{L_1}_{\reg{S}_1} \ket{R_1}_{\reg{T}_1} \ket{L_2}_{\reg{S}_2} \ket{R_2}_{\reg{T}_2} \ket{k_1}_{\reg{K}_1} \ket{k_3}_{\reg{K}_3}.
\end{align}
Otherwise, $\cS_{k_1,k_3}$ maps all the other basis vectors to the zero vector.
\end{definition}

\begin{definition}[Operator $\cS_{k_2}$]    \label{def:S_k2}
Define the operator $\cS_{k_2}$ such that for any $L_1, L_2 \in \Rinj, R_1, R_2 \in \RDdist, k_1 \notin \cB_1(L_1, L_2, R_1, R_2), k_3 \notin \cB_3(L_1, L_2, R_1, R_2)$,
\begin{multline}
\cS_{k_2} \colon \ket{L_1}_{\reg{S}_1} \ket{R_1}_{\reg{T}_1} \ket{L_2}_{\reg{S}_2} \ket{R_2}_{\reg{T}_2} \ket{k_1}_{\reg{K}_1} \ket{k_3}_{\reg{K}_3} \mapsto \\
\frac{1}{\sqrt{N}} 
\sum_{k_2 \notin \cB_2(L_1, L_2, R_1, R_2, k_1, k_3)} 
\ket{L_1}_{\reg{S}_1} \ket{R_1}_{\reg{T}_1} \ket{L_2}_{\reg{S}_2} \ket{R_2}_{\reg{T}_2}
\ket{k_1}_{\reg{K}_1} \ket{k_2}_{\reg{K}_2} \ket{k_3}_{\reg{K}_3}.
\end{multline}
Otherwise, $\cS_{k_2}$ maps all the other basis vectors to the zero vector.
\end{definition}

\begin{definition}[Operator $\cS_{\vec{z}}$]    \label{def:S_vecz}
Define the operator $\cS_{\vec{z}}$ such that for any $L_1, L_2 \in \Rinj, R_1, R_2 \in \RDdist, k_1 \notin \cB_1(L_1, L_2, R_1, R_2), k_3 \notin \cB_3(L_1, L_2, R_1, R_2), k_2 \notin \cB_2(L_1, L_2, R_1, R_2, k_1, k_3)$,
\begin{multline}
\cS_{\vec{z}} \colon \ket{L_1}_{\reg{S}_1} \ket{R_1}_{\reg{T}_1} \ket{L_2}_{\reg{S}_2} \ket{R_2}_{\reg{T}_2} \ket{k_1}_{\reg{K}_1} \ket{k_2}_{\reg{K}_2} \ket{k_3}_{\reg{K}_3} \mapsto \\
\frac{1}{\sqrt{N^{|L_2|+|R_2|}}} 
\sum_{ \substack{ 
\vec{z}_L \notin \cB_L\left(\substack{L_1,L_2,\\ R_1,R_2,\\ k_1,k_2,k_3}\right),
\vec{z}_R \notin \cB_R\left(\substack{L_1,L_2,\\ R_1,R_2,\\ k_1,k_2,k_3}\right)
} }
\ket{L_1}_{\reg{S}_1} \ket{R_1}_{\reg{T}_1} \ket{L_2}_{\reg{S}_2} \ket{R_2}_{\reg{T}_2} 
\ket{\vec{z}_L}_{\reg{Z_L}} \ket{\vec{z}_R}_{\reg{Z_R}} 
\ket{k_1}_{\reg{K}_1} \ket{k_2}_{\reg{K}_2} \ket{k_3}_{\reg{K}_3}.
\end{multline}
Otherwise, $\cS_{\vec{z}}$ maps all the other basis vectors to the zero vector.
\noindent 
\end{definition}

Now, we define the operator $\Ospru$.
\begin{definition}[Operator $\Ospru$] \label{def:Ospru}
Define the operator 
\begin{align}
\Ospru \coloneqq \Osprus^\dagger \cdot \cS_{\vec{z}} \cdot \cS_{k_2} \cdot \cS_{k_1,k_3}.
\end{align}
\end{definition}

\noindent Note that $\cS_{k_1,k_3}, \cS_{k_2}, \cS_{\vec{z}}$ are contractions, that is, their operator norms are all bounded $1$. Thus, $\Ospru$ is \emph{not} a partial isometry. Importantly, the action of $\Ospru$ satisfies the following.
\begin{lemma}   \label{lem:Ospru}
For any $L_1, L_2 \in \Rinj, R_1, R_2 \in \RDdist$,
\begin{multline*}
\Ospru \colon \ket{L_1}_{\reg{S}_1} \ket{R_1}_{\reg{T}_1} \ket{L_2}_{\reg{S}_2} \ket{R_2}_{\reg{T}_2} \mapsto \\
\frac{1}{\sqrt{N^{|L_2|+|R_2|+3}}}
\sum_{(k_1,k_2,k_3,\vec{z}_L,\vec{z}_R) \in \Sspru\left(\substack{L_1,L_2 \\ R_1,R_2}\right)}
\ket{L^{(k_1,k_3)}_1\cup L_2^{(k_2,\vec{z}_L)}}_{\reg{S}} 
\ket{R^{(k_1,k_3)}_1\cup R_2^{(k_2,\vec{z}_R)}}_{\reg{T}}
\ket{k_1}_{\reg{K_1}}
\ket{k_2}_{\reg{K_2}}
\ket{k_3}_{\reg{K_3}}.
\end{multline*}
\end{lemma}

\begin{proof}
Recall~\Cref{def:S_k1_k3,def:S_k2,def:S_vecz}, we have
\begin{align*}
&\ket{L_1}_{\reg{S}_1} \ket{R_1}_{\reg{T}_1} \ket{L_2}_{\reg{S}_2} \ket{R_2}_{\reg{T}_2}
\xmapsto{\cS_{\vec{z}} \cdot \cS_{k_2} \cdot \cS_{k_1,k_3}} \\
&\hspace{.2\textwidth} \frac{1}{\sqrt{N^{|L_2|+|R_2|+3}}}
\sum_{(\bfk,\bfz) \in \Sspru\left(\substack{L_1,L_2 \\ R_1,R_2}\right)}
\ket{L_1}_{\reg{S}_1} \ket{R_1}_{\reg{T}_1} \ket{L_2}_{\reg{S}_2} \ket{R_2}_{\reg{T}_2} 
\ket{\vec{z}_L}_{\reg{Z_L}} \ket{\vec{z}_R}_{\reg{Z_R}} 
\ket{k_1}_{\reg{K}_1} \ket{k_2}_{\reg{K}_2} \ket{k_3}_{\reg{K}_3}.
\end{align*}
From~\Cref{lem:good_tuples_satisfy_conditions}, every $(\bfk,\bfz) \in \Sspru\left(\substack{L_1,L_2 \\ R_1,R_2}\right)$ is decodable \wrt $(L_1,R_1,L_2,R_2)$. Thus, $(L_1,R_1,L_2,R_2,\bfk,\bfz)$ is in $\Supp(\Dec)$. Therefore, applying $\Osprus^{\dagger}$ is equivalently to applying $\Enc$, which the inverse of $\Dec$ defined in the proof of~\Cref{lem:D_is_partial_isometry}. This completes the proof.
\end{proof}

 The following lemma will be used in~\Cref{sec:spru:1st_oracle,sec:spru:2nd_oracle}.
\begin{lemma} \label{lem:punc_S}
Let $\set{\cP_\tau}_\tau$ be a collection of sets where the index $\tau$ ranges over $(y \in [N], L_1 \in \RIdist, R_1 \in \RDdist, L_2 \in \RIdist, R_2 \in \RDdist)$ and $\cP_\tau \subseteq [N]^3 \times [N]^{|L_2|} \times [N]^{|R_2|}$. Define the operator 
\begin{align*}
\Ospru^{\bullet}
& \colon \ket{y}_{\reg{A}} \ket{L_1}_{\reg{S}_1} \ket{R_1}_{\reg{T}_1} \ket{L_2}_{\reg{S}_2} \ket{R_2}_{\reg{T}_2} \\
& \mapsto 
\ket{y}_{\reg{A}} \;
\frac{1}{\sqrt{N^{|L_2|+|R_2|+3}}}
\sum_{   \substack{
    (\bfk,\bfz) \in \Sspru\left(\substack{L_1,L_2 \\ R_1,R_2}\right) \colon \\
    (\bfk,\bfz) \notin \cP_{y,L_1,R_1,L_2,R_2}
}   }
\ket{L^{(k_1,k_3)}_1 \cup L_2^{(k_2,\vec{z}_L)}}_{\reg{S}} 
\ket{R^{(k_1,k_3)}_1 \cup R_2^{(k_2,\vec{z}_R)}}_{\reg{T}}
\ket{k}_{\reg{K}}.
\end{align*}
If there exists $\delta \ge 0$ such that for any $\tau$,
\begin{align*}
    \frac{\big|\cP_{y,L_1,R_1,L_2,R_2} \cap \Sspru\left(\substack{L_1,L_2 \\ R_1,R_2}\right)\big|}{N^{|L_2|+|R_2|+3}} \le \delta,
\end{align*}
then
\begin{align*}
    \|\Ospru^{\bullet} - \Ospru\|_{\op} = \sqrt{\delta}.
\end{align*}
\end{lemma}
\noindent Namely, $\Ospru^{\bullet}$ is further controlled by the register $\reg{A}$ and 
imposes extra conditions ensuring that good tuples $(\bfk,\bfz)$ do not lie in some ``bad set'' $\cP_\tau$. We provide the proof in~\Cref{sec:lemmas_F}.

\subsection{Main Lemmas} \label{sec:main_lemmas}
\noindent We introduce important lemmas regarding $\Ospru$ below. Their proofs are deferred to further subsections. We will first use these lemmas to prove~\Cref{lem:spru:induction} in~\Cref{sec:spru:induction}.

\begin{lemma}[$\Ospru$ is Close to a Partial Isometry]  \label{lem:spru:iso}
There exists a partial isometry $\wt{\Ospru}$ such that for any integer $t \geq 0$,
\[
\|(\wt{\Ospru} - \Ospru) \Pi_{\leq t}\|_{\op}
\leq O( \sqrt{t^2/N} ).
\]
\end{lemma}

\noindent We provide the proof of above lemma in~\Cref{sec:spru:iso}.

\begin{lemma}[Closeness of the First Oracle]    \label{lem:1st_oracle}
For any integer $t \geq 0$,
\begin{itemize}
    \item {\bf Forward query:} $\|(X^{k_3} F X^{k_1} \Ospru - \Ospru F_1) \Pi_{\leq t}\|_{\op} \leq O ( \sqrt{t/N} )$,
    \item {\bf Inverse query:} $\|(X^{k_1} F^{\dagger} X^{k_3} \Ospru - \Ospru F_1^{\dagger}) \Pi_{\leq t}\|_{\op} \leq O ( \sqrt{t/N} )$.
\end{itemize}
\end{lemma}

\noindent We provide the proof of above lemma in~\Cref{sec:spru:1st_oracle}.

\begin{lemma}[Closeness of the Second Oracle]   \label{lem:spru:2nd_oracle}
For any integer $t \geq 0$,
\begin{itemize}
    \item {\bf Forward query:} $\|(F X^{k_2} F \Ospru - \Ospru F_2) \Pi_{\leq t}\|_{\op} \leq O(t/\sqrt{N})$,
    \item {\bf Inverse query:} $\|(F^{\dagger} X^{k_2} F^{\dagger} \Ospru - \Ospru F_2^{\dagger}) \Pi_{\leq t}\|_{\op} \leq O(t/\sqrt{N})$.
\end{itemize}
\end{lemma}

\noindent We provide the proof of above lemma in~\Cref{sec:spru:2nd_oracle}.

\subsection{Statistical Closeness between $\hyb_3$ and $\hyb_4$:~Proving~\Cref{lem:spru:induction}}
\label{sec:spru:induction}

Now, we use the lemmas in~\Cref{sec:main_lemmas} to prove~\Cref{lem:spru:induction}. The structure of the proof is similar to the commutator-style analysis in~\cite{CMS19,DFMS22}.

\begin{proof}[Proof of~\Cref{lem:spru:induction}]
We first introduce some notations. Let $\ket{\psi_0}$ denote the initial state in hybrid $\hyb_3$, \ie
\[
\ket{\psi_0} \coloneqq \ket{0}_{\reg{A}}\ket{0}_{\reg{B}}\ket{\varnothing}_{\reg{S_1}}\ket{\varnothing}_{\reg{T_1}}\ket{\varnothing}_{\reg{S_2}}\ket{\varnothing}_{\reg{T_2}}.
\]
For $i \in [4t]$, let $\ket{\psi_i}$ denote the state right after the $i$-th query,
\[
\ket{\psi_i} \coloneqq \cO_i A_i \ket{\psi_{i-1}},
\]
where $\cO_i$ cycles through $F_1, F_2, F_1^{\dagger}, F_2^{\dagger}$ according to $i \bmod 4$. Similarly, denote the initial state in $\hyb_4$ by
\[
\ket{\phi_0} \coloneqq 
\frac{1}{\sqrt{N^3}} \sum_{k_1,k_2,k_3 \in [N]} \ket{0}_{\reg{A}} \ket{0}_{\reg{B}} \ket{\varnothing}_{\reg{S}} \ket{\varnothing}_{\reg{T}} \ket{k_1}_{\reg{K_1}} \ket{k_2}_{\reg{K_2}} \ket{k_3}_{\reg{K_3}}.
\]
For $i \in [4t]$, let $\ket{\psi_i}$ denote the state right after the $i$-th query,
\[
\ket{\phi_i} \coloneqq \cO_i A_i \ket{\phi_{i-1}},
\]
where $\cO_i$ cycles through $X^{k_3} F X^{k_1}, F X^{k_2} F, X^{k_1} F^{\dagger} X^{k_3}, F^{\dagger} X^{k_2} F^{\dagger}$ according to $i \bmod 4$. \\

\noindent Now, we prove that $\|\Ospru\ket{\psi_i}-\ket{\phi_i}\|_2 = O(i^2/\sqrt{N})$ for $i \in [4t]$ by induction. \\

\noindent \textbf{Base case ($i=0$):} $\Ospru \ket{\psi_0} = \ket{\phi_0}$ holds trivially. \\

\noindent \textbf{Induction step:} Suppose $\|\Ospru\ket{\psi_{i-1}}-\ket{\phi_{i-1}}\|_2 = O((i-1)^2/\sqrt{N})$. Consider the following four cases: \\

\noindent {\bf Case 1. $i \equiv 1 \bmod{4}$}:
\begin{align*}
& \|\Ospru\ket{\psi_{i}}-\ket{\phi_{i}}\|_2 \\
= & \|\Ospru F_{1} A_i \ket{\psi_{i-1}} - X^{k_3} F X^{k_1} A_i \ket{\phi_{i-1}}\|_2
\tag{by expanding the definition of $\ket{\psi_{i}}$ and $\ket{\phi_{i}}$} \\
\leq & \|\Ospru F_1 A_i \ket{\psi_{i-1}} \rcolor{- X^{k_3} F X^{k_1} A_i \Ospru \ket{\psi_{i-1}}} \|_2 \\
& \hspace{.15\textwidth} + \| \rcolor{X^{k_3} F X^{k_1} A_i \Ospru \ket{\psi_{i-1}}} - X^{k_3} F X^{k_1} A_i \ket{\phi_{i-1}}\|_2 
\tag{by the triangle inequality} \\
= & \|(\Ospru F_i - X^{k_3} F X^{k_1} \Ospru) A_i \ket{\psi_{i-1}}\|_2
+ \|X^{k_3} F X^{k_1} A_i (\Ospru \ket{\psi_{i-1}} - \ket{\phi_{i-1}})\|_2 
\tag{since $\Ospru$ and $A_i$ commute} \\
\leq & \|(\Ospru F_1 - X^{k_3} F X^{k_1} \Ospru) \Pi_{\leq t}\|_{\op}
+ \norm{\Ospru \ket{\psi_{i-1}} - \ket{\phi_{i-1}}}_2 
\tag{by~\Cref{lem:op_norm}} \\
= & O(i^2/\sqrt{N}).
\tag{by~\Cref{lem:1st_oracle} and the induction hypothesis}
\end{align*}

\noindent Other three cases follow from the same argument. Hence, the induction holds true. In particular, when $i = 4t$, we have
\begin{equation} \label{eq:induction}
\|\Ospru\ket{\psi_{4t}} -\ket{\phi_{4t}}\|_2 = O(t^2/\sqrt{N}).
\end{equation}
Let $\wt{\Ospru}$ be the partial isometry guaranteed to exist in~\Cref{lem:spru:iso}. By the triangle inequality, \Cref{lem:spru:iso}, and~\Cref{eq:induction}, we have
\begin{align}   \label{eq:psi_phi_4t}
\|\wt{\Ospru} \ket{\psi_{4t}} - \ket{\phi_{4t}}\|_2
\leq \|\wt{\Ospru} \ket{\psi_{4t}} - \Ospru \ket{\psi_{4t}}\|_2 + \|\Ospru \ket{\psi_{4t}}-\ket{\phi_{4t}}\|_2
= O \left(t^2/\sqrt{N}\right).
\end{align}

\noindent Finally, the trace distance between the output of~$\hyb_3$ and that of~$\hyb_4$ satisfies
\begin{align}
& \TD(\rho_3, \rho_4) \\
& = \TD(\Tr_{\reg{S_1S_2T_1T_2}}(\ketbra{\psi_{4t}}{\psi_{4t}}),\Tr_{\reg{STK_1K_2K_3}}(\ketbra{\phi_{4t}}{\phi_{4t}})) 
\notag \\
& = \TD(\Tr_{\reg{STK_1K_2K_3}}(\wt{\Ospru} \ketbra{\psi_{4t}}{\psi_{4t}}\wt{\Ospru}^{\dagger}),\Tr_{\reg{STK_1K_2K_3}}(\ketbra{\phi_{4t}}{\phi_{4t}}))  
\label{eq:isometry_partial_trace} \\
& \leq \TD(\wt{\Ospru} \ketbra{\psi_{4t}}{\psi_{4t}} \wt{\Ospru}^{\dagger}, \ketbra{\phi_{4t}}{\phi_{4t}})
\tag{trace distance is non-increasing under partial trace} \\
& \leq \|\wt{\Ospru} \ket{\psi_{4t}} - \ket{\phi_{4t}}\|_2 
\tag{the trace distance between pure states is bounded by their Euclidean distance} \\
& = O(t^2/\sqrt{N}),
\tag{by~\Cref{eq:psi_phi_4t}}
\end{align}
where~\Cref{eq:isometry_partial_trace} is because $\wt{\Ospru}$ is a partial isometry that acts on the registers being traced out, and $\ket{\psi_{4t}}$ is in the domain of $\Ospru'$. This completes the proof of~\Cref{lem:spru:induction}.
\end{proof}

\section{$\Ospru$ is Close to a Partial Isometry: Proving~\Cref{lem:spru:iso}} 
\label{sec:spru:iso}
We will define the ``normalized'' version of $\Ospru_{k_1,k_3}, \Ospru_{k_2}, \Ospru_{\vec{z}}$ such that the coefficients match the number of terms in the sum. First, define the partial isometry $\wt{\Ospru}_{k_1,k_3}$ such that for any $L_1, L_2 \in \Rinj, R_1, R_2 \in \RDdist$,
\begin{align*}
\wt{\Ospru}_{k_1,k_3} \colon
&\ket{L_1}_{\reg S_1}\ket{R_1}_{\reg T_1}\ket{L_2}_{\reg S_2}\ket{R_2}_{\reg T_2}
\\
&\mapsto
\frac{1}{\sqrt{(N-\abs{\cB_1(L_1,L_2,R_1,R_2)})(\,N-\abs{\cB_3(L_1,L_2,R_1,R_2)}\,)}}\\
&\hspace{.2\textwidth} \times \sum_{\substack{
    k_1 \notin \cB_1(L_1,L_2,R_1,R_2)\\
    k_3 \notin \cB_3(L_1,L_2,R_1,R_2)
}}
\ket{L_1}_{\reg S_1}\ket{R_1}_{\reg T_1}\ket{L_2}_{\reg S_2}\ket{R_2}_{\reg T_2}
\ket{k_1}_{\reg K_1}\ket{k_3}_{\reg K_3}\,.
\end{align*}

\noindent Next, define the partial isometry $\wt{\Ospru}_{k_2}$ such that for any $L_1, L_2 \in \Rinj, R_1, R_2 \in \RDdist, k_1 \notin \cB_1(L_1, L_2,$ $R_1, R_2), k_3 \notin \cB_3(L_1, L_2, R_1, R_2)$,
\begin{align*}
\wt{\Ospru}_{k_2} \colon
&\ket{L_1}_{\reg{S}_1} \ket{R_1}_{\reg{T}_1} \ket{L_2}_{\reg{S}_2} \ket{R_2}_{\reg{T}_2} \ket{k_1}_{\reg{K}_1} \ket{k_3}_{\reg{K}_3}\\
&\mapsto
\frac{1}{\sqrt{N - \abs{\cB_2(L_1, L_2, R_1, R_2, k_1, k_3)}}}\\
&\hspace{.15\textwidth} \times \sum_{k_2 \notin \cB_2(L_1, L_2, R_1, R_2, k_1, k_3)}
\ket{L_1}_{\reg{S}_1}\ket{R_1}_{\reg{T}_1}\ket{L_2}_{\reg{S}_2}\ket{R_2}_{\reg{T}_2}\ket{k_1}_{\reg{K}_1}\ket{k_2}_{\reg{K}_2}\ket{k_3}_{\reg{K}_3}\,.
\end{align*}

\noindent Finally, define the partial isometry $\wt{\Ospru}_{\vec{z}}$ such that for any $L_1, L_2 \in \Rinj, R_1, R_2 \in \RDdist, k_1 \notin \cB_1(L_1, L_2,$ $R_1, R_2), k_3 \notin \cB_3(L_1, L_2, R_1, R_2), k_2 \notin \cB_2(L_1, L_2, R_1, R_2, k_1, k_3)$,
\begin{align*}
\wt{\Ospru}_{\vec{z}} \colon
&\ket{L_1}_{\reg{S}_1} \ket{R_1}_{\reg{T}_1} \ket{L_2}_{\reg{S}_2} \ket{R_2}_{\reg{T}_2} \ket{k_1}_{\reg{K}_1} \ket{k_2}_{\reg{K}_2} \ket{k_3}_{\reg{K}_3}\\
&\mapsto
\frac{1}{\sqrt{N^{|L_2|} - \abs{\cB_L \left(\substack{L_1,L_2,\\ R_1,R_2,\\ k_1,k_2,k_3} \right)}}} 
\frac{1}{\sqrt{N^{|R_2|} - \abs{\cB_R \left(\substack{L_1,L_2,\\ R_1,R_2,\\ k_1,k_2,k_3} \right)}}} \\
& \hspace{.1\textwidth} \times \sum_{  \substack{
\vec{z}_L \notin \cB_L\left(\substack{L_1,L_2,\\ R_1,R_2,\\ k_1,k_2,k_3}\right), 
\vec{z}_R \notin \cB_R\left(\substack{L_1,L_2,\\ R_1,R_2,\\ k_1,k_2,k_3}\right)
}   }
\ket{L_1}_{\reg{S}_1} \ket{R_1}_{\reg{T}_1} \ket{L_2}_{\reg{S}_2} \ket{R_2}_{\reg{T}_2} \ket{\vec{z}_L}_{\reg{Z_L}} \ket{\vec{z}_R}_{\reg{Z_R}} \ket{k_1}_{\reg{K}_1} \ket{k_2}_{\reg{K}_2} \ket{k_3}_{\reg{K}_3}.
\end{align*}

\begin{lemma} \label{lem:isometries_close}
For any interger $t \ge 0$,
\begin{align*}
& \|(\Ospru_{k_1,k_3} - \wt{\Ospru}_{k_1,k_3}) \Pi_{\leq t}\|_{\op} \leq O(\sqrt{t^2/N}), \\
& \|(\Ospru_{k_2} - \wt{\Ospru}_{k_2}) \Pi_{\leq t}\|_{\op} \leq O(\sqrt{t^2/N}), \\
& \|(\Ospru_{\vec{z}} - \wt{\Ospru}_{\vec{z}}) \Pi_{\leq t}\|_{\op} \leq O(\sqrt{t^2/N}).
\end{align*}
\end{lemma}

\begin{proof}
Since $\Ospru_{k_1,k_3} - \wt{\Ospru}_{k_1,k_3}$ preserve the orthogonality of input of the form $\ket{L_1} \ket{R_1} \ket{L_2}\ket{R_2}$. Thus, by~\Cref{lem:op_norm_orthogonal}, it is suffices to maximize 
\begin{align*}
& \|(\Ospru_{k_1,k_3} - \wt{\Ospru}_{k_1,k_3}) \Pi_{\leq t} \ket{L_1} \ket{R_1} \ket{L_2}\ket{R_2}\|_2.
\end{align*}
This follow from~\Cref{lem:number_of_good_tuples} and an elementary calculation. Items~2 and~3 follow similarly.
\end{proof}

\begin{proof}[Proof of~\Cref{lem:spru:iso}]
Define the operator $\wt{\Ospru} \coloneqq {\Osprus}^\dagger \cdot \wt{\Ospru}_{\vec{z}} \cdot \wt{\Ospru}_{k_2} \cdot \wt{\Ospru}_{k_1,k_3}$. To see that $\wt{\Ospru}$ is a partial isometry, one can easily verify that $\wt{\Ospru}$ preserves the inner product between basis vectors in the domain.\footnote{In contrast to isometries, partial isometries are \emph{not} necessarily closed under composition.} Then we obtain
\begin{align}
& \|(\Ospru - \wt{\Ospru}) \Pi_{\leq t}\|_{\op}
\notag \\
& = \|{\Osprus}^\dagger \cdot (\Ospru_{\vec{z}} \cdot \Ospru_{k_2} \cdot \Ospru_{k_1,k_3} - \wt{\Ospru}_{\vec{z}} \cdot \wt{\Ospru}_{k_2} \cdot \wt{\Ospru}_{k_1,k_3}) \Pi_{\leq t}\|_{\op} 
\notag \\
& = \|(\Ospru_{\vec{z}} \cdot \Ospru_{k_2} \cdot \Ospru_{k_1,k_3} - \wt{\Ospru}_{\vec{z}} \cdot \wt{\Ospru}_{k_2} \cdot \wt{\Ospru}_{k_1,k_3}) \Pi_{\leq t}\|_{\op} 
\tag{since $\|{\Osprus}^\dagger\|_{\op} = \|\Osprus\|_{\op} = 1$} \\
& \leq \|(\Ospru_{\vec{z}} - \wt{\Ospru}_{\vec{z}}) \Pi_{\leq t}\|_{\op}
+ \|(\Ospru_{k_2} - \wt{\Ospru}_{k_2})  \Pi_{\leq t}\|_{\op}
+ \|(\Ospru_{k_1,k_3} - \wt{\Ospru}_{k_1,k_3})  \Pi_{\leq t}\|_{\op} 
\label{eq:triangle_ineq} \\
& = O (\sqrt{t^2/N}),
\tag{by~\Cref{lem:isometries_close}}
\end{align}
where \Cref{eq:triangle_ineq} uses (i) the triangle inequality; (ii) that $\Pi_{\le t}$ commutes with each of $\Ospru_{\vec z}, \Ospru_{k_2}, \Ospru_{k_1,k_3}$; (iii) that the operator norm is submultiplicative; and (iv) that $\Ospru_{\vec z}, \Ospru_{k_2}, \Ospru_{k_1,k_3}$ are partial isometries (so their operator norm is $1$). This completes the proof of~\Cref{lem:spru:iso}.
\end{proof}

\section{Closeness of the First Oracle: Proving~\Cref{lem:1st_oracle}}
\label{sec:spru:1st_oracle}

Our approach is through expanding oracles $F_1,F_2,F$ as a sum of smaller terms. Then we carefully bound each pair of terms. Before proving~\Cref{lem:1st_oracle}, we first introduce several lemmas. 

\subsection{Closeness of $F_1^L$ and $F_1^R$}
Intuitively, the following lemma states the following. Suppose the adversary in hybrid $\hyb_3$ has made $t = \poly(\secp)$ queries in total at the moment. Then the state obtained by applying $F_1^L$ followed by $\Ospru$ is negligibly close to the state obtained by applying $\Ospru$ followed by $X^{k_3} F^L X^{k_1}$. The intuition is straightforward. If we apply $F_1^L$ and then $\Ospru$, the resulting state is entirely supported by decomposable relations with the correct number of correlated pairs, owing to the definition of $\Ospru$. On the other hand, if we first apply $\Ospru$, then $X^{k_3} F^L X^{k_1}$, there is a small chance that the $y$ sampled by $F^L$ might generate unwanted correlated pairs. Fortunately, since the relations are of polynomial size, all but a negligible fraction of $y$ behave correctly.

\begin{lemma}[Closeness of $F_1^L$ and $F_1^R$] \label{lem:fl1_fr1}
For any integer $t \geq 0$,
\begin{align*}
& \|( X^{k_3} F^{L} X^{k_1} \Ospru - \Ospru F_1^{L} ) \Pi_{\le t} \|_{\op}
= O( \sqrt{t/N} ) \\
& \|( X^{k_1} F^R X^{k_3} \Ospru - \Ospru F_1^R ) \Pi_{\leq t }\|_{\op} 
= O( \sqrt{t/N} ).
\end{align*}
\end{lemma}

\begin{proof}
Fix $t \in \N, x \in [N], L_1, L_2 \in \Rinj_{\leq t}$, and $ R_1, R_2 \in \RDdist_{\leq t}$. We start by calculating the following states:
\begin{align*}  
& \ket{\psi_{x,L_1,R_1,L_2,R_2}}_{\reg{AST} \reg{K}_1 \reg{K}_2 \reg{K}_3} \coloneqq 
X^{k_3} F^L X^{k_1} \Ospru \ket{x}_{\reg{A}} \ket{L_1}_{\reg{S}_1} \ket{R_1}_{\reg{T}_1} \ket{L_2}_{\reg{S}_2} \ket{R_2}_{\reg{T}_2}, \\
& \ket{\phi_{x,L_1,R_1,L_2,R_2}}_{\reg{AST} \reg{K}_1\reg{K}_2\reg{K}_3} \coloneqq 
\Ospru F_1^L \ket{x}_{\reg{A}} \ket{L_1}_{\reg{S}_1} \ket{R_1}_{\reg{T}_1} \ket{L_2}_{\reg{S}_2} \ket{R_2}_{\reg{T}_2}.
\end{align*}
To simplify notation, we write $\ket{\bfk}_{\reg{K}}$ as shorthand for $\ket{k_1}_{\reg{K}_1} \ket{k_2}_{\reg{K}_2} \ket{k_3}_{\reg{K}_3}$, and $(\bfk, \bfz) \in \Sspru\left(\substack{L_1,L_2, \\ R_1,R_2}\right)$ as shorthand for $(k_1, k_2, k_3, \vec{z}_L, \vec{z}_R) \in \Sspru\left(\substack{L_1,L_2, \\ R_1,R_2}\right)$.

\myparagraph{Computing $\ket{\psi_{x,L_1,R_1,L_2,R_2}}$} Expanding the definitions of $\Ospru$ and $F^L$, we have
\begin{align}
& \ket{x}_{\reg{A}} \ket{L_1}_{\reg{S}_1} \ket{R_1}_{\reg{T}_1} \ket{L_2}_{\reg{S}_2} \ket{R_2}_{\reg{T}_2}
\notag \\
& \xmapsto{\Ospru} \frac{1}{\sqrt{N^{|L_2|+|R_2|+3}}} 
\sum_{ \substack{
    (\bfk, \bfz) \in \Sspru\left(\substack{L_1,L_2, \\ R_1,R_2}\right)
} }
\ket{x}_{\reg{A}}
\ket{L_1^{(k_1,k_3)} \cup L_2^{(k_2,\vec{z}_L)}}_{\reg{S}} 
\ket{R^{(k_1,k_3)}_1 \cup R_2^{(k_2,\vec{z}_R)}}_{\reg{T}}
\ket{\bfk}_{\reg{K}}
\tag{by~\Cref{lem:Ospru}} \\
& \xmapsto{X^{k_3} F^L X^{k_1}}
\frac{1}{\sqrt{N^{|L_2|+|R_2|+4}}} 
\sum_{ \substack{
    (\bfk, \bfz) \in \Sspru\left(\substack{L_1,L_2, \\ R_1,R_2}\right) \\
    y \notin \Im(L_1^{(k_1,k_3)} \cup L_2^{(k_2,\vec{z}_L)})
} }
\ket{y \oplus k_3}_{\reg{A}}
\ket{L_1^{(k_1,k_3)} \cup L_2^{(k_2,\vec{z}_L)} \cup \set{(x \oplus k_1, y)}}_{\reg{S}}
\ket{R^{(k_1,k_3)}_1 \cup R_2^{(k_2,\vec{z}_R)}}_{\reg{T}}
\ket{\bfk}_{\reg{K}}
\tag{by~\Cref{eq:def:F_L}} \\
& = \frac{1}{\sqrt{N^{|L_2|+|R_2|+4}}}
\sum_{ \substack{ 
    (\bfk, \bfz) \in \Sspru\left(\substack{L_1,L_2, \\ R_1,R_2}\right) \\
    y \oplus k_3 \notin \Im(L_1^{(k_1,k_3)} \cup L_2^{(k_2,\vec{z}_L)}) 
} }
\ket{y}_{\reg{A}}
\ket{L_1^{(k_1,k_3)} \cup L_2^{(k_2,\vec{z}_L)} \cup \set{(x \oplus k_1, y \oplus k_3)}}_{\reg{S}}
\ket{R^{(k_1,k_3)}_1 \cup R_2^{(k_2,\vec{z}_R)}}_{\reg{T}} 
\ket{\bfk}_{\reg{K}}
\tag{by relabeling $y \mapsto y \oplus k_3$} \\
& = \frac{1}{\sqrt{N^{|L_2|+|R_2|+4}}}
\sum_{ \substack{ 
    (\bfk, \bfz) \in \Sspru\left(\substack{L_1,L_2, \\ R_1,R_2}\right) \\
    y \notin \Im(L_1)  \cup \qty( \Im(L_2^{(k_2,\vec{z}_L)}) \oplus k_3 )
} }
\ket{y}_{\reg{A}} 
\ket{\left(L_1 \cup \set{(x,y)}\right)^{(k_1,k_3)} \cup L_2^{(k_2,\vec{z}_L)}}_{\reg{S}}
\ket{R^{(k_1,k_3)}_1 \cup R_2^{(k_2,\vec{z}_R)}}_{\reg{T}} 
\ket{\bfk}_{\reg{K}}, \label{eq:psi}
\end{align}
where the last line is by the definition of augmented relations in~\Cref{eq:aug_L1}.

\myparagraph{Computing $\ket{\phi_{x,L_1,R_1,L_2,R_2}}$}
Similarly, expanding the definitions of $\Ospru$ and $F_1^L$, we have
\begin{align}
& \frac{1}{\sqrt{N^{|L_2|+|R_2|+4}}} 
\sum_{ \substack{ 
    y \notin \Im(L_1), \\
    (\bfk, \bfz) \in \Sspru\left(\substack{L_1 \cup \set{(x,y)},L_2, \\ R_1,R_2}\right)
} } 
\ket{y}_{\reg{A}}
\ket{\left( L_1 \cup \set{(x,y)} \right)^{(k_1,k_3)} \cup L_2^{(k_2,\vec{z}_L)}}_{\reg{S}}
\ket{R^{(k_1,k_3)}_1 \cup R_2^{(k_2,\vec{z}_R)}}_{\reg{T}}
\ket{\bfk}_{\reg{K}}. \label{eq:phi}
\end{align}

\myparagraph{Orthogonality}
Consider distinct $(x,L_1,R_1,L_2,R_2)$ and $(x',L'_1,R'_1,L'_2,R'_2)$. We claim that 
\begin{itemize}
    \item $\ket{\psi_{x,L_1,R_1,L_2,R_2}}$ is orthogonal to $\ket{\psi_{x',L'_1,R'_1,L'_2,R'_2}}$,
    \item $\ket{\phi_{x,L_1,R_1,L_2,R_2}}$ is orthogonal to $\ket{\phi_{x',L'_1,R'_1,L'_2,R'_2}}$,
    \item $\ket{\psi_{x,L_1,R_1,L_2,R_2}}$ is orthogonal to $\ket{\phi_{x',L'_1,R'_1,L'_2,R'_2}}$.
\end{itemize}
They together implies that $\ket{\psi_{x,L_1,R_1,L_2,R_2}} - \ket{\phi_{x,L_1,R_1,L_2,R_2}}$ is orthogonal to $\ket{\psi_{x',L'_1,R'_1,L'_2,R'_2}} - \ket{\phi_{x',L'_1,R'_1,L'_2,R'_2}}$. Thus, by~\Cref{lem:op_norm_orthogonal}, it suffices to maximize the norm over input states of the form $\ket{x} \ket{L_1} \ket{R_1} \ket{L_2} \ket{R_2}$. To prove the claim, we define the operator
\[
\cI \coloneqq X^{k_3} \cdot X^{k_1} \cdot \Osprus \cdot F_{\extract}^L \cdot X^{k_3}
\]
where the partial isometry $F_L^{\extract}$ is defined in~\Cref{eq:FL_extract}. Note that $\cI$ preserves inner product between the states under consideration. To see this, we may compute the states obtained by applying $\cI$ to them:
\begin{align}
& \ket{\psi_{x,L_1,R_1,L_2,R_2}}
\notag \\ 
&  \xmapsto{F_{\extract}^L \cdot X^{k_3}} 
\frac{1}{\sqrt{N^{|L_2|+|R_2|+4}}} 
\sum_{ \substack{ 
    (\bfk, \bfz) \in \Sspru\left(\substack{L_1,L_2, \\ R_1,R_2}\right) \\
    y \notin \Im(L_1)  \cup \qty( \Im(L_2^{(k_2,\vec{z}_L)}) \oplus k_3 )
} }
\ket{y \oplus k_3}_{\reg{A}'} 
\ket{x \oplus k_1}_{\reg{A}} 
\ket{L_1^{(k_1,k_3)} \cup L_2^{(k_2,\vec{z}_L)}}_{\reg{S}}
\ket{R^{(k_1,k_3)}_1 \cup R_2^{(k_2,\vec{z}_R)}}_{\reg{T}} 
\ket{\bfk}_{\reg{K}}
\tag{by~\Cref{eq:psi,eq:FL_extract}} \\
& \xmapsto{X^{k_3} \cdot X^{k_1} \cdot \Osprus} 
\frac{1}{\sqrt{N^{|L_2|+|R_2|+4}}} 
\sum_{ \substack{ 
    (\bfk, \bfz) \in \Sspru\left(\substack{L_1,L_2, \\ R_1,R_2}\right) \\
    y \notin \Im(L_1)  \cup \qty( \Im(L_2^{(k_2,\vec{z}_L)}) \oplus k_3 )
} }
\ket{y}_{\reg{A}'} 
\ket{x}_{\reg{A}} 
\ket{L_1}_{\reg{S}_1} \ket{L_2}_{\reg{S}_2} 
\ket{R_1}_{\reg{T}_1} \ket{R_2}_{\reg{T}_2} 
\ket{\bfz}_{\reg{Z}}
\ket{\bfk}_{\reg{K}}, \label{eq:Ipsi}
\end{align}
where the last line is by~\Cref{def:Ospru_split}. From the above calculation, it is clear that $\cI \ket{\psi_{x,L_1,R_1,L_2,R_2}}$ is orthogonal to $\cI \ket{\psi_{x',L'_1,R'_1,L'_2,R'_2}}$ whenever $(x, L_1, L_2, {R}_1, {R}_2) \neq (x', L'_1, L'_2, {R'}_1, {R'}_2)$. Moreover, $X^{k_3} \ket{\psi_{x,L_1,R_1,L_2,R_2}}$ is in the domain of the partial isometry $F_L^{\extract}$, and $F_L^{\extract} X^{k_3} \ket{\psi_{x,L_1,R_1,L_2,R_2}}$ is in the domain of the partial isometry $\Osprus$. Thus, $\cI$ preserves the inner product between $\ket{\psi_{x,L_1,R_1,L_2,R_2}}$ and $\ket{\psi_{x',L'_1,R'_1,L'_2,R'_2}}$, which implies Item~1, namely, $\ket{\psi_{x,L_1,R_1,L_2,R_2}}$ is orthogonal to $\ket{\psi_{x',L'_1,R'_1,L'_2,R'_2}}$.  \\

\noindent Similarly, we have
\begin{align}
\ket{\phi_{x,L_1,R_1,L_2,R_2}}
\xmapsto{\cI} 
\frac{1}{\sqrt{N^{|L_2|+|R_2|+4}}} 
\sum_{ \substack{ 
    y \notin \Im(L_1), \\
    (\bfk, \bfz) \in \Sspru\left(\substack{L_1 \cup \set{(x,y)},L_2, \\ R_1,R_2}\right)
} } 
\ket{y}_{\reg{A}'} 
\ket{x}_{\reg{A}} 
\ket{L_1}_{\reg{S}_1} \ket{L_2}_{\reg{S}_2} 
\ket{R_1}_{\reg{T}_1} \ket{R_2}_{\reg{T}_2}
\ket{\bfz}_{\reg{Z}} 
\ket{\bfk}_{\reg{K}}.  \label{eq:Iphi}
\end{align}
Likewise, $\cI \ket{\phi_{x,L_1,R_1,L_2,R_2}}$ is orthogonal to $\cI \ket{\phi_{x',L'_1,R'_1,L'_2,R'_2}}$ and $\cI$ preserves the inner product between $\ket{\phi_{x,L_1,R_1,L_2,R_2}}$ and $\ket{\phi_{x',L'_1,R'_1,L'_2,R'_2}}$. Thus, $\ket{\phi_{x,L_1,R_1,L_2,R_2}}$ is orthogonal to $\ket{\phi_{x',L'_1,R'_1,L'_2,R'_2}}$, proving Item~2. Finally, from the above calculation, we can easily conclude that $\cI \ket{\psi_{x,L_1,R_1,L_2,R_2}}$ is orthogonal to $\ket{\phi_{x',L'_1,R'_1,L'_2,R'_2}}$
which imply that $\ket{\psi_{x,L_1,R_1,L_2,R_2}}$ is orthogonal to $\ket{\phi_{x',L'_1,R'_1,L'_2,R'_2}}$, proving Item~3.

\myparagraph{Wrap-up}
According to the above argument and~\Cref{lem:op_norm_orthogonal}, it suffices to bound the maximum of
\[
\|\ket{\psi_{x,L_1,R_1,L_2,R_2}} - \ket{\phi_{x,L_1,R_1,L_2,R_2}}\|_2.
\]
over all $x \in [N], L_1, L_2 \in \Rinj_{\leq t}, R_1, R_2 \in \RDdist_{\leq t}$. From the above calculation, this is equivalently reduced to bounding
\[
\|\cI \ket{\psi_{x,L_1,R_1,L_2,R_2}} - \cI \ket{\phi_{x,L_1,R_1,L_2,R_2}}\|_2.
\]

\noindent Finally, by~\Cref{lem:comb_F1L_tuple_difference}, we obtain
\begin{align*}
\|\cI \ket{\psi_{x,L_1,R_1,L_2,R_2}} - \cI \ket{\phi_{x,L_1,R_1,L_2,R_2}}\|^2_2 
=  O( t/N ).
\end{align*}
\noindent This concludes the proof of~\Cref{lem:fl1_fr1}.
\end{proof}

\subsection{Closeness of $F_1^{L,\dagger}$ and $F_1^{R,\dagger}$}

The following lemma implies that in~$\hyb_3$, any state orthogonal to the image of $F_1^L$ remains nearly orthogonal to the image of $F^L$ after the action of $X^{k_3} \Ospru$. Intuitively, this prevents unintended ``cancellation'' between oracle calls to $F^L$.

\begin{lemma}[Image Lemma for $F_1^L$] \label{lem:spru:fl1rev:zero}
For any integer $t \geq 0$ and any normalized state $\ket{\psi}$ on registers $\reg{A}, \reg{B}, \reg{S}_1, \reg{T}_1, \reg{S}_2, \reg{T}_2$ such that $\Pi_{\leq t} \ket{\psi} = \ket{\psi}$ and $F_1^{L,\dagger} \ket{\psi} = 0$, it holds that
\[
\|F^{L,\dagger} X^{k_3} \Ospru \ket{\psi}\|_2 
= O( \sqrt{t/N} ).
\]
\end{lemma}

\begin{proof}

Suppose $\ket{\psi}$ can be written as
\[
\ket{\psi}
= \sum_{ \substack{
    y, b \\
    L_1, L_2, R_1, R_2
} }
\alpha_{y,b,L_1,R_1,L_2,R_2} \ket{y}_{\reg{A}} \ket{b}_{\reg{B}} \ket{L_1}_{\reg{S}_1} \ket{R_1}_{\reg{T}_1} \ket{L_2}_{\reg{S}_2} \ket{R_2}_{\reg{T}_2},
\]
where $y \in [N]$, $L_1, L_2 \in \RIdist_{\le t}$ and $R_1, R_2 \in \RDdist_{\le t}$; recall that $\reg{B}$ is the adversary's auxiliary register, and $b$ ranges from some finite set that we do not explicitly specify.

\myparagraph{Zero condition} The premise implies that
\begin{align*}
0 & = F_1^{L,\dagger} \cdot \ket{\psi} \\
& = F_1^{L,\dagger} \cdot
\sum_{ \substack{
    y, b \\
    L_1, L_2, R_1, R_2
} }
\alpha_{y,b,L_1,R_1,L_2,R_2} \ket{y}_{\reg{A}} \ket{b}_{\reg{B}} \ket{L_1}_{\reg{S}_1} \ket{R_1}_{\reg{T}_1} \ket{L_2}_{\reg{S}_2} \ket{R_2}_{\reg{T}_2} \\
& =  \frac{1}{\sqrt{N}}
\sum_{ \substack{
    b, L_1, L_2, R_1, R_2 \\ 
    (x, y) \in L_1
} }
\alpha_{y,b,L_1,R_1,L_2,R_2} \ket{x}_{\reg{A}} \ket{b}_{\reg{B}} \ket{L_1 \setminus \set{(x,y)}}_{\reg{S}_1} \ket{R_1}_{\reg{T}_1} \ket{L_2}_{\reg{S}_2} \ket{R_2}_{\reg{T}_2}.
\tag{by~\Cref{eq:FL_dagger}}
\end{align*}
By re-writing $L_1 = L'_1 \cup \set{(x,y)}$, we obtain
\begin{align*}
& \frac{1}{\sqrt{N}}
\sum_{ \substack{
    x, b \\
    L'_1, L_2, R_1, R_2 \\
    y \notin \Im(L'_1)
} }
\alpha_{y,b,L'_1\cup\set{(x,y)},R_1,L_2,R_2} \ket{x}_{\reg{A}} \ket{b}_{\reg{B}} \ket{L'_1}_{\reg{S}_1} \ket{R_1}_{\reg{T}_1} \ket{L_2}_{\reg{S}_2} \ket{R_2}_{\reg{T}_2} \\
& = \frac{1}{\sqrt{N}}
\sum_{ \substack{
    x, b \\
    L'_1, L_2, R_1, R_2
} }
\left( \sum_{ \substack{
y \notin \Im(L'_1)
} }
\alpha_{y, b, L'_1 \cup \set{(x,y)},R_1,L_2,R_2} \right) \ket{x}_{\reg{A}} \ket{b}_{\reg{B}} \ket{L'_1}_{\reg{S}_1} \ket{R_1}_{\reg{T}_1} \ket{L_2}_{\reg{S}_2} \ket{R_2}_{\reg{T}_2},
\end{align*}
where $x \in [N]$ and $L'_1 \in \RIdist_{\le t-1}$. Therefore, for any fixed $x \in [N]$, $b$, $L'_1  \in \Rinj_{\leq t-1}, L_2 \in \Rinj_{\leq t}$, and $R_1, R_2 \in \RDdist_{\leq t}$, it holds that
\begin{equation} \label{eq:zero_condition}
\sum_{y \notin \Im(L'_1)} \alpha_{y, b, L'_1 \cup \set{(x,y)}, R_1, L_2, R_2} = 0.
\end{equation}

\myparagraph{Computing $F^{L,\dagger} X^{k_3} \Ospru \ket{\psi}$} Next, we will compute $F^{L,\dagger} X^{k_3} \Ospru \ket{\psi}$. Firstly, by~\Cref{lem:Ospru}, we obtain
\begin{align*}
\ket{\psi}
\xmapsto{X^{k_3} \Ospru}
\sum_{ \substack{
    y, b \\
    L_1, L_2, R_1, R_2 \\ 
    (\bfk, \bfz) \in \Sspru\left(\substack{L_1,L_2 \\ R_1,R_2}\right)
} }
\frac{\alpha_{y,b,L_1,R_1,L_2,R_2}}{\sqrt{N^{|L_2|+|R_2|+3}}}
\ket{y \oplus k_3}_{\reg{A}} \ket{b}_{\reg{B}} 
\ket{L^{(k_1,k_3)}_1 \cup L^{(k_2,\vec{z}_L)}_2}_{\reg{S}} 
\ket{R^{(k_1,k_3)}_1 \cup R^{(k_2,\vec{z}_R)}_2}_{\reg{T}}
\ket{\bfk}_{\reg{K}}.
\end{align*}
Before we move on to apply $F^{L,\dagger}$, recall~\Cref{eq:FL_dagger}. For the expression to be nonzero, it is necessary that $y \oplus k_3 \in L^{(k_1,k_3)}_1 \cup L^{(k_2,\vec{z}_L)}_2$. By~\Cref{lem:good_tuples_satisfy_conditions}, $L^{(k_1,k_3)}_1$ and $L^{(k_2,\vec{z}_L)}_2$ are disjoint. Therefore, $(x, y \oplus k_3)$ must belong to exactly one of the following: (i) $(x, y \oplus k_3) \in L^{(k_1,k_3)}_1$, (ii) $(x, y \oplus k_3) \in L^{(k_2,\vec{z}_L)}_2$. Define two projectors
\begin{align*}
    &\Pi_1 \coloneqq 
    \sum_{  \substack{ 
        (L,\bfk) \colon G^{\cI}_{L,k_2} \text{ is decomposable} \\ 
        y \in \Im(V_{\isolate}(G^{\ell}_{L,k_2}))
    } } 
    \ketbra{y}{y}_{\reg{A}} \otimes \ketbra{L}{L}_{\reg{S}} \otimes \proj{\bfk}_{\reg{K}},\\
    &\Pi_2 \coloneqq 
    \sum_{  \substack{ 
        (L,\bfk) \colon G^{\ell}_{L,k_2} \text{ is decomposable} \\ 
        y \in \Im(V_{\target}(G^{\ell}_{L,k_2}) \cup V_{\source}(G^{\ell}_{L,k_2}))
    } } 
    \ketbra{y}{y}_{\reg{A}} \otimes \ketbra{L}{L}_{\reg{S}} \otimes \proj{\bfk}_{\reg{K}}
\end{align*}
Thus, it holds that 
\begin{align*}
F^{L,\dagger} X^{k_3} \Ospru \ket{\psi} 
= F^{L,\dagger} \Pi_1 X^{k_3} \Ospru \ket{\psi} 
+ F^{L,\dagger} \Pi_2 X^{k_3} \Ospru \ket{\psi}.    
\end{align*}
By the triangle inequality, it suffices to the bound the norm of each term.

\myparagraph{Bounding $F^{L,\dagger} \Pi_1 X^{k_3} \Ospru \ket{\psi}$}
Using~\Cref{eq:FL_dagger}, we obtain
\begin{align*}
& F^{L,\dagger} \Pi_1 X^{k_3} \Ospru \ket{\psi} \\
& = \sum_{ \substack{
    b, L_1, L_2, R_1, R_2 \\
    (\bfk,\bfz) \in \Sspru\left(\substack{L_1,L_2 \\ R_1,R_2}\right) \\
    x, y \colon (x, y \oplus k_3) \in L_1^{(k_1,k_3)}
} }
\frac{\alpha_{y,b,L_1,R_1,L_2,R_2}}{\sqrt{N^{|L_2|+|R_2|+4}}}
\ket{x}_{\reg{A}} 
\ket{b}_{\reg{B}} 
\ket{L^{(k_1,k_3)}_1 \cup L^{(k_2,\vec{z}_L)}_2 \setminus \set{(x,y \oplus k_3)}}_{\reg{S}}
\ket{R^{(k_1,k_3)}_1 \cup R^{(k_2,\vec{z}_R)}_2}_{\reg{T}}
\ket{\bfk}_{\reg{K}}.
\end{align*}
Crucially, we use~\Cref{lem:conditions_robust_decodability} to conclude that $V_{\isolate}\left(G^{\ell}_{L^{(k_1,k_3)}_1 \cup L^{(k_2,\vec{z}_L)}_2,k_2}\right) = L^{(k_1,k_3)}_1$. Thus, we are summing over elements in $L^{(k_1,k_3)}_1$ in the above line. \\

\noindent By substituting $x = x' \oplus k_1$ and using that $(\hat{x} \oplus k_1, \hat{y} \oplus k_3) \in L^{(k_1,k_3)}_1$ if and only if $(\hat{x}, \hat{y}) \in L_1$, we obtain
\begin{align*}
& \sum_{ \substack{
    b, L_1, L_2, R_1, R_2 \\
    (\bfk,\bfz) \in \Sspru\left(\substack{L_1,L_2 \\ R_1,R_2}\right) \\
    \rcolor{(x', y) \in L_1}
} }
\frac{\alpha_{y, b, L_1, R_1, L_2, R_2}}{\sqrt{N^{|L_2|+|R_2|+4}}}
\ket{\rcolor{x' \oplus k_1}}_{\reg{A}} 
\ket{b}_{\reg{B}}
\ket{(\rcolor{L_1 \setminus \set{(x', y)}})^{(k_1,k_3)} \cup L^{(k_2,\vec{z}_L)}_2}_{\reg{S}}
\ket{R^{(k_1,k_3)}_1 \cup R^{(k_2,\vec{z}_R)}_2}_{\reg{T}}
\ket{\bfk}_{\reg{K}}.
\end{align*}
Substituting $L_1 = L'_1 \cup \set{(x',y)}$, we obtain
\begin{align}   \label{eq:w1_simplifed}
& \sum_{ \substack{
    b, \rcolor{L'_1}, L_2, R_1, R_2 \\
    x', \rcolor{y \notin \Im(L'_1)} \\
    (\bfk,\bfz) \in \Sspru\left(\substack{\rcolor{L'_1 \cup \set{(x',y)}},L_2 \\ R_1,R_2}\right)
} }
\frac{\alpha_{y, b, \rcolor{L'_1 \cup \set{(x',y)}}, R_1, L_2, R_2}}{\sqrt{N^{|L_2|+|R_2|+4}}}
\ket{x' \oplus k_1}_{\reg{A}} \ket{b}_{\reg{B}}
\ket{(\rcolor{L'_1})^{(k_1,k_3)} \cup L^{(k_2,\vec{z}_L)}_2}_{\reg{S}} 
\ket{R^{(k_1,k_3)}_1 \cup R^{(k_2,\vec{z}_R)}_2}_{\reg{T}}
\ket{\bfk}_{\reg{K}}.
\end{align}
We will further simplify the state. By monotonicity (\Cref{lem:good_tuples_monotonicity}), we can see that \Cref{eq:w1_simplifed} is in the domain of the partial isometry $\Osprus$. Thus, we can equivalently evaluate the norm of $\Osprus \cdot X^{k_1} \cdot F^{L,\dagger} \Pi_1 X^{k_3} \Ospru \ket{\psi}$. Thus, applying $\Osprus \cdot X^{k_1}$ to~\Cref{eq:w1_simplifed}, we obtain
\begin{align*}
& \sum_{ \substack{
    b, L'_1, L_2, R_1, R_2 \\
    x', y \notin \Im(L'_1) \\
    (\bfk,\bfz) \in \Sspru\left(\substack{L'_1 \cup \set{(x',y)},L_2 \\ R_1,R_2}\right)
} }
\frac{\alpha_{y, b, L'_1 \cup \set{(x', y)}, R_1, L_2, R_2}}{\sqrt{N^{|L_2|+|R_2|+4}}}
\ket{\rcolor{x'}}_{\reg{A}} \ket{b}_{\reg{B}}
\rcolor{\ket{L'_1}_{\reg{S}_1} \ket{R_1}_{\reg{T}_1} 
\ket{L_2}_{\reg{S}_2} \ket{R_2}_{\reg{T}_2} 
\ket{\bfz}_{\reg{Z}}}
\ket{\bfk}_{\reg{K}}.
\end{align*}
Rearranging, we obtain
\begin{align}
& \sum_{ \substack{
    b, L'_1, L_2, R_1, R_2, x', \bfk, \bfz \\
    y \notin \Im(L'_1) \colon (\bfk,\bfz) \in \Sspru\left(\substack{L'_1 \cup \set{(x',y)},L_2 \\ R_1,R_2}\right)
} }
\frac{\alpha_{y, b, L'_1 \cup \set{(x', y)}, R_1, L_2, R_2}}{\sqrt{N^{|L_2|+|R_2|+4}}}
\ket{x'}_{\reg{A}} \ket{b}_{\reg{B}}
\ket{L'_1}_{\reg{S}_1} \ket{L_2}_{\reg{S}_2} 
\ket{R_1}_{\reg{T}_1} \ket{R_2}_{\reg{T}_2} 
\ket{\bfz}_{\reg{Z}}
\ket{\bfk}_{\reg{K}},
\label{eq:OsXk1w1}
\end{align}
where $x' \in [N], k_1, k_2, k_3 \in [N], \vec{z}_L \in [N]^{|L_2|}_\dist, \vec{z}_R \in [N]^{|R_2|}_\dist$ and we range over $y$ at the end. \\

\noindent The squared norm of~\Cref{eq:OsXk1w1} is
\begin{align*}
\sum_{  \substack{
    b, L'_1, L_2, R_1, R_2, x' \\
    \bfk, \bfz
}   }
\Bigl|
\sum_{y \notin \Im(L'_1) \colon (\bfk,\bfz) \in \Sspru\left(\substack{L'_1 \cup \set{(x',y)},L_2 \\ R_1,R_2}\right)}
\frac{\alpha_{y, b, L'_1 \cup \set{(x', y)}, R_1, L_2, R_2}}{\sqrt{N^{|L_2|+|R_2|+4}}}
\Bigr|^2.
\end{align*}

By~\Cref{lem:comb_F1Ldagger_y_difference}, once $(L'_1, L_2, R_1, R_2, x', \bfk, \bfz)$ is fixed, there is either zero or at least $N - g(t)$ values of $y$ that satisfy the condition, where $g(t) = O(t)$ is some function guaranteed by~\Cref{lem:comb_F1Ldagger_y_difference}. Let $\mathsf{BAD}$ denote the set of tuples for which the latter case holds. We obtain
\begin{align*}
\sum_{b, \rcolor{(L'_1, L_2, R_1, R_2, x', \bfk, \bfz) \in \mathsf{BAD}}}
\Bigl|
\sum_{y \notin \Im(L'_1) \colon (\bfk,\bfz) \in \Sspru\left(\substack{L'_1 \cup \set{(x',y)},L_2 \\ R_1,R_2}\right)}
\frac{\alpha_{y, b, L'_1 \cup \set{(x', y)}, R_1, L_2, R_2}}{\sqrt{N^{|L_2|+|R_2|+4}}}
\Bigr|^2.
\end{align*}
Now, we make crucial use of the condition implied by the premise (\Cref{eq:zero_condition}) to obtain
\begin{align}   \label{eq:w1_norm}
& \sum_{b, (L'_1, L_2, R_1, R_2, x', \bfk, \bfz) \in \mathsf{BAD}}
\Bigl|
\rcolor{-} 
\sum_{y \notin \Im(L'_1) \colon (\bfk,\bfz) \rcolor{\notin} \Sspru\left(\substack{L'_1 \cup \set{(x',y)},L_2 \\ R_1,R_2}\right)}
\frac{\alpha_{y, b, L'_1 \cup \set{(x', y)}, R_1, L_2, R_2}}{\sqrt{N^{|L_2|+|R_2|+4}}}
\Bigr|^2 
\notag \\
= & \sum_{b, (L'_1, L_2, R_1, R_2, x', \bfk, \bfz) \in \mathsf{BAD}}
\frac{1}{N^{|L_2|+|R_2|+4}} \cdot
\Bigl|
\sum_{y \notin \Im(L'_1) \colon (\bfk,\bfz) \notin \Sspru\left(\substack{L'_1 \cup \set{(x',y)},L_2 \\ R_1,R_2}\right)}
\alpha_{y, b, L'_1 \cup \set{(x', y)}, R_1, L_2, R_2}
\Bigr|^2.
\end{align}
Using the Cauchy-Schwarz inequality and the definition of $\mathsf{BAD}$ to bound the number of $y$ in the sum, we can bound~\Cref{eq:w1_norm} by 
\begin{align*}
\sum_{ \substack{
    b, (L'_1, L_2, R_1, R_2, x', \bfk, \bfz) \in \mathsf{BAD} \\
    y \notin \Im(L'_1) \colon (\bfk,\bfz) \notin \Sspru\left(\substack{L'_1 \cup \set{(x',y)},L_2 \\ R_1,R_2}\right)
}   }
\frac{\rcolor{g(t)}}{N^{|L_2|+|R_2|+4}} \cdot
\Bigl|\alpha_{y, b, L'_1 \cup \set{(x', y)}, R_1, L_2, R_2}\Bigr|^2,
\end{align*}
Since we are summing over non-negative terms, by relaxing the constraints, we can bound it by
\begin{align*}
\sum_{  \substack{
    b, \rcolor{L'_1, L_2, R_1, R_2, x', \bfk,\bfz} \\
    \rcolor{y \notin \Im(L'_1)}
}   }
\frac{g(t)}{N^{|L_2|+|R_2|+4}} \cdot
\Bigl|\alpha_{y, b, L'_1 \cup \set{(x', y)}, R_1, L_2, R_2}\Bigr|^2.
\end{align*}
By summing over $(\bfk,\bfz)$, and noting that there are at most $N^{|L_2|+|R_2|+3}$ such tuples, we can bounded it by
\begin{align*}
\frac{g(t)}{N} \cdot 
\sum_{b, L'_1, L_2, R_1, R_2, x', y \notin \Im(L'_1)}
\Bigl|\alpha_{y, b, L'_1 \cup \set{(x', y)}, R_1, L_2, R_2}\Bigr|^2.
\end{align*}
By substituting $L = L'_1 \cup \set{(x',y)}$, we obtain
\begin{align*}
\frac{g(t)}{N} \cdot 
\sum_{b, \rcolor{L_1}, L_2, R_1, R_2, \rcolor{y \in \Im(L_1)}}
\Bigl|\alpha_{y, b, \rcolor{L_1}, R_1, L_2, R_2}\Bigr|^2 
= O(t/N)
\end{align*}
by the normalization condition of $\ket{\psi}$.

\myparagraph{Bounding $F^{L,\dagger} \Pi_2 X^{k_3} \Ospru \ket{\psi}$}

Recall~\Cref{lem:punc_S}. For any $(y,L_1,R_1,L_2,R_2)$, define the set
\begin{align*}
    \cP_{y,L_1,R_1,L_2,R_2} := \set{(\bfk,\bfz) \colon y \oplus k_3 \in \Im(L_2^{(k_2,\vec{z}_L)})}.
\end{align*}
By sampling $k_3$ at the end and the union bound, it is clear that $\cP_{y,L_1,R_1,L_2,R_2}$ occupies at most a $2t/N$ fraction of its universe. Thus, we obtain
\begin{align*}
    \|F^{L,\dagger} \Pi_2 X^{k_3} \Ospru \ket{\psi}\|_2
    & \le \|F^{L,\dagger} \Pi_2 X^{k_3} \Ospru^{\bullet} \ket{\psi}\|_2
    + \|F^{L,\dagger} \Pi_2 X^{k_3} (\Ospru - \Ospru^{\bullet}) \ket{\psi}\|_2\\
    & \le \|\Pi_2 X^{k_3} \Ospru^{\bullet} \ket{\psi}\|_2
    + \|(\Ospru^{\bullet} - \Ospru)\Pi_{\le t}\|_{\op} \\
    & \le O(\sqrt{t/N}).
\end{align*}
The first term is zero by the definitions of $\set{\cP_\tau}_\tau$ and $\Ospru^{\bullet}$. The second term is bounded by~\Cref{lem:punc_S}. This completes the proof of~\Cref{lem:spru:fl1rev:zero}.
\end{proof}

\begin{lemma}[Closeness of $F_1^{L,\dagger}$ and $F_1^{R,\dagger}$] \label{lem:fl1dagger_fr1dagger}
For any integer $t \geq 0$,
\begin{align*}
& \|(X^{k_1} F^{L,\dagger} X^{k_3} \Ospru - \Ospru F_1^{L,\dagger}) \Pi_{\leq t}\|_{\op} 
= O ( \sqrt{t/N} ) \\
& \|(X^{k_1} F^{R,\dagger} X^{k_3} \Ospru - \Ospru F_1^{R,\dagger}) \Pi_{\leq t}\|_{\op} 
= O ( \sqrt{t/N} ).
\end{align*}
\end{lemma}

\begin{proof}
Let $\Pi^{\Im(F_1^L)}$ denote the projection onto the image of $F_1^L$. For an arbitrary normalized state $\ket{\psi}$ on registers $\reg{A}, \reg{S}_1, \reg{T}_1, \reg{S}_2, \reg{T}_2$ in the subspace of $\Pi_{\leq t}$, we can decompose it as
\[
\ket{\psi}
= \Pi^{\Im(F_1^L)} \ket{\psi} + ( \id - \Pi^{\Im(F_1^L)} ) \ket{\psi}.
\]
We will show the following two bounds
\begin{align}
& \|( X^{k_1} F^{L,\dagger} X^{k_3} \Ospru - \Ospru F_1^{L,\dagger} ) \Pi^{\Im(F_1^L)} \ket{\psi}\|_2 
\leq O ( \sqrt{t/N} ) \label{ineq:Im} \\
& \|( X^{k_1} F^{\dagger,L} X^{k_3} \Ospru - \Ospru F_1^{L,\dagger} ) ( \id - \Pi^{\Im(F_1^L)} ) \ket{\psi}\|_2 
\leq O ( \sqrt{t/N} ) \label{ineq:I-Im}
\end{align}
which then the first bound follows by the triangle inequality. Notice that $F^{L,\dagger} ( \id - \Pi^{\Im(F_1^L)} ) = 0$, so \Cref{ineq:I-Im} follows from~\Cref{lem:spru:fl1rev:zero}. Hence, it remains to prove~\Cref{ineq:Im}. Since $\Pi^{\Im(F_1^L)} \ket{\psi}$ is in the image of $F_1^L$, there exists some state $\ket{\phi}$ such that $\Pi^{\Im(F_1^L)} \ket{\psi} = F_1^L \ket{\phi}$. Now, we bound~\Cref{ineq:Im} by the triangle inequality as follows:
\begin{align}
& \|(X^{k_1} F^{L,\dagger} X^{k_3} \Ospru - \Ospru F_1^{L,\dagger}) \Pi^{\Im(F_1^L)} \ket{\psi}\|_2 
\notag \\
= & \|( X^{k_1} F^{L,\dagger} X^{k_3} \Ospru 
- \Ospru F_1^{L,\dagger} ) F_1^L \ket{\phi}\|_2 
\notag \\
\leq & \|X^{k_1} F^{L,\dagger} X^{k_3} \Ospru F_1^L \ket{\phi}
\rcolor{ - X^{k_1} F^{L,\dagger} F^L X^{k_1}  \Ospru  \ket{\phi} }\|_2 
\label{eq:term_1} \\
& \hspace{.2\textwidth} + \|\rcolor{ X^{k_1} F^{L,\dagger} F^L X^{k_1}  \Ospru  \ket{\phi} } 
- \Ospru F_1^{L,\dagger} F_1^L \ket{\phi} \|_2.
\label{eq:term_2}
\end{align}
To bound~\Cref{eq:term_1}, we have
\begin{align*}
& \|X^{k_1} F^{L,\dagger} X^{k_3} \Ospru F_1^L \ket{\phi}
- X^{k_1} F^{L,\dagger} F^L X^{k_1} \Ospru \ket{\phi} \|_2 \\
= & \|X^{k_1} F^{L,\dagger} X^{k_3} \Ospru F_1^L \ket{\phi}
- X^{k_1} F^{L,\dagger} \rcolor{ X^{k_3} X^{k_3} } F^L X^{k_1} \Ospru \ket{\phi} \|_2 
\tag{since $X^{k_3} X^{k_3} = \id$} \\
= & \|X^{k_1} F^{L,\dagger} X^{k_3} \cdot ( \Ospru F_1^L - X^{k_3} F^L X^{k_1} \Ospru ) \ket{\phi} \|_2 \\
\le & \|X^{k_1} F^{L,\dagger} X^{k_3} \cdot ( \Ospru F_1^L - X^{k_3} F^L X^{k_1} \Ospru ) \Pi_{\leq t} \|_{\op} \\
\leq & \|F^{L,\dagger}\|_{\op} \cdot \|( \Ospru F_1^L - X^{k_3} F^L X^{k_1} \Ospru ) \Pi_{\leq t}\|_{\op} 
\tag{since operator norm is submultiplicative} \\
\leq & O ( \sqrt{t/N} ).
\tag{by the fact that $F^L$ is a contraction and~\Cref{lem:fl1_fr1}}
\end{align*}

\noindent To bound~\Cref{eq:term_2}, we first expand the terms. By~\Cref{lem:domain_FL_FR}, we can replace $F^{L,\dagger} F^L$ with $\id$ as follows:
\begin{align*}
& \|(X^{k_1} F^{L,\dagger} F^L X^{k_1}  \Ospru - \Ospru F_1^{L,\dagger} F_1^L) \Pi_{\le t}\|_{\op}\\
\le & \|(X^{k_1} F^{L,\dagger} F^L X^{k_1} \Ospru - \rcolor{X^{k_1}X^{k_1} \Ospru}) \Pi_{\le t}\|_{\op} 
+ \|(\rcolor{X^{k_1}X^{k_1} \Ospru} - \Ospru F_1^{L,\dagger} F_1^L) \Pi_{\le t}\|_{\op}\\
\le & \|(F^{L,\dagger} F^L - \id) X^{k_1} \Ospru \Pi_{\le t}\|_{\op}
+ \|(\id - F_1^{L,\dagger} F_1^L) \Pi_{\le t}\|_{\op}\\
\le & \|(F^{L,\dagger} F^L - \id)\Pi_{\le 3t}\|_{\op}
+ \|(\id - F_1^{L,\dagger} F_1^L) \Pi_{\le t}\|_{\op}\\
\le & O(t/N).
\end{align*}
This proves the first bound. The second bound follows symmetrically.
\end{proof}

\subsection{Putting Things Together}

Finally, we use the above lemmas to prove~\Cref{lem:1st_oracle}. We restate~\Cref{lem:1st_oracle} for convenience.

\begin{lemma}[\Cref{lem:1st_oracle}, restated]
For any integer $t \geq 0$,
\begin{itemize}
    \item {\bf Forward query:} $\|(X^{k_3} F X^{k_1} \Ospru - \Ospru F_1) \Pi_{\leq t}\|_{\op} \leq O ( \sqrt{t/N} )$,
    \item {\bf Inverse query:} $\|(X^{k_1} F^{\dagger} X^{k_3} \Ospru - \Ospru F_1^{\dagger}) \Pi_{\leq t}\|_{\op} \leq O ( \sqrt{t/N} )$.
\end{itemize}
\end{lemma}

\begin{proof}[Proof of~\Cref{lem:1st_oracle}]
\ \\
\myparagraph{Forward query}
Recall the definition of $F_1$:
\begin{align*}
F_1 
& = F_1^L \cdot (\id - F_1^R \cdot F_1^{R,\dagger}) + (\id - F_1^L \cdot F_1^{L,\dagger}) \cdot F_1^{R,\dagger} \\
& = F_1^L - F_1^L \cdot F_1^R \cdot F_1^{R,\dagger} + F_1^{R,\dagger} - F_1^L \cdot F_1^{L,\dagger} \cdot F_1^{R,\dagger}.
\end{align*}
\noindent Similarly, we expand $X^{k_3} F X^{k_1}$ in the following:
\begin{align*}
X^{k_3} F X^{k_1} 
= & X^{k_3} F^L \cdot ( \id - F^R F^{R,\dagger} ) X^{k_1} + X^{k_3} ( \id - F^L F^{L,\dagger} ) \cdot F^{R,\dagger} X^{k_1} \\
= & X^{k_3} F^L X^{k_1} - X^{k_3} F^L F^R F^{R,\dagger} X^{k_1} + X^{k_3} F^{R,\dagger} X^{k_1} - X^{k_3} F^L F^{L,\dagger} F^{R,\dagger} X^{k_1} \\
= & X^{k_3} F^L X^{k_1} - ( X^{k_3} F^L X^{k_1} ) ( X^{k_1} F^R X^{k_3} )( X^{k_3} F^{R,\dagger} X^{k_1} ) \\
& + X^{k_3} F^{R,\dagger} X^{k_1} - ( X^{k_3} F_L X^{k_1} )( X^{k_1} F^{L,\dagger} X^{k_3} )( X^{k_3} F^{R,\dagger} X^{k_1} ).
\end{align*}
\noindent By the triangle inequality, it suffices to bound each of the following terms:
\begin{align}
& \|(X^{k_3} F X^{k_1} \Ospru - \Ospru F_1) \Pi_{\leq t}\|_{\op}
\notag \\
\leq & \|(X^{k_3} F^L X^{k_1} \Ospru - \Ospru F_1^L) \Pi_{\leq t}\|_{\op} 
\label{eq:1st_oracle_term1} \\
+ & \|( ( X^{k_3} F^L X^{k_1} ) ( X^{k_1} F^R X^{k_3} ) ( X^{k_3} F^{R,\dagger} X^{k_1} ) \Ospru - \Ospru F_1^L F_1^R F_1^{R,\dagger} ) \Pi_{\leq t}\|_{\op}
\label{eq:1st_oracle_term2} \\
+ & \|(X^{k_3} F^{R,\dagger} X^{k_1} \Ospru - \Ospru F_1^{R,\dagger}) \Pi_{\leq t}\|_{\op}
\label{eq:1st_oracle_term3} \\
+ & \|( ( X^{k_3} F^L X^{k_1} ) ( X^{k_1} F^{L,\dagger} X^{k_3} ) ( X^{k_3} F^{R,\dagger} X^{k_1} ) \Ospru - \Ospru F_1^L F_1^{L,\dagger} F_1^{R,\dagger} ) \Pi_{\leq t}\|_{\op}.
\label{eq:1st_oracle_term4}
\end{align}
From~\Cref{lem:fl1_fr1}, we can bound \Cref{eq:1st_oracle_term1} by $O( \sqrt{t/N} )$. From~\Cref{lem:fl1dagger_fr1dagger}, we can bound \Cref{eq:1st_oracle_term3} by $O( \sqrt{t^2/N} )$. To bound~\Cref{eq:1st_oracle_term2}, by the triangle inequality, we have
\begin{align*}
& \|( (X^{k_3} F^L X^{k_1}) (X^{k_1} F^R X^{k_3}) (X^{k_3} F^{R,\dagger} X^{k_1}) \Ospru - \Ospru F_1^L F_1^R F_1^{R,\dagger} ) \Pi_{\leq t}\|_{\op} \\
\leq & \|((X^{k_3} F^L X^{k_1}) (X^{k_1} F^R X^{k_3}) \rcolor{( X^{k_3} F^{R,\dagger} X^{k_1} ) \Ospru} - (X^{k_3} F^L X^{k_1}) (X^{k_1} F^R X^{k_3}) \rcolor{\Ospru F_1^{R,\dagger}}) \Pi_{\leq t}\|_{\op} \\
&\hspace{.1\textwidth} + \|( (X^{k_3} F^L X^{k_1}) \rcolor{(X^{k_1} F^R X^{k_3}) \Ospru} F_1^{R,\dagger} - (X^{k_3} F^L X^{k_1}) \rcolor{\Ospru F_1^R} F_1^{R,\dagger} ) \Pi_{\leq t}\|_{\op} \\
&\hspace{.2\textwidth} + \|( \rcolor{( X^{k_3} F^L X^{k_1} ) \Ospru} F_1^R F_1^{R,\dagger}
- \rcolor{\Ospru F_1^L} F_1^R F_1^{R,\dagger} ) \Pi_{\leq t}\|_{\op} \\
\leq & O ( \sqrt{t/N} )
\end{align*}
where we use (i) that operator norm is submultiplicative; (ii) that $F^L, F^R$ are contractions; and (iii) \Cref{lem:fl1dagger_fr1dagger}, \Cref{lem:fl1_fr1}, and \Cref{lem:fl1_fr1}, in order, for the colored parts. Similarly, \Cref{eq:1st_oracle_term4} is at most $O( t/\sqrt{N} )$ by~\Cref{lem:fl1dagger_fr1dagger}, \Cref{lem:fl1dagger_fr1dagger}, and \Cref{lem:fl1_fr1}. Combining the above terms, we obtain
\begin{align*}
\|(X^{k_3} F X^{k_1} \Ospru - \Ospru F_1) \Pi_{\leq t}\|_{\op} 
= O( \sqrt{t/N} ).
\end{align*}
This proves the first item of~\Cref{lem:1st_oracle}. The proof of the second item follows by symmetry. This completes the proof of~\Cref{lem:1st_oracle}.
\end{proof}

\section{Closeness of the Second Oracle: Proving~\Cref{lem:spru:2nd_oracle}} \label{sec:spru:2nd_oracle}
The high-level proof strategy is similar to that in~\Cref{sec:spru:1st_oracle}, where we bound each pair of terms separately. The main distinction is that a query to $F X^{k_2} F$ involves two calls to $F$, which makes the calculation more intricate. We start by listing some helpful lemmas.

\subsection{Closeness of $F_2^L$ and $F_2^R$}

\begin{lemma}[Closeness of $F_2^L$ and $F_2^R$]   \label{lem:fl2_fr2}
For any integer $t \geq 0$,
\begin{align*}
& \|(F^L X^{k_2} F^L \Ospru - \Ospru F_2^L) \Pi_{\leq t}\|_{\op}
= O( \sqrt{t/N} ) \\
& \|(F^R X^{k_2} F^R \Ospru - \Ospru F_2^R) \Pi_{\leq t}\|_{\op}
= O( \sqrt{t/N} ).
\end{align*}
\end{lemma}

\begin{proof}
Fix $t \in \N$, $x \in [N], L_1, L_2 \in \Rinj_{\leq t}$, and $R_1, R_2 \in \RDdist_{\leq t}$. We start by calculating the following states:
\begin{align*}
    & \ket{\psi_{x,L_1,R_1,L_2,R_2}}_{\reg{ASTK}} := F^L X^{k_2} F^L \Ospru \ket{x}_{\reg{A}} \ket{L_1}_{\reg{S}_1} \ket{R_1}_{\reg{T}_1} \ket{L_2}_{\reg{S}_2} \ket{R_2}_{\reg{T}_2}, \\
    & \ket{\phi_{x,L_1,R_1,L_2,R_2}}_{\reg{ASTK}} := \Ospru F_2^L \ket{x}_{\reg{A}} \ket{L_1}_{\reg{S}_1} \ket{R_1}_{\reg{T}_1} \ket{L_2}_{\reg{S}_2} \ket{R_2}_{\reg{T}_2}.
\end{align*}

\myparagraph{Computing $\ket{\psi_{x,L_1,R_1,L_2,R_2}}$}
Expanding the definitions of $\Ospru$ and $F^L$, we have

\begin{align}
& \ket{x}_{\reg{A}} \ket{L_1}_{\reg{S}_1} \ket{R_1}_{\reg{T}_1} \ket{L_2}_{\reg{S}_2} \ket{R_2}_{\reg{T}_2}
\notag \\
& \xmapsto{\Ospru} \frac{1}{\sqrt{N^{|L_2|+|R_2|+3}}}
\sum_{\Sspru\left(\substack{L_1,L_2 \\ R_1,R_2}\right)}
\ket{x}_{\reg{A}}
\ket{L_1^{(k_1,k_3)} \cup L_2^{(k_2,\vec{z}_L)}}_{\reg{S}} 
\ket{R^{(k_1,k_3)}_1 \cup R_2^{(k_2,\vec{z}_R)}}_{\reg{T}}
\ket{\bfk}_{\reg{K}}
\tag{by~\Cref{lem:Ospru}} \\
& \xmapsto{F^L} 
\frac{1}{\sqrt{N^{|L_2|+|R_2|+4}}}
\sum_{  \substack{
    \Sspru\left(\substack{L_1,L_2 \\ R_1,R_2}\right) \\
    z' \notin \Im(L_1^{(k_1,k_3)} \cup L_2^{(k_2,\vec{z}_L)})
}   }
\ket{z'}_{\reg{A}} 
\ket{L_1^{(k_1,k_3)} \cup L_2^{(k_2,\vec{z}_L)} \cup \set{(x,z')}}_{\reg{S}} 
\ket{R^{(k_1,k_3)}_1 \cup R_2^{(k_2,\vec{z}_R)}}_{\reg{T}}
\ket{\bfk}_{\reg{K}}
\tag{by~\Cref{eq:def:F_L}} \\
& \xmapsto{F^L X^{k_2}} \frac{1}{\sqrt{N^{|L_2|+|R_2|+5}}}
\sum_{  \substack{
    \Sspru\left(\substack{L_1,L_2 \\ R_1,R_2}\right) \\
    z', y \notin \Im(L_1^{(k_1,k_3)} \cup L_2^{(k_2,\vec{z}_L)}): \\
    y \neq z'
}   }
\ket{y}_{\reg{A}} 
\ket{L_1^{(k_1,k_3)} \cup L_2^{(k_2,\vec{z}_L)} \cup \set{(x,z'), (z' \oplus k_2, y)}}_{\reg{S}}
\ket{R^{(k_1,k_3)}_1 \cup R_2^{(k_2,\vec{z}_R)}}_{\reg{T}} 
\ket{\bfk}_{\reg{K}}, 
\label{eq:psi_2}
\end{align}
where the last line is by~\Cref{eq:def:F_L}.

\myparagraph{Calculating $\ket{\phi_{x,L_1,R_1,L_2,R_2}}$}
Similarly, expanding the definitions of $\Ospru$ and $F_2^L$, we have
\begin{align}
& \frac{1}{\sqrt{N^{|L_2|+|R_2|+5}}}
\sum_{  \substack{
    y \notin \Im(L_2) \\
    (\bfk,\bfz) \in \Sspru\left(\substack{L_1,L_2 \cup \set{(x,y)}\\ R_1,R_2}\right)
}   }
\ket{y}_{\reg{A}}
\ket{L_1^{(k_1,k_3)} \cup ( L_2 \cup \set{(x,y)} )^{(k_2,\vec{z}_L)}}_{\reg{S}} 
\ket{R^{(k_1,k_3)}_1 \cup R_2^{(k_2,\vec{z}_R)}}_{\reg{T}}
\ket{\bfk}_{\reg{K}}.
\end{align}

\myparagraph{Orthogonality}
Consider distinct $(x,L_1,R_1,L_2,R_2)$ and $(x',L'_1,R'_1,L'_2,R'_2)$. We claim that 
\begin{itemize}
    \item $\ket{\psi_{x,L_1,R_1,L_2,R_2}}$ is orthogonal to $\ket{\psi_{x',L'_1,R'_1,L'_2,R'_2}}$,
    \item $\ket{\phi_{x,L_1,R_1,L_2,R_2}}$ is orthogonal to $\ket{\phi_{x',L'_1,R'_1,L'_2,R'_2}}$,
    \item $\ket{\psi_{x,L_1,R_1,L_2,R_2}}$ is orthogonal to $\ket{\phi_{x',L'_1,R'_1,L'_2,R'_2}}$.
\end{itemize}
They together implies that $\ket{\psi_{x,L_1,R_1,L_2,R_2}} - \ket{\phi_{x,L_1,R_1,L_2,R_2}}$ is orthogonal to $\ket{\psi_{x',L'_1,R'_1,L'_2,R'_2}} - \ket{\phi_{x',L'_1,R'_1,L'_2,R'_2}}$. Thus, by~\Cref{lem:op_norm_orthogonal}, it suffices to maximize the norm over input states of the form $\ket{x} \ket{L_1} \ket{R_1} \ket{L_2} \ket{R_2}$. To prove the claim, we define the operator
\[
\cI \coloneqq \Osprus \cdot F_{\extract}^L \cdot X^{k_2} \cdot F_{\extract}^L
\]
where the partial isometry $F_{\extract}^L$ is defined in~\Cref{eq:FL_extract}. Note that $\cI$ preserves inner product between the states under consideration. To see this, we may compute the states obtained by applying $\cI$ to them:
\begin{align}
&\ket{\psi_{x,L_1,R_1,L_2,R_2}} \\
&\xmapsto{F_{\extract}^L} 
\frac{1}{\sqrt{N^{|L_2|+|R_2|+5}}}
\sum_{  \substack{
    (\bfk,\bfz) \in \Sspru\left(\substack{L_1,L_2 \\ R_1,R_2}\right) \\
    z', y \notin \Im(L_1^{(k_1,k_3)} \cup L_2^{(k_2,\vec{z}_L)}): \\
    y \neq z'
}   }
\ket{y}_{\reg{A}'}
\ket{z' \oplus k_2}_{\reg{A}} 
\ket{L_1^{(k_1,k_3)} \cup L_2^{(k_2,\vec{z}_L)} \cup \set{(x,z')}}_{\reg{S}} 
\ket{R^{(k_1,k_3)}_1 \cup R_2^{(k_2,\vec{z}_R)}}_{\reg{T}} 
\ket{\bfk}_{\reg{K}}
\tag{by~\Cref{eq:FL_extract}} \\
&\xmapsto{F_{\extract}^L X^{k_2}} \frac{1}{\sqrt{N^{|L_2|+|R_2|+5}}}
\sum_{  \substack{
    (\bfk,\bfz) \in \Sspru\left(\substack{L_1,L_2 \\ R_1,R_2}\right) \\
    z', y \notin \Im(L_1^{(k_1,k_3)} \cup L_2^{(k_2,\vec{z}_L)}): \\
    y \neq z'
}   }
\ket{z'}_{\reg{A}''}
\ket{y}_{\reg{A}'}
\ket{x}_{\reg{A}}
\ket{L_1^{(k_1,k_3)} \cup L_2^{(k_2,\vec{z}_L)}}_{\reg{S}} 
\ket{R^{(k_1,k_3)}_1 \cup R_2^{(k_2,\vec{z}_R)}}_{\reg{T}} 
\ket{\bfk}_{\reg{K}}
\tag{by~\Cref{eq:FL_extract}} \\
&\xmapsto{\Osprus} 
\frac{1}{\sqrt{N^{|L_2|+|R_2|+5}}}
\sum_{  \substack{
    (\bfk,\bfz) \in \Sspru\left(\substack{L_1,L_2 \\ R_1,R_2}\right) \\
    z', y \notin \Im(L_1^{(k_1,k_3)} \cup L_2^{(k_2,\vec{z}_L)}): \\
    y \neq z'
}   }
\ket{z'}_{\reg{A}''}
\ket{y}_{\reg{A}'}
\ket{x}_{\reg{A}} 
\ket{L_1}_{\reg{S}_1} \ket{L_2}_{\reg{S}_2} 
\ket{R_1}_{\reg{T}_1} \ket{R_2}_{\reg{T}_2}
\ket{\bfz}_{\reg{Z}}
\ket{\bfk}_{\reg{K}},
\label{eq:Ipsi2}
\end{align}
where the last line is by~\Cref{def:Ospru_split} and noting that the state in the third line is entirely in the domain of $\Osprus$. From the above calculation, it is clear that $\cI \ket{\psi_{x,L_1,R_1,L_2,R_2}}$ is orthogonal to $\cI \ket{\psi_{x',L'_1,R'_1,L'_2,R'_2}}$ whenever $(x, L_1, L_2, {R}_1, {R}_2) \neq (x', L'_1, L'_2, {R'}_1, {R'}_2)$. \\

\noindent Similarly, we have
\begin{align}
&\ket{\phi_{x,L_1,R_1,L_2,R_2}}
\notag \\
&\xmapsto{\cI} \frac{1}{\sqrt{N^{|L_2|+|R_2|+5}}}
\sum_{  \substack{
    y \notin \Im(L_2) \\
    (\bfk,\bfz) \in \Sspru\left(\substack{L_1,L_2 \cup \set{(x,y)} \\ R_1,R_2}\right) \\
    i \st y\, \in_i\, \Im(L_2) \cup \set{y}
}   }
\ket{z_{L,i}}_{\reg{A}''}
\ket{y}_{\reg{A}'}
\ket{x}_{\reg{A}}
\ket{L_1}_{\reg{S}_1} \ket{L_2}_{\reg{S}_2} 
\ket{R_1}_{\reg{T}_1} \ket{R_2}_{\reg{T}_2}  
\ket{\vec{z}_{L,-i}}_{\reg{Z_L}} \ket{\vec{z}_R}_{\reg{Z_R}}
\ket{\bfk}_{\reg{K}}.
\end{align}
Note that $\vec{z}_L$ is of length $\Im(|L_2|)+1$. In the above expression, $i \in [|\Im(L_2)|+1]$ is the index such that $y$ is the $i$-th largest element in $\Im(L_2) \cup \set{y}$; $z_{L,i}$ is the $i$-th coordinate of $\vec{z}_{L}$; and $\vec{z}_{L,-i}$ is the vector obtained by removing the $i$-th coordinate of $\vec{z}_{L}$. \\

\noindent By rearranging, we can express it as
\begin{align}
&\frac{1}{\sqrt{N^{|L_2|+|R_2|+5}}}
\sum_{  \substack{
    y \notin \Im(L_2), \bfk, \rcolor{z \in [N], \vec{z}_L \in [N]^{|L_2|}_\dist}, \vec{z}_R \in [N]^{|R_2|}_\dist: \\
    i \st y\, \in_i\, \Im(L_2) \cup \set{y} \\
    (\bfk,\rcolor{\vec{z}^{\,(i \gets z)}_L},\vec{z}_R) \in \Sspru\left(\substack{L_1,L_2 \cup \set{(x,y)} \\ R_1,R_2}\right)
}   }
\ket{\rcolor{z}}_{\reg{A}''}
\ket{y}_{\reg{A}'}
\ket{x}_{\reg{A}}
\ket{L_1}_{\reg{S}_1} \ket{L_2}_{\reg{S}_2} 
\ket{R_1}_{\reg{T}_1} \ket{R_2}_{\reg{T}_2}
\ket{\rcolor{\vec{z}_L}}_{\reg{Z_L}} \ket{\vec{z}_R}_{\reg{Z_R}}
\ket{\bfk}_{\reg{K}},
\label{eq:Iphi2}
\end{align}
where $\vec{z}^{\,(i \gets z)}_L$ denotes the vector obtained by inserting $z$ into the $i$-th coordinate of $\vec{z}_L$ and shifting all subsequent coordinates by one. \\

\noindent Likewise, $\cI \ket{\phi_{x,L_1,R_1,L_2,R_2}}$ is orthogonal to $\cI \ket{\phi_{x',L'_1,R'_1,L'_2,R'_2}}$ and $\cI$ preserves the inner product between $\ket{\phi_{x,L_1,R_1,L_2,R_2}}$ and $\ket{\phi_{x',L'_1,R'_1,L'_2,R'_2}}$. Thus, $\ket{\phi_{x,L_1,R_1,L_2,R_2}}$ is orthogonal to $\ket{\phi_{x',L'_1,R'_1,L'_2,R'_2}}$, proving Item~2. Finally, from the above calculation, we can easily conclude that $\cI \ket{\psi_{x,L_1,R_1,L_2,R_2}}$ is orthogonal to $\ket{\phi_{x',L'_1,R'_1,L'_2,R'_2}}$
which imply that $\ket{\psi_{x,L_1,R_1,L_2,R_2}}$ is orthogonal to $\ket{\phi_{x',L'_1,R'_1,L'_2,R'_2}}$, proving Item~3.

\myparagraph{Wrap-up}
According to~\Cref{lem:op_norm_orthogonal}, it is sufficient to bound the maximum of
\[
\| \ket{\psi_{x,L_1,R_1,L_2,R_2}} - \ket{\phi_{x,L_1,R_1,L_2,R_2}} \|_2
\]
over all $x \in [N], L_1, L_2 \in \Rinj_{\leq t}, R_1, R_2 \in \RDdist_{\leq t}$. From the above calculation, this is equivalently reduced to bounding
\[
\| \cI \ket{\psi_{x,L_1,R_1,L_2,R_2}} - \cI \ket{\phi_{x,L_1,R_1,L_2,R_2}} \|_2.
\]
Finally, we use~\Cref{,lem:comb_F2L_tuple_difference} to bound the number of terms in~\Cref{eq:Ipsi2,eq:Iphi2} to obtain
\begin{align*}
\|\cI \ket{\psi_{x,L_1,R_1,L_2,R_2}} - \cI \ket{\phi_{x,L_1,R_1,L_2,R_2}}\|^2_2
= O (t/N).
\end{align*}
This concludes the proof of~\Cref{lem:fl2_fr2}.
\end{proof}

\subsection{Closeness of $F_2^{L,\dagger}$ and $F_2^{R,\dagger}$}

\begin{lemma}[Image Lemma for $F_2^L$]  \label{lem:fl2dagger_fr2dagger:zero}
For any integer $t \geq 0$ and any normalized state $\ket{\psi}$ on registers $\reg{A}, \reg{B}, \reg{S}_1, \reg{T}_1, \reg{S}_2, \reg{T}_2$ such that $\Pi_{\leq t} \ket{\psi} = \ket{\psi}$ and $F_2^{L,\dagger} \ket{\psi} = 0$, it holds that
\[
\|F^{L,\dagger} \Ospru \ket{\psi}\|_2 = O(\sqrt{t/N}).
\]
\end{lemma}

\begin{proof}
Suppose $\ket{\psi}$ can be written as
\[
\ket{\psi}
= \sum_{ \substack{
    y, b \\
    L_1, L_2, R_1, R_2
} }
\alpha_{y,b,L_1,R_1,L_2,R_2} \ket{y}_{\reg{A}} \ket{b}_{\reg{B}} \ket{L_1}_{\reg{S}_1} \ket{R_1}_{\reg{T}_1} \ket{L_2}_{\reg{S}_2} \ket{R_2}_{\reg{T}_2},
\]
where $y \in [N]$, $L_1, L_2 \in \RIdist_{\le t}$ and $R_1, R_2 \in \RDdist_{\le t}$; recall that $\reg{B}$ is the adversary's auxiliary register, and $b$ ranges from some finite set that we do not explicitly specify.

\myparagraph{Zero condition} 
The premise implies that
\begin{align*}
0 & = F_2^{L,\dagger} \cdot \ket{\psi}_{\reg{AB}\reg{S}_1\reg{T}_1\reg{S}_2\reg{T}_2} \\
& = F_2^{L,\dagger} \cdot
\sum_{ \substack{
    y, b \\
    L_1, L_2, R_1, R_2
} }
\alpha_{y,b,L_1,R_1,L_2,R_2} \ket{y}_{\reg{A}} \ket{b}_{\reg{B}} \ket{L_1}_{\reg{S}_1} \ket{R_1}_{\reg{T}_1} \ket{L_2}_{\reg{S}_2} \ket{R_2}_{\reg{T}_2} \\
& = \frac{1}{\sqrt{N}}
\sum_{ \substack{
    b, L_1, L_2, R_1, R_2\\
    (x,y) \in L_2
}   }
\alpha_{y,b,L_1,R_1,L_2,R_2} \ket{x}_{\reg{A}} \ket{b}_{\reg{B}} \ket{L_1}_{\reg{S}_1} \ket{R_1}_{\reg{T}_1} \ket{L_2 \setminus \set{(x,y)}}_{\reg{S}_2} \ket{R_2}_{\reg{T}_2}.
\tag{by~\Cref{eq:FL_dagger}}
\end{align*}
By re-writing $L_2 = L'_2 \cup \set{(x,y)}$, we obtain
\begin{align*}
&\frac{1}{\sqrt{N}}
\sum_{ \substack{
    x, b \\
    L_1, L'_2, R_1, R_2 \\
    y \notin \Im(L'_2)
} }
\alpha_{y,b,L_1,R_1,L'_2 \cup \set{(x,y)},R_2} 
\ket{x}_{\reg{A}} \ket{b}_{\reg{B}} \ket{L_1}_{\reg{S}_1} \ket{R_1}_{\reg{T}_1} \ket{L'_2}_{\reg{S}_2} \ket{R_2}_{\reg{T}_2} \\
& = \frac{1}{\sqrt{N}}
\sum_{ \substack{
    x, b \\
    L_1, L'_2, R_1, R_2
} }
\left( \sum_{ \substack{
y \notin \Im(L'_2)
} }
\alpha_{y, b, L_1, R_1, L'_2  \cup \set{(x,y)}, R_2} \right) 
\ket{x}_{\reg{A}} \ket{b}_{\reg{B}} \ket{L_1}_{\reg{S}_1} \ket{R_1}_{\reg{T}_1} \ket{L'_2}_{\reg{S}_2} \ket{R_2}_{\reg{T}_2}.
\end{align*}
Hence, for any $x \in [N]$, $b$, and $L_1 \in \Rinj_{\leq t}, L'_2 \in \Rinj_{\leq t-1}, R_1, R_2 \in \RDdist_{\leq t}$, it holds that
\begin{equation} \label{eq:zero_condition_II}
\sum_{y \notin \Im(L'_2)} \alpha_{y, b, L_1, R_1, L'_2  \cup \set{(x,y)}, R_2} = 0.
\end{equation}

\myparagraph{Computing $F^{L,\dagger} \Ospru \ket{\psi}$}
Next, we will compute $F^{L,\dagger} \Ospru \ket{\psi}$. Firstly, by~\Cref{lem:Ospru}, we obtain
\begin{align*}
\ket{\psi}
\xmapsto{\Ospru} 
\sum_{  \substack{
    y, b, \\
    L_1, L_2, R_1, R_2 \\
    (\bfk, \bfz) \in \Sspru\left(\substack{L_1, L_2 \\ R_1, R_2}\right)
}   }
\frac{\alpha_{y,b,L_1,R_1,L_2,R_2}}{\sqrt{N^{|L_2|+|R_2|+3}}}
\ket{y}_{\reg{A}} 
\ket{b}_{\reg{B}} 
\ket{L_1^{(k_1,k_3)} \cup L_2^{(k_2,\vec{z}_L)}}_{\reg{S}}
\ket{R^{(k_1,k_3)}_1 \cup R_2^{(k_2,\vec{z}_R)}}_{\reg{T}} 
\ket{\bfk}_{\reg{K}}.
\end{align*}
Before we move on to apply $F^{L,\dagger}$, recall~\Cref{eq:FL_dagger}. For the expression to be nonzero, it is necessary that $y \oplus k_3 \in L^{(k_1,k_3)}_1 \cup L^{(k_2,\vec{z}_L)}_2$. By~\Cref{lem:conditions_robust_decodability}, $L^{(k_1,k_3)}_1$ and $L^{(k_2,\vec{z}_L)}_2$ are disjoint. Therefore, $(x, y)$ must belong to exactly one of the following: (i) $(x, y) \in L^{(k_1,k_3)}_1$, (ii) $(x, y) \in L^{(k_2,\vec{z}_L)}_2$. Define two projectors
\begin{align*}
    &\Pi_1 \coloneqq 
    \sum_{  \substack{ 
        (L,\bfk) \colon G^{\ell}_{L,k_2} \text{ is decomposable} \\ 
        y \in \Im(V_{\isolate}(G^{\ell}_{L,k_2}) \cup V_{\source}(G^{\ell}_{L,k_2}))
    } } 
    \ketbra{y}{y}_{\reg{A}} \otimes \ketbra{L}{L}_{\reg{S}} \otimes \proj{\bfk}_{\reg{K}},\\
    &\Pi_2 \coloneqq 
    \sum_{  \substack{ 
        (L,\bfk) \colon G^{\ell}_{L,k_2} \text{ is decomposable} \\ 
        y \in \Im(V_{\target}(G^{\ell}_{L,k_2}))
    } } 
    \ketbra{y}{y}_{\reg{A}} \otimes \ketbra{L}{L}_{\reg{S}} \otimes \proj{\bfk}_{\reg{K}}
\end{align*}
Thus, it holds that 
\begin{align*}
F^{L,\dagger} X^{k_3} \Ospru \ket{\psi} 
= F^{L,\dagger} \Pi_1 X^{k_3} \Ospru \ket{\psi} 
+ F^{L,\dagger} \Pi_2 X^{k_3} \Ospru \ket{\psi}.    
\end{align*}
By the triangle inequality, it suffices to the bound the norm of each term.

\myparagraph{Bounding $F^{L,\dagger} \Pi_1 \Ospru \ket{\psi}$}

Recall~\Cref{lem:punc_S}. For any $(y,L_1,R_1,L_2,R_2)$, define the set
\begin{align*}
    \cP_{y,L_1,R_1,L_2,R_2} \coloneqq \set{(\bfk,\bfz) \colon y \in \Im(L^{(k_1,k_3)}_1 \cup L_{2,\source}^{(k_2,\vec{z}_L)})}.
\end{align*}
By respectively sampling $k_3$ and $\vec{z}_L$ at the end and the union bound, it is clear that $\cP_{y,L_1,R_1,L_2,R_2}$ occupies at most a $2t/N$ fraction of its universe. Thus, we obtain
\begin{align*}
    \|F^{L,\dagger} \Pi_1 \Ospru \ket{\psi}\|_2
    & \le \|F^{L,\dagger} \Pi_1 \Ospru^{\bullet} \ket{\psi}\|_2
    + \|F^{L,\dagger} \Pi_1 (\Ospru - \Ospru^{\bullet}) \ket{\psi}\|_2\\
    & \le \|\Pi_1 \Ospru^{\bullet} \ket{\psi}\|_2
    + \|(\Ospru^{\bullet} - \Ospru)\Pi_{\le t}\|_{\op} \\
    & \le O(\sqrt{t/N}).
\end{align*}
The first term is zero by the definitions of $\set{\cP_\tau}_\tau$ and $\Ospru^{\bullet}$. The second term is bounded by~\Cref{lem:punc_S}

\myparagraph{Bounding $F^{L,\dagger} \Pi_2 \Ospru \ket{\psi}$}
Using~\Cref{eq:FL_dagger}, we obtain
\begin{align*}
&F^{L,\dagger} \Pi_2 \Ospru \ket{\psi} \\
&= \sum_{ \substack{
    b, L_1, L_2, R_1, R_2, \\
    (\bfk, \bfz) \in \Sspru\left(\substack{L_1, L_2 \\ R_1, R_2}\right) \\
    (x, y) \in L^{(k_2,\vec{z}_L)}_{2,\target}
} }
\frac{\alpha_{y,b,L_1,R_1,L_2,R_2}}{\sqrt{N^{|L_2|+|R_2|+4}}}
\ket{x}_{\reg{A}} \ket{b}_{\reg{B}} 
\ket{L^{(k_1,k_3)}_1 \cup L^{(k_2,\vec{z}_L)}_2 \setminus \set{(x, y)}}_{\reg{S}} 
\ket{R^{(k_1,k_3)}_1 \cup R^{(k_2,\vec{z}_R)}_2}_{\reg{T}} 
\ket{\bfk}_{\reg{K}}.
\end{align*}

\noindent By substituting $L_2 = L'_2 \cup \set{(x',y)}$ and $x = z_{L,i} \oplus k_2$, where $i$ is the index such that $y \, \in_i \Im(L'_2) \cup \set{y}$, we obtain
\begin{align*}
&\sum_{ \substack{
    b, L_1, \rcolor{L'_2}, R_1, R_2 \\
    \rcolor{x', y \notin \Im(L'_2)} \\
    (\bfk, \bfz) \in \Sspru\left(\substack{L_1, \rcolor{L'_2 \cup \set{(x,y)}} \\ R_1, R_2}\right) \\
    \rcolor{i \st y \, \in_i \Im(L'_2) \cup \set{y}}
} }
\frac{\alpha_{y, b, L_1, R_1, \rcolor{L'_2 \cup \set{(x',y)}}, R_2}}{\sqrt{N^{\rcolor{(|L'_2|+1)} + |R_2| + 4}}} \ket{\rcolor{z_{L,i} \oplus k_2}}_{\reg{A}} \ket{b}_{\reg{B}} \\
&\hspace{.2\textwidth} \otimes \ket{L^{(k_1,k_3)}_1 \cup (\rcolor{L'_2 \cup \set{(x',y)}})^{(k_2,\vec{z}_L)} \setminus \set{(\rcolor{z_{L,i} \oplus k_2}, y)}}_{\reg{S}} 
\ket{R^{(k_1,k_3)}_1 \cup R^{(k_2,\vec{z}_R)}_2}_{\reg{T}} 
\ket{\bfk}_{\reg{K}}.
\end{align*}

\noindent By~\Cref{lem:good_tuples_satisfy_conditions}, the state is in the domain of the partial isometry $\Osprus$. Thus, by~\Cref{eq:robust_target}, we can instead calculate the norm of the following state:
\begin{align*}
&\xmapsto{\Osprus}
\sum_{ \substack{
    b, L_1, L'_2, R_1, R_2 \\
    x', y \notin \Im(L'_2) \\
    (\bfk, \bfz) \in \Sspru\left(\substack{L_1, L'_2 \cup \set{(x,y)} \\ R_1, R_2}\right) \\
    i \st y \, \in_i \Im(L'_2) \cup \set{y}
} }
\frac{\alpha_{y, b, L_1, R_1, L'_2 \cup \set{(x',y)}, R_2}}{\sqrt{N^{|L'_2| + |R_2| + 5}}} 
\ket{z_{L,i} \oplus k_2}_{\reg{A}} 
\ket{b}_{\reg{B}} \\
&\hspace{.3\textwidth} \otimes 
\ket{L_1 \cup \set{(x' \oplus k_1, z_{L,i} \oplus k_3)}}_{\reg{S}_1}
\ket{R_1}_{\reg{T}_1}
\ket{L'_2}_{\reg{S}_2}
 \ket{R_2}_{\reg{T}_2}
 \ket{\vec{z}_{L,-i}}_{\reg{Z_L}} 
 \ket{\vec{z}_R}_{\reg{Z_R}} 
\ket{\bfk}_{\reg{K}}.
\end{align*}

\noindent We now apply the following sequence of partial isometries to simply the expression without changing the norm:
\begin{align*}
&\xmapsto{F^L_{\extract} \cdot X^{k_3} \cdot X^{k_2}} \\
&\sum_{ \substack{
    b, L_1, L'_2, R_1, R_2 \\
    x', y \notin \Im(L'_2) \\
    (\bfk, \bfz) \in \Sspru\left(\substack{L_1, L'_2 \cup \set{(x,y)} \\ R_1, R_2}\right) \\
    i \st y \, \in_i \Im(L'_2) \cup \set{y}
} }
\frac{\alpha_{y, b, L_1, R_1, L'_2 \cup \set{(x',y)}, R_2}}{\sqrt{N^{|L'_2| + |R_2| + 5}}} 
\ket{\rcolor{z_{L,i} \oplus k_3}}_{\reg{A'}} 
\ket{\rcolor{x' \oplus k_1}}_{\reg{A}}
\ket{b}_{\reg{B}}
\ket{\rcolor{L_1}}_{\reg{S}_1}
\ket{R_1}_{\reg{T}_1}
\ket{L'_2}_{\reg{S}_2}
 \ket{R_2}_{\reg{T}_2}
 \ket{\vec{z}_{L,-i}}_{\reg{Z_L}} 
 \ket{\vec{z}_R}_{\reg{Z_R}} 
\ket{\bfk}_{\reg{K}} \\
&\xmapsto{X^{k_3} \otimes X^{k_1}} \\
&\sum_{ \substack{
    b, L_1, L'_2, R_1, R_2 \\
    x', y \notin \Im(L'_2) \\
    (\bfk, \bfz) \in \Sspru\left(\substack{L_1, L'_2 \cup \set{(x,y)} \\ R_1, R_2}\right) \\
    i \st y \, \in_i \Im(L'_2) \cup \set{y}
} }
\frac{\alpha_{y, b, L_1, R_1, L'_2 \cup \set{(x',y)}, R_2}}{\sqrt{N^{|L'_2| + |R_2| + 5}}} 
\ket{\rcolor{z_{L,i}}}_{\reg{A'}} 
\ket{\rcolor{x'}}_{\reg{A}}
\ket{b}_{\reg{B}}
\ket{L_1}_{\reg{S}_1}
\ket{R_1}_{\reg{T}_1}
\ket{L'_2}_{\reg{S}_2}
\ket{R_2}_{\reg{T}_2}
\ket{\vec{z}_{L,-i}}_{\reg{Z_L}} 
\ket{\vec{z}_R}_{\reg{Z_R}} 
\ket{\bfk}_{\reg{K}}.
\end{align*}
By substituting $\vec{q}_L = \vec{z}_{L,-i}$, $z = z_i$, and $\vec{z}_L = \vec{q}^{\,(i \gets z)}_L$, we obtain
\begin{align*}
\sum_{ \substack{
    b, L_1, L'_2, R_1, R_2 \\
    x', y \notin \Im(L'_2) \\
    i \st y \, \in_i \Im(L'_2) \cup \set{y} \\
    \bfk, \rcolor{\vec{q}_L \in [N]^{|L'_2|-1}},\, \vec{z}_R,\, \rcolor{z \in [N]:} \\
    (\bfk, \rcolor{\vec{q}^{\,(i \gets z)}_L}, \vec{z}_R) \in \Sspru\left(\substack{L_1, L'_2 \cup \set{(x,y)} \\ R_1, R_2}\right)
} }
\frac{\alpha_{y, b, L_1, R_1, L'_2 \cup \set{(x',y)}, R_2}}{\sqrt{N^{|L'_2| + |R_2| + 5}}} 
\ket{\rcolor{z}}_{\reg{A'}} 
\ket{x'}_{\reg{A}}
\ket{b}_{\reg{B}}
\ket{L_1}_{\reg{S}_1}
\ket{R_1}_{\reg{T}_1}
\ket{L'_2}_{\reg{S}_2}
\ket{R_2}_{\reg{T}_2}
\ket{\rcolor{\vec{q}_L}}_{\reg{Z_L}} 
\ket{\vec{z}_R}_{\reg{Z_R}} 
\ket{\bfk}_{\reg{K}}.
\end{align*}
We may compute its squared norm:
\begin{align*}
& \sum_{    \substack{
    b, L_1, L'_2, R_1, R_2, x' \\
    \bfk,\, \vec{q}_L \in [N]^{|L'_2|-1},\, \vec{z}_R,\, z \in [N]
}   }
\frac{1}{N^{|L'_2| + |R_2| + 5}} \cdot
\Biggl|
\sum_{  \substack{
    y \notin \Im(L'_2): \\
    i \st y \, \in_i \Im(L'_2) \cup \set{y} \\
    (\bfk, \vec{q}^{\,(i \gets z)}_L, \vec{z}_R) \in \Sspru\left(\substack{L_1, L'_2 \cup \set{(x,y)} \\ R_1, R_2}\right)
}   }
\alpha_{y, b, L_1, R_1, L'_2 \cup \set{(x',y)}, R_2}
\Biggr|^2.
\end{align*}
By~\Cref{lem:comb_F2Ldagger_y_difference}, once $(L'_1, L_2, R_1, R_2, x', \bfk, \vec{q}, \vec{z}_R, z)$ is fixed, there is either zero or at least $N - g(t)$ values of $y$ that satisfy the condition, where $g(t) = O(t)$ is some function guaranteed by~\Cref{lem:comb_F2Ldagger_y_difference}. Let $\mathsf{BAD}$ denote the set of tuples for which the latter case holds. We obtain
\begin{align*}
& \sum_{b, \rcolor{(L_1, L'_2, R_1, R_2, x', \bfk, \vec{q}_L, \vec{z}_R, z) \in \mathsf{Bad}}}
\frac{1}{N^{|L'_2| + |R_2| + 5}} \cdot
\Biggl|
\sum_{  \substack{
    y \notin \Im(L'_2): \\
    i \st y \, \in_i \Im(L'_2) \cup \set{y} \\
    (\bfk, \vec{q}^{\,(i \gets z)}_L, \vec{z}_R) \in \Sspru\left(\substack{L_1, L'_2 \cup \set{(x,y)} \\ R_1, R_2}\right)
}   }
\alpha_{y, b, L_1, R_1, L'_2 \cup \set{(x',y)}, R_2}
\Biggr|^2.
\end{align*}
Now, we make crucial use of the condition implied by the premise (\Cref{eq:zero_condition_II}) to obtain
\begin{align*}
\sum_{b, (L_1, L'_2, R_1, R_2, x', \bfk, \vec{q}_L, \vec{z}_R, z) \in \mathsf{Bad}}
\frac{1}{N^{|L'_2| + |R_2| + 5}} \cdot
\Biggl| \rcolor{-}
\sum_{  \substack{
    y \notin \Im(L'_2): \\
    i \st y \, \in_i \Im(L'_2) \cup \set{y} \\
    (\bfk, \vec{q}^{\,(i \gets z)}_L, \vec{z}_R) \rcolor{\notin} \Sspru\left(\substack{L_1, L'_2 \cup \set{(x,y)} \\ R_1, R_2}\right)
}   }
\alpha_{y, b, L_1, R_1, L'_2 \cup \set{(x',y)}, R_2}
\Biggr|^2.
\end{align*}
Using the Cauchy-Schwarz inequality and the definition of $\mathsf{BAD}$ to bound the number of $y$ in the sum, we can bound it by
\begin{align*}
\sum_{  \substack{
    b, (L_1, L'_2, R_1, R_2, x', \bfk, \vec{q}_L, \vec{z}_R, z) \in \mathsf{Bad}, y \notin \Im(L'_2): \\
    i \st y \, \in_i \Im(L'_2) \cup \set{y} \\
    (\bfk, \vec{q}^{\,(i \gets z)}_L, \vec{z}_R) \notin \Sspru\left(\substack{L_1, L'_2 \cup \set{(x,y)} \\ R_1, R_2}\right)
}   }
\frac{g(t)}{N^{|L'_2| + |R_2| + 5}} \cdot
\Biggl|\alpha_{y, b, L_1, R_1, L'_2 \cup \set{(x',y)}, R_2}\Biggr|^2.
\end{align*}
Since we are summing over non-negative terms, by relaxing the constraints, we can bound it by
\begin{align*}
\sum_{  \substack{
    b, L_1, L'_2, R_1, R_2, x', \bfk, \vec{q}_L, \vec{z}_R, z), y \notin \Im(L'_2)
}   }
\frac{g(t)}{N^{|L'_2| + |R_2| + 5}} \cdot
\Biggl|\alpha_{y, b, L_1, R_1, L'_2 \cup \set{(x',y)}, R_2}\Biggr|^2.
\end{align*}
By summing over $(\bfk, \vec{q}_L, \vec{z}_R, z)$, and noting that there are at most $N^{|L'_2|+|R_2|+4}$ such tuples, we can bounded it by
\begin{align*}
\sum_{  \substack{
    b, x', L_1, L'_2, R_1, R_2, y \notin \Im(L'_2)
}   }
\frac{g(t)}{N} \cdot
\Biggl|\alpha_{y, b, L_1, R_1, L'_2 \cup \set{(x',y)}, R_2}\Biggr|^2.
\end{align*}
By substituting $L = L'_2 \cup \set{(x',y)}$, we obtain
\begin{align*}
\frac{g(t)}{N} \cdot 
\sum_{b, L_1, \rcolor{L_2}, R_1, R_2, \rcolor{y \in \Im(L_2)}}
\Bigl|\alpha_{y, b, L_1, R_1, \rcolor{L_2}, R_2}\Bigr|^2 
= O(t/N)
\end{align*}
by the normalization condition of $\ket{\psi}$.
\end{proof}

\begin{lemma}[Closeness of $F_2^{L,\dagger}$ and $F_2^{R,\dagger}$]   \label{lem:fl2dagger_fr2dagger}
For any integer $t \geq 0$,
\begin{align*}
& \|(F^{L,\dagger} X^{k_2} F^{L,\dagger} \Ospru - \Ospru F_2^{L,\dagger}) \Pi_{\le t}\|_{\op} 
= O( \sqrt{t/N} ) \\
& \|(F^{R,\dagger} X^{k_2} F^{R,\dagger} \Ospru - \Ospru F_2^{R,\dagger}) \Pi_{\le t}\|_{\op} 
= O( \sqrt{t/N} ).
\end{align*}
\end{lemma}

\begin{proof}
Let $\Pi^{\Im(F_2^L)}$ denote the projection onto the image of $F_2^L$. For an arbitrary state $\ket{\psi}$ in the subspace of $\Pi_{\le t}$, we can decompose it as
\begin{align*}
    \ket{\psi} = \Pi^{\Im(F_2^L)} \ket{\psi} + (\id - \Pi^{\Im(F_2^L)}) \ket{\psi}.
\end{align*}
We will show the following two bounds
\begin{align}
&\|(F^{L,\dagger} X^{k_2} F^{L,\dagger} \Ospru - \Ospru F_2^{L,\dagger}) \Pi^{\Im(F_2^L)} \ket{\psi}\|_2 \le O( \sqrt{t/N} )
\label{eq:psi_FL2_Im} \\
&\|(F^{L,\dagger} X^{k_2} F^{L,\dagger} \Ospru - \Ospru F_2^{L,\dagger}) (\id - \Pi^{\Im(F_2^L)}) \ket{\psi}\|_2 \le O( \sqrt{t/N} ).
\label{eq:psi_FL2_Not_Im}
\end{align}
which then complete the proof by the triangle inequality. Notice that $F_2^{L, \dagger} (\id - \Pi^{\Im(F_2^L)}) = 0$. Thus, \Cref{eq:psi_FL2_Not_Im} can be bounded as follows:
\begin{align}
& \|(F^{L,\dagger} X^{k_2} F^{L,\dagger} \Ospru - \Ospru F_2^{L,\dagger}) (\id - \Pi^{\Im(F_2^L)}) \ket{\psi}\|_2
\notag \\
= & \|F^{L,\dagger} X^{k_2} F^{L,\dagger} \Ospru (\id - \Pi^{\Im(F_2^L)}) \ket{\psi}\|_2
\notag \\
\leq & \|F^{L,\dagger} \Ospru (\id - \Pi^{\Im(F_2^L)}) \ket{\psi}\|_2
\tag{by~\Cref{lem:op_norm}} \\
\leq & O(\sqrt{t/N}).
\tag{by~\Cref{lem:fl2dagger_fr2dagger:zero}}
\end{align}
Hence, it suffices to bound~\Cref{eq:psi_FL2_Im}. Since $\Pi^{\Im(F_2^L)} \ket{\psi}$ is in the image of $F_2^L$, there exists some state $\ket{\phi}$ such that $\Pi^{\Im(F_2^L)} \ket{\psi} = F_2^L \ket{\phi}$. Now. we bound~\Cref{eq:psi_FL2_Im} by the triangle inequality as follows:
\begin{align}
& \|(F^{L,\dagger} X^{k_2} F^{L,\dagger} \Ospru - \Ospru F_2^{L,\dagger}) \Pi^{\Im(F_2^L)} \ket{\psi}\|_2 
\notag \\
= & \|(F^{L,\dagger} X^{k_2} F^{L,\dagger} \Ospru - \Ospru F_2^{L,\dagger}) F_2^L \ket{\phi}\|_2 
\notag \\
\le & \|(F^{L,\dagger} X^{k_2} F^{L,\dagger} \Ospru F_2^L - \rcolor{F^{L,\dagger} X^{k_2} F^{L,\dagger} F^L X^{k_2} F^L \Ospru}) \ket{\phi}\|_2 
\label{eq:term1} \\
&\hspace{.2\textwidth} + \|(\rcolor{F^{L,\dagger} X^{k_2} F^{L,\dagger} F^L X^{k_2} F^L \Ospru} - \Ospru F_2^{L,\dagger} F_2^L) \ket{\phi}\|_2.
\label{eq:term2}
\end{align}
We use~\Cref{lem:fl2_fr2} and that operator norm is submultiplicative to bound~\Cref{eq:term1}. Finally, we use~\Cref{lem:domain_FL_FR} to bound~\Cref{eq:term2} by replacing $F^{L,\dagger} F^L$ and $F_2^{L,\dagger} F_2^L$ with the identity. This completes the proof.
\end{proof}

\noindent We need the following corollaries for the next subsection. The structure of the proof is similar to that of~\Cref{lem:fl2dagger_fr2dagger}. We sketch the proof below and omit the details.

\begin{corollary}[Closeness of $F_2^L F_2^{L,\dagger}$ and $F_2^R F_2^{R,\dagger}$]   \label{cor:spru:fl2proj}
For any integer $t \ge 0$,
\[
\|(F^L F^{L,\dagger} \Ospru - \Ospru F_2^L F_2^{L,\dagger}) \Pi_{\le t}\|_{\op}
= O(\sqrt{t/N})
\quad \text{and} \quad
\|(F^R F^{R,\dagger} \Ospru - \Ospru F_2^R F_2^{R,\dagger}) \Pi_{\le t}\|_{\op}
= O(\sqrt{t/N}).
\]
\end{corollary}

\begin{proof}[Proof sketch]
Let $\Pi^{\Im(F_2^L)} \ket{\psi} = F_2^L \ket{\phi}$. We will show that
\begin{align*}
    F^L F^{L,\dagger} \Ospru \Pi^{\Im(F_2^L)} \ket{\psi} 
    \approx \Ospru F_2^L F_2^{L,\dagger} \Pi^{\Im(F_2^L)} \ket{\psi}.
\end{align*}
For the left-hand side, consider the following sequence of hybrids:
\begin{align*}
& F^L F^{L,\dagger} \Ospru \Pi^{\Im(F_2^L)} \ket{\psi}
\notag \\
= & F^L F^{L,\dagger} \Ospru F_2^L \ket{\phi}
\notag \\
\approx & F^L F^{L,\dagger} F^L X^{k_2} F^L \Ospru \ket{\phi}
\tag{$\Ospru F_2^L \approx F^L X^{k_2} F^L \Ospru$ by~\Cref{lem:fl2_fr2}} \\
\approx & F^L X^{k_2} F^L \Ospru \ket{\phi}.
\tag{$F_2^{L,\dagger} F_2^L \approx \id$ by~\Cref{lem:domain_FL_FR}}
\end{align*}
For the right-hand side, consider the following sequence of hybrids:
\begin{align*}
& \Ospru F_2^L F_2^{L,\dagger} \Pi^{\Im(F_2^L)} \ket{\psi}
\notag \\
= & \Ospru F_2^L F_2^{L,\dagger} F_2^L \ket{\phi}
\notag \\
\approx & \Ospru F_2^L \ket{\phi}
\tag{$F_2^{L,\dagger} F_2^L \approx \id$ by~\Cref{lem:domain_FL_FR}} \\
\approx & F^L X^{k_2} F^L \Ospru \ket{\phi}.
\tag{$\Ospru F_2^L \approx F^L X^{k_2} F^L \Ospru$ by~\Cref{lem:fl2_fr2}}
\end{align*}
This completes the proof.
\end{proof}

\subsection{Putting Things Together}
We use the above lemmas to prove~\Cref{lem:spru:2nd_oracle}. We restate~\Cref{lem:spru:2nd_oracle} for convenience.
\begin{lemma}[\Cref{lem:spru:2nd_oracle}, restated]
For any integer $t \geq 0$,
\begin{itemize}
    \item {\bf Forward query:} $\|(F X^{k_2} F \Ospru - \Ospru F_2) \Pi_{\leq t}\|_{\op} \leq O(t/\sqrt{N})$,
    \item {\bf Inverse query:} $\|(F^{\dagger} X^{k_2} F^{\dagger} \Ospru - \Ospru F_2^{\dagger}) \Pi_{\leq t}\|_{\op} \leq O(t/\sqrt{N})$.
\end{itemize}
\end{lemma}

\begin{proof}[Proof of~\Cref{lem:spru:2nd_oracle}] 
\ \\
\myparagraph{Forward query:}
Recall the definition of $F_2$:
\begin{align}
F_2 
& = F_2^L \cdot (\id - F_2^R \cdot F_2^{R,\dagger}) + (\id - F_2^L \cdot F_2^{L,\dagger}) \cdot F_2^{R,\dagger} 
\notag \\
& = F_2^L - F_2^L \cdot F_2^R \cdot F_2^{R,\dagger} + F_2^{R,\dagger} - F_2^L \cdot F_2^{L,\dagger} \cdot F_2^{R,\dagger}.
\label{eq:F2_expand}
\end{align}
We expand $F X^{k_2} F$ in the following way:
\begin{align*}
& F \cdot X^{k_2} \cdot F \\
& = \Big(F^L + \underbrace{(\id - F^L \cdot F^R - F^L \cdot F^{L,\dagger}) \cdot F^{R,\dagger}}_{A}\Big)
\cdot X^{k_2} \cdot
\Big(F^{R,\dagger} +  \underbrace{F^L \cdot (\id - F^{R,\dagger} \cdot F^R - F^{L,\dagger} \cdot F^{R,\dagger})}_{B}\Big) \\
& = F^L X^{k_2} F^{R,\dagger} + F^L X^{k_2} B + A X^{k_2} F^{R,\dagger} + A B.
\end{align*}
Here, the term $AB$ can be viewed as a negligibly small error. Since there is $F^{R,\dagger} X^{k_2} F^L$ in the middle of $AB$, its operator norm is at most $O(t/\sqrt{N})$ by~\Cref{lem:FLdagger_U_FR:zero}. Thus, it suffices to show the closeness of the remaining terms. We expand and arrange them in the following way:
\begin{align*}
& F^L X^{k_2} F^{R,\dagger} + F^L X^{k_2} B + A X^{k_2} F^{R,\dagger} \\
= & F^L X^{k_2} F^L - F^L X^{k_2} F^L F^R F^{R,\dagger} + F^{R,\dagger} X^{k_2} F^{R,\dagger} - F^L F^{L,\dagger} F^{R,\dagger} X^{k_2} F^{R,\dagger} \\
&\hspace{.3\textwidth} - F^L X^{k_2} F^L F^{L,\dagger} F^{R,\dagger}
+ F^L (\id - F^R F^{R,\dagger} ) X^{k_2} F^{R,\dagger}.
\end{align*}
In what follows, we will show each term in the second line is negligibly close to a corresponding term in~\Cref{eq:F2_expand} in operator norm, and both terms in the third line have negligibly small operator norms. Concretely, we have the following claims, which together imply the lemma:
\begin{enumerate}
    \item $\|(F^L X^{k_2} F^L \Ospru - \Ospru F_2^L) \Pi_{\le t}\|_{\op} \le O(\sqrt{t/N})$
    \item $\|(F^L X^{k_2} F^L F^R F^{R,\dagger} \Ospru - \Ospru F_2^L F_2^R F_2^{R,\dagger}) \Pi_{\le t}\|_{\op} \le O(\sqrt{t/N})$
    \item $\|(F^{R,\dagger} X^{k_2} F^{R,\dagger} \Ospru - \Ospru F_2^{R,\dagger}) \Pi_{\le t}\|_{\op} \le O(\sqrt{t/N})$
    \item $\|(F^L F^{L,\dagger} F^{R,\dagger} X^{k_2} F^{R,\dagger} \Ospru - \Ospru F_2^L F_2^{L,\dagger} F_2^{R,\dagger}) \Pi_{\le t}\|_{\op} \le O(\sqrt{t/N})$
    \item $\|F^L X^{k_2} F^L F^{L,\dagger} F^{R,\dagger} \Ospru \Pi_{\le t}\|_{\op} \le O(t/\sqrt{N})$
    \item $\|F^L (\id - F^R F^{R,\dagger} ) X^{k_2} F^{R,\dagger} \Ospru \Pi_{\le t}\|_{\op} \le O(\sqrt{t/N})$
\end{enumerate}
Items~1 and~3 immediately follow from~\Cref{lem:fl2_fr2,lem:fl2dagger_fr2dagger}, respectively. To prove Item~2, we use the following sequence of hybrids to sketch the proof:
\begin{align*}
& F^L X^{k_2} F^L F^R F^{R,\dagger} \Ospru \\
\approx & F^L X^{k_2} F^L \rcolor{\Ospru F_2^R F_2^{R,\dagger}} 
\tag{$F^L F^R F^{R,\dagger} \Ospru \approx \Ospru F_2^R F_2^{R,\dagger}$ by~\Cref{cor:spru:fl2proj}} \\
\approx & \rcolor{\Ospru F_2^L} F_2^R F_2^{R,\dagger}.
\tag{$F^L X^{k_2} F^L \Ospru \approx \Ospru F_2^L$ by~\Cref{lem:fl2_fr2}}
\end{align*}
Item~4 can be proven in a similar way using~\Cref{lem:fl2dagger_fr2dagger} and~\Cref{cor:spru:fl2proj}. To prove Item~5, we use a similar idea as in the proof of~\Cref{lem:fl2dagger_fr2dagger}. Suppose $\Pi^{\Im(F_2^R)} \ket{\psi} = F_2^R \ket{\phi}$. We can bound it as follows:
\begin{align*}
& \|F^L X^{k_2} F^L F^{L,\dagger} F^{R,\dagger} \Ospru \Pi^{\Im(F_2^R)} \ket{\psi}\|_2 \\
= & \|F^L X^{k_2} F^L F^{L,\dagger} F^{R,\dagger} \Ospru F_2^R \ket{\phi}\|_2 \\
\le & \|F^L X^{k_2} F^L F^{L,\dagger} F^{R,\dagger} \rcolor{F^R X^{k_2} F^R \Ospru} \ket{\phi}\|_2 + O(\sqrt{t/N})
\tag{$\Ospru F_2^R \approx F^R X^{k_2} F^R$ by~\Cref{lem:fl2_fr2}} \\
\le & \|F^L X^{k_2} F^L F^{L,\dagger} X^{k_2} F^R \Ospru \ket{\phi}\|_2 + O(\sqrt{t/N})
\tag{$F^{R,\dagger} F^R \approx \id$ by~\Cref{lem:domain_FL_FR}} \\
\le & \|F^{L,\dagger} X^{k_2} F^R \|_{\op} + O(\sqrt{t/N})
\tag{by~\Cref{lem:op_norm}}\\
\le & O(t/\sqrt{N}).
\tag{by~\Cref{lem:FLdagger_U_FR:zero}}
\end{align*}
To prove Item~6, we use the same decomposition on $\ket{\psi}$ and obtain:
\begin{align*}
& \|F^L (\id - F^R F^{R,\dagger} ) X^{k_2} F^{R,\dagger} \Ospru \Pi^{\Im(F_2^R)} \ket{\psi}\|_2 \\
= & \|F^L (\id - F^R F^{R,\dagger} ) X^{k_2} F^{R,\dagger} \Ospru F_2^R \ket{\phi}\|_2 \\
\le & \|F^L (\id - F^R F^{R,\dagger} ) X^{k_2} F^{R,\dagger} \rcolor{F^R X^{k_2} F^R \Ospru} \ket{\phi}\|_2 + O(\sqrt{t/N}) 
\tag{$\Ospru F_2^R \approx F^R X^{k_2} F^R$ by~\Cref{lem:fl2_fr2}} \\
\le & \|F^L (\id - F^R F^{R,\dagger} ) F^R \Ospru \ket{\phi}\|_2 + O(\sqrt{t/N}) 
\tag{$F^{R,\dagger} F^R \approx \id$ by~\Cref{lem:domain_FL_FR}} \\
\le & O(\sqrt{t/N}).
\tag{$F^{R,\dagger} F^R \approx \id$ by~\Cref{lem:domain_FL_FR}}
\end{align*}
This completes the proof.
\end{proof}

\section*{Acknowledgments}
PA, AG and YTL are supported by the National Science Foundation under the grants FET-2329938, CAREER-2341004 and, FET-2530160.

\printbibliography

\newpage
\appendix
\section{Path-Recording with Independent Left and Right Operators}

For the strong path-recording isometry defined in~\cite{MH25}, the left (resp. right) isometry $V^{L}$ (resp. $V^{R}$) outputs strings $y$ (resp. $x$) that are in in the image (resp. domain) of \emph{both} the left and right relation states in the purifying register.  In our case, however, it will be easier to work with a similar pair of isometries that only look at their own relation state, as opposed to both relations taken together.  In this section we define a pair of new isometries, $F^{L}$ and $F^{R}$ that are independent of each other. 

\subsection{Closeness to the Path-Recording Isometries} \label{sec:FvsV}

\begin{lemma}   \label{lem:FL_VL_FR_VR_close}
For any integer $t \geq 0$,
\[
\|(V^L-F^L) \Pi_{\leq t}\|_{\op} \leq \sqrt{\frac{t(t+2)}{N}}
\quad \text{and} \quad
\|(V^R-F^R) \Pi_{\leq t}\|_{\op} \leq \sqrt{\frac{t(t+2)}{N}}.
\]
\end{lemma}

\begin{proof}
Consider an arbitrary (normalized) state in the support of $\Pi_{\leq t}$
\[
\ket{\psi}_{\reg{AST}} = \sum_{x,L,R} \alpha_{x,L,R} \ket{x}_{\reg{A}} \ket{L}_{\reg{S}} \ket{R}_{\reg{T}},
\]
where $\alpha_{x,L,R} = 0$ whenever $|L|$ or $|R| > t$. Then we expand out
\begin{align*}
V^L \ket{\psi}_{\reg{AST}}
&= \sum_{x,L,R} \frac{\alpha_{x,L,R}}{\sqrt{N-|\Im(L\cup R)|}}
   \sum_{y \notin \Im(L \cup R)}
   \ket{y}_{\reg A}\ket{L \cup \{(x,y)\}}_{\reg S}\ket{R}_{\reg T}, \; \text{and} \\
F^L \ket{\psi}_{\reg{AST}}
&= \sum_{x,L,R} \frac{\alpha_{x,L,R}}{\sqrt{N}}
   \sum_{y \notin \Im(L)}
   \ket{y}_{\reg A}\ket{L \cup \{(x,y)\}}_{\reg S}\ket{R}_{\reg T}.
\end{align*}
Subtracting, 
\begin{align*}
(V^L & - F^L) \ket{\psi}_{\reg{AST}} \\
& = \underbrace{\sum_{x,L,R} \alpha_{x,L,R} \sum_{y \notin \Im(L\cup R)} \ket{y}_{\reg{A}} \ket{L \cup \set{(x,y)}}_{\reg{S}} \ket{R}_{\reg{T}} \left(\frac{1}{\sqrt{N-|\Im(L\cup R)|}}-\frac{1}{\sqrt{N}}\right)}_{\ket{v}} \\
& + \underbrace{\sum_{x,L,R} \alpha_{x,L,R} \sum_{y \in \Im(R) \setminus \Im(L)} \ket{y}_{\reg{A}} \ket{L \cup \set{(x,y)}}_{\reg{S}} \ket{R}_{\reg{T}} \left(-\frac{1}{\sqrt{N}}\right)}_{\ket{w}}.
\end{align*}
Note that $\ket{w}$ and $\ket{v}$ are orthogonal, since $\ket{v}$ is a superposition of states $\ket{y}\ket{L'}\ket{R}$ where $y$ appears exactly once in $\Im(L')$ and does not appear in $\Im(R)$, while $\ket{w}$ is a superposition of states $\ket{y}\ket{L'}\ket{R}$ where $y$ appears exactly once in both $\Im(L')$ and $\Im(R)$. Thus,
\[
\norm{(V^L - F^L)\ket{\psi}}^2_2 = \braket{v}{v} + \braket{w}{w}.
\]

\myparagraph{Bounding $\braket{v}{v}$} Similar to~\cite{MH25}, by changing the order of summation, we can rewrite $\ket{v}$ as
    \begin{align*}
        \ket{v} = \sum_{\substack{y,L',R}} \ket{y} \ket{L'} \ket{R} \left( \sum_{\substack{(x,L): \\ L' = L \cup \{(x,y)\},\\ y\not\in \Im(L \cup R)}} \alpha_{x,L,R} \left(\frac{1}{\sqrt{N - \abs{\Im(L \cup R)}}} - \frac{1}{\sqrt{N}}\right) \right),
    \end{align*}
    and thus
    \begin{align*}
        \braket{v}{v} &= \sum_{\substack{y,L',R}} \left( \sum_{\substack{(x,L):\\ L' = L \cup \{(x,y)\},\\ y\not\in \Im(L \cup R)}} \alpha_{x,L,R} \left(\frac{1}{\sqrt{N - \abs{\Im(L \cup R)}}} - \frac{1}{\sqrt{N}}\right) \right)^2\\
        &\leq \sum_{\substack{y,L',R}} \left(\sum_{\substack{(x,L):\\ L' = L \cup \{(x,y)\},\\ y\not\in \Im(L \cup R)}} \abs{\alpha_{x,L,R}}^2 \right) \cdot \left(\sum_{\substack{(x,L):\\ L' = L \cup \{(x,y)\},\\ y\not\in \Im(L \cup R)}} \left(\frac{1}{\sqrt{N - \abs{\Im(L \cup R)}}} - \frac{1}{\sqrt{N}}\right)^2 \right),
    \end{align*}
    where the last inequality is by Cauchy-Schwarz. We can bound the summand by writing
    \begin{align*}
        \sum_{\substack{(x,L):\\ L' = L \cup \{(x,y)\},\\ y\not\in \Im(L \cup R)}} \left(\frac{1}{\sqrt{N - \abs{\Im(L \cup R)}}} - \frac{1}{\sqrt{N}}\right)^2 &= \sum_{\substack{(x,L):\\ L' = L \cup \{(x,y)\},\\ y\not\in \Im(L \cup R)}} \left(\frac{\sqrt{N} - \sqrt{N - \abs{\Im(L \cup R)}}}{\sqrt{N(N - \abs{\Im(L \cup R)})}}\right)^2\\
        & \leq \sum_{\substack{(x,L):\\ L' = L \cup \{(x,y)\},\\ y\not\in \Im(L \cup R)}} \left(\frac{\sqrt{\abs{\Im(L \cup R)}}}{\sqrt{N(N - \abs{\Im(L \cup R)})}}\right)^2 \tag{since $\sqrt{a} - \sqrt{b} \leq \sqrt{a-b}$ when $a\geq b \geq 0$}\\
        &= \sum_{\substack{(x,L):\\ L' = L \cup \{(x,y)\},\\ y\not\in \Im(L \cup R)}} \frac{\abs{\Im(L \cup R)}}{N(N - \abs{\Im(L \cup R)})}\\
        & \leq \frac{(\abs{L} + 1) \cdot \abs{\Im(L \cup R)}}{N(N - \abs{\Im(L \cup R)})} 
    \end{align*}
    where the last inequality uses the fact that for any fixed $L'$, there are at most $\abs{L} +1$ choices of $(x,L)$ that can satisfy $L' = L \cup \{(x,y)\}$. Thus,
    \begin{align*}
        \braket{v}{v} &\leq \frac{(\abs{L} + 1) \cdot \abs{\Im(L \cup R)}}{N(N - \abs{\Im(L \cup R)})} \cdot \sum_{\substack{y,L',R}} \left(\sum_{\substack{(x,L):\\ L' = L \cup \{(x,y)\},\\ y\not\in \Im(L \cup R)}} \abs{\alpha_{x,L,R}}^2 \right) \\
        &=\frac{(\abs{L} + 1) \cdot \abs{\Im(L \cup R)}}{N(N - \abs{\Im(L \cup R)})} \cdot \sum_{\substack{x,L,R}} \abs{\alpha_{x,L,R}}^2 \cdot \left( \sum_{y \in [N]} \indic(y \not\in \Im(L \cup R)) \right) \\
        &\leq \frac{(\abs{L} + 1) \cdot \abs{\Im(L \cup R)}}{N} \cdot \sum_{\substack{x,L,R}} \abs{\alpha_{x,L,R}}^2 = \frac{(\abs{L} + 1) \cdot \abs{\Im(L \cup R)}}{N}.
    \end{align*}

\myparagraph{Bounding $\braket{w}{w}$} We know that 
    $$\ket{w} = \frac{-1}{\sqrt{N}}\sum_{y,(L',R)}\ket{y}\ket{L'}\ket{R}\sum_{\substack{(x,L):\\ L'=L\cup\set{(x,y)}\\ y\in \Im(L\cup R)\setminus\Im(L)}} \alpha_{x,L,R}$$

    Then 
    \begin{align*}
        \braket{w}{w} &= \frac{1}{N}\sum_{y,(L',R)}\left|\sum_{\substack{(x,L):\\ L'=L\cup\set{(x,y)}\\ y\in \Im(L\cup R)\setminus\Im(L)}} \alpha_{x,L,R}\right|^{2} \leq \frac{1}{N}\sum_{y,(L',R)}\sum_{\substack{(x,L):\\ L'=L\cup\set{(x,y)}\\ y\in \Im(L\cup R)\setminus\Im(L)}} \left|\alpha_{x,L,R}\right|^{2} \\
        &= \frac{1}{N}\sum_{x,L,R} \left|\alpha_{x,L,R}\right|^{2}\left(\sum_{\substack{y\in \Im(L\cup R)\setminus\Im(L)}}1\right) \leq \frac{t}{N} \sum_{x,L,R} \left|\alpha_{x,L,R}\right|^{2} = \frac{t}{N}
    \end{align*}

\noindent Hence, it holds that
\[
\|(V^L - F^L) \Pi_{\leq t}\|_{\op}
\leq \sqrt{\frac{t(t+2)}{N}}.
\]
By a symmetric argument, we have
\[
\|(V^R - F^R) \Pi_{\leq t}\|_{\op}
\leq \sqrt{\frac{t(t+2)}{N}}. \qedhere
\]
\end{proof}

\noindent We have the following corollaries.
\begin{corollary}   \label{cor:FL_dagger_VL_dagger_FR_dagger_VR_dagger_close}
For any integer $t \geq 0$,
\[
\|(V^{L,\dagger} - F^{L,\dagger}) \Pi_{\leq t}\|_{\op}
\leq \sqrt{\frac{t(t+2)}{N}}
\quad \text{and} \quad
\|(V^{R,\dagger} - F^{R,\dagger}) \Pi_{\leq t}\|_{\op}
\leq \sqrt{\frac{t(t+2)}{N}}.
\]
\end{corollary}
\begin{proof}
Using the fact that $\|A\|_{\op} = \|A^\dagger\|_{\op}$ for any operator $A$, we have 
\begin{align*}
    \|(V^{L,\dagger} - F^{L,\dagger}) \Pi_{\leq t}\|_{\op} 
    = \|\Pi_{\leq t} (V^L - F^L) \|_{\op}.
\end{align*}
Next, since applying $V^L$ and $F^L$ only increase the size of each relation on the register $\reg{S}$ by one, we have 
\begin{align*}
    \Pi_{\leq t} (V^L - F^L) = (V^L - F^L) \Pi_{\leq t-1}.    
\end{align*}
Finally, the desired bound then follows from~\Cref{lem:FL_VL_FR_VR_close}. The second bound follows by the same argument.
\end{proof}

\noindent Finally, we have:

\begin{lemma}   \label{lem:F_V_close}
For any integer $t \geq 0$,
\[
\|(V - F) \Pi_{\leq t} \|_{\op}
\leq 8 \cdot \sqrt{\frac{(t+2)(t+4)}{N}}.
\]
\end{lemma}
\begin{proof}
Recall the definitions
\[
V = V^L \cdot (\id - V^R \cdot V^{R,\dagger}) + (\id - V^L \cdot V^{L,\dagger}) \cdot V^{R,\dagger}
\quad \text{and} \quad
F = F^L \cdot (\id - F^R \cdot F^{R,\dagger}) + (\id - F^L \cdot F^{L,\dagger}) \cdot F^{R,\dagger}.
\]
By the triangle inequality and the above expression, we have 
\begin{align*}
& \norm{(V - F) \Pi_{\leq t}}_{\op} \\
& \qquad \leq \norm{(V^L - F^L) \Pi_{\leq t}}_{\op}
+ \norm{(V^L V^R V^{R,\dagger} - F^L F^R F^{R,\dagger}) \Pi_{\leq t}}_{\op} \\
& \qquad \qquad + \norm{(V^L V^{L,\dagger} V^{R,\dagger} - F^L F^{L,\dagger} F^{R,\dagger}) \Pi_{\leq t}}_{\op}
+ \norm{(V^{R,\dagger} - F^{R,\dagger}) \Pi_{\leq t}}_{\op}.
\end{align*}
We will bound each term in the rest of the proof. From~\Cref{lem:FL_VL_FR_VR_close} and~\Cref{cor:FL_dagger_VL_dagger_FR_dagger_VR_dagger_close}, we can bound the first term and the last term by
\[
\|(V^L - F^L) \Pi_{\leq t}\|_{\op} \leq \sqrt{\frac{t(t+2)}{N}}
\quad \text{and} \quad
\|(V^{R,\dagger} - F^{R,\dagger}) \Pi_{\leq t}\|_{\op} \leq \sqrt{\frac{t(t+2)}{N}}.
\]
\noindent Next, by the triangle inequality, we can bound the second term by
\begin{align*}
& \|(V^L V^R V^{R,\dagger} - F^L F^R F^{R,\dagger}) \Pi_{\leq t}\|_{\op} \\
& \leq \|(V^L V^R V^{R,\dagger} - F^L V^R V^{R,\dagger})  \Pi_{\leq t}\|_{\op}
+ \|(F^L V^R V^{R,\dagger} - F^L F^R V^{R,\dagger}) \Pi_{\leq t}\|_{\op}
+ \|(F^L F^R V^{R,\dagger} - F^L F^R F^{R,\dagger}) \Pi_{\leq t}\|_{\op} \\
& = \|(V^L - F^L) V^R V^{R,\dagger}  \Pi_{\leq t}\|_{\op}
+ \|F^L (V^R - F^R) F^R V^{R,\dagger} \Pi_{\leq t}\|_{\op}
+ \|F^L F^R (V^{R,\dagger} - F^{R,\dagger}) \Pi_{\leq t}\|_{\op}.
\end{align*}
Since applying $V^R$ and $F^R$ can increase the size of the relation by at most one, while applying $V^\dagger_R$ and $F^\dagger_R$ never increases it, we can bound the above sum as follows:
\begin{align*}
& \leq \|(V^L - F^L) \Pi_{\leq t+1} V^R V^{R,\dagger}  \Pi_{\leq t}\|_{\op} \\
& \hspace{.2\textwidth} + \|F^L \Pi_{\leq t+2} (V^R - F^R) \Pi_{\leq t+1} F^R \Pi_{\leq t} V^{R,\dagger} \Pi_{\leq t}\|_{\op} \\
& \hspace{.4\textwidth} + \|F^L \Pi_{\leq t+2} F^R \Pi_{\leq t+1} (V^{R,\dagger} - F^{R,\dagger}) \Pi_{\leq t}\|_{\op} \\
& \leq \|(V^L - F^L) \Pi_{\leq t+1}\|_{\op} \\
& \hspace{.2\textwidth} + \|F^L \Pi_{\leq t+2}\|_{\op} \cdot \|(V^R - F^R) \Pi_{\leq t+1}\|_{\op} \cdot \|F^L \Pi_{\leq t}\|_{\op} \\
& \hspace{.4\textwidth}  + \|F^L \Pi_{\leq t+2}\|_{\op} \cdot \|F^R \Pi_{\leq t+1}\|_{\op} \cdot \|(V^{R,\dagger} - F^{R,\dagger}) \Pi_{\leq t}\|_{\op} 
\tag{by submultiplicity of operator norm}\\
& \leq 3 \cdot \sqrt{\frac{(t+2)(t+4)}{N}}.
\tag{by~\Cref{lem:FL_VL_FR_VR_close}, \Cref{cor:FL_dagger_VL_dagger_FR_dagger_VR_dagger_close} and the fact that $F^L$ and $F^R$ are contractions}
\end{align*}
Similarly, we can bound the third term by
\[
\|(V^L V^{L,\dagger} V^{R,\dagger} - F^L F^{L,\dagger} F^{R,\dagger}) \Pi_{\leq t}\|_{\op} 
\leq 3 \cdot \sqrt{\frac{(t+2)(t+4)}{N}}.
\]
\noindent Collecting the bounds, we have 
\[
\|(V - F) \Pi_{\leq t}\|_{\op}
\leq 8 \cdot \sqrt{\frac{(t+2)(t+4)}{N}}
\]
as desired.
\end{proof}
\noindent Using a symmetric argument, we also have:
\begin{lemma}   \label{lem:Fdagger_Vdagger_close}
For any integer $t \geq 0$,
\[
\|(V^{\dagger} - F^{\dagger}) \Pi_{\leq t}\|_{\op}
\leq 8 \cdot \sqrt{\frac{(t+2)(t+4)}{N}}.
\]
\end{lemma}

\begin{lemma}   \label{lem:path:FV}
For any adversary $\Adversary$ that makes $t$ forward queries and $t$ inverse queries,
\[
\norm{\ket{\Adversary^{F, F^{\dagger}}}_{\reg{ABST}} 
- \ket{\Adversary^{V, V^{\dagger}}}_{\reg{ABST}}}_2
= O \qty(\sqrt{\frac{t^4}{N}}).
\]
\end{lemma}
\begin{proof}
The lemma follows from~\Cref{lem:F_V_close,lem:Fdagger_Vdagger_close} and a sequence of hybrids that replace query to $F$ with $V$ one-by-one.
\end{proof}

\subsection{Missing Proofs in~\Cref{sec:path_recording}}    \label{sec:lemmas_F}

\begin{lemma}[\Cref{lem:domain_FL_FR}, restated]
For any integer $t \ge 0$,
\begin{align*}
    \|(F^{L,\dagger} F^L - \id) \Pi_{\le t}\| \le t/N
    \quad \text{and} \quad
    \|(F^{R,\dagger} F^R - \id) \Pi_{\le t}\| \le t/N \\
\end{align*}
\end{lemma}

\begin{proof}
For any $x, y \in [N], L \in \Rinj, R \in \RDdist$ such that $|L|+|R| \le t$, from~\Cref{eq:FL_dagger,eq:FR_dagger}, we have
\begin{align*}
F^{L,\dagger} \cdot F^L \cdot \ket{x}_{\reg{A}} \ket{L}_{\reg{S}} 
= \frac{N - |L|}{N} \cdot \ket{x}_{\reg{A}} \ket{L}_{\reg{S}}
\quad \text{and} \quad
F^{R,\dagger} \cdot F^R \cdot \ket{y}_{\reg{A}} \ket{R}_{\reg{T}} 
= \frac{N - |R|}{N} \cdot \ket{y}_{\reg{A}} \ket{R}_{\reg{T}}.
\end{align*}
Therefore, by~\Cref{lem:op_norm}, we have 
\begin{align*}
    \|(F^{L,\dagger} F^L - \id) \Pi_{\le t}\|
    = \max_{x,L \in \RIdist_{\le t}} \frac{|L|}{N} = \frac{t}{N}.
\end{align*}
The second bound follows from the same argument.
\end{proof}

We require the following fact, which is a consequence of~\cite[Definition~37, Claim~22, Equations~(11.22) and (11.26)]{MH25}.
\begin{fact}    \label{fact:operator_E}
There exist operators $E^L$ and $E^R$ that satisfy
\begin{itemize}
    \item For any $t \ge 0$, it holds that
    \begin{align*}
    \|(V^L-E^L) \Pi_{\leq t}\|_{\op} \leq \sqrt{t(t+2)/N}
    \quad \text{and} \quad 
    \|(V^R-E^R) \Pi_{\leq t}\|_{\op} \leq \sqrt{t(t+2)/N}.
    \end{align*}
    \item For any $\ell, r \ge 0$ such that $\ell + r \leq N$, it holds that
    \begin{align*}
    &E^L_{\ell,r} \cdot E^{L,\dagger}_{\ell,r}
    = \sum_{i \in [\ell+1]} 
    \bigg( \Pi_{\ell+1,\reg{S}} \cdot \Pi^{\sf EPR}_{\reg{A},\reg{S}^{(\ell+1)}_{\reg{Y},i}} \cdot \Pi_{\ell+1,\reg{S}} \bigg)
    \otimes \Pi_{r,\reg{T}}, \quad \text{and} \\
    &E^R_{\ell,r} \cdot E^{R,\dagger}_{\ell,r}
    = \sum_{i \in [r+1]}
    \Pi_{\ell,\reg{S}} \otimes 
    \bigg( \Pi_{r+1,\reg{T}} \cdot \Pi^{\sf EPR}_{\reg{A},\reg{T}^{(r+1)}_{\reg{X},i}} \cdot \Pi_{r+1,\reg{T}} \bigg).
    \end{align*}
\end{itemize}
\end{fact}

\noindent The operators $E^L$ and $E^R$ satisfy the following property.
\begin{lemma}   \label{lem:ELdagger_U_ER:zero}
For any integer $t \ge 0$ and any unitary $U$ acting non-trivially on the register $\reg{A}$,
\begin{align*}
\|E^{L,\dagger} U E^R \Pi_{\leq t}\|_{\op} \leq \sqrt{t(t+1)}/N
\quad \text{and} \quad
\|E^{R,\dagger} U E^L \Pi_{\leq t}\|_{\op} \leq \sqrt{t(t+1)}/N \, .
\end{align*}
\end{lemma}

\begin{proof}
By~\Cref{lem:op_norm_orthogonal}, we have
\begin{align}
\|E^{R,\dagger} U E^L \Pi_{\leq t}\|^2_{\op} 
= \max_{0 \le \ell, r \le t} \|E^{R,\dagger}_{\ell+1,r-1} U E^L_{\ell,r}\|^2_{\op}.
\label{eq:EL_ER_1}
\end{align}
We next bound the squared operator norm in the following. Using the fact $\|A\|^2_{\op} = \|A A^\dagger\|_{\op} = \|A^\dagger A\|_{\op}$, we have
\begin{align}
\eqref{eq:EL_ER_1} =
& \max_{0 \le \ell, r \le t} 
\|E^{R,\dagger}_{\ell+1,r-1} U E^L_{\ell,r} \cdot E^{L,\dagger}_{\ell,r} U^\dagger E^R_{\ell+1,r} \|_{\op} 
\notag \\
\le & \max_{0 \le \ell, r \le t} 
\sum_{i \in [\ell+1]} 
\bigg\|
E^{R,\dagger}_{\ell+1,r-1} \cdot
U \cdot 
\bigg( \Pi_{\ell+1,\reg{S}} \cdot
\Pi^{\sf EPR}_{\reg{A},\reg{S}^{(\ell+1)}_{\reg{Y},i}} \otimes \Pi_{r,\reg{T}} \cdot 
\Pi_{\ell+1,\reg{S}} \bigg) \cdot
U^\dagger \cdot
E^R_{\ell+1,r-1} \bigg\|_{\op}.
\label{eq:EL_ER_2}
\end{align}
Since $U$ only acts non-trivially on the register $\reg{A}$, we simplify the notation as follows:
\begin{align}
\eqref{eq:EL_ER_2}
= \max_{0 \le \ell, r \le t} 
\sum_{i \in [\ell+1]} 
\bigg\|
E^{R,\dagger}_{\ell+1,r-1} \cdot
U_{\reg{A}} \otimes \Pi_{\ell+1,\reg{S}} \cdot
\Pi^{\sf EPR}_{\reg{A},\reg{S}^{(\ell+1)}_{\reg{Y},i}} \otimes \Pi_{r,\reg{T}} \cdot 
U_{\reg{A}}^\dagger \otimes \Pi_{\ell+1,\reg{S}} \cdot 
E^R_{\ell+1,r-1} 
\bigg\|_{\op}.
\label{eq:EL_ER_3}
\end{align}
Using the facts that $\|A A^\dagger\|_{\op} = \|A^\dagger A\|_{\op}$ and $\Pi^{\sf EPR}_{\reg{A},\reg{S}^{(\ell+1)}_{\reg{Y},i}} \otimes \Pi_{r,\reg{T}}$ is a projector, we have
\begin{align}
& \eqref{eq:EL_ER_3} 
\notag \\
& = \max_{0 \le \ell, r \le t} 
\sum_{i \in [\ell+1]} 
\bigg\|
\Pi^{\sf EPR}_{\reg{A},\reg{S}^{(\ell+1)}_{\reg{Y},i}} \otimes \Pi_{r,\reg{T}} \cdot 
U_{\reg{A}}^\dagger \otimes \Pi_{\ell+1,\reg{S}} \cdot 
E^R_{\ell+1,r-1} \cdot
E^{R,\dagger}_{\ell+1,r-1} \cdot
U_{\reg{A}} \otimes \Pi_{\ell+1,\reg{S}} \cdot
\Pi^{\sf EPR}_{\reg{A},\reg{S}^{(\ell+1)}_{\reg{Y},i}} \otimes \Pi_{r,\reg{T}}
\bigg\|_{\op} 
\notag \\
& \le \max_{0 \le \ell, r \le t} 
\sum_{\substack{i \in [\ell+1] \\ j \in [r]}}
\bigg\|
\Pi^{\sf EPR}_{\reg{A},\reg{S}^{(\ell+1)}_{\reg{Y},i}} \otimes \Pi_{r,\reg{T}} \cdot 
U_{\reg{A}}^\dagger \otimes \Pi_{\ell+1,\reg{S}} \cdot
\Pi^{\sf EPR}_{\reg{A},\reg{T}^{(r)}_{\reg{X},j}} \cdot
U_{\reg{A}} \otimes \Pi_{\ell+1,\reg{S}} \cdot
\Pi^{\sf EPR}_{\reg{A},\reg{S}^{(\ell+1)}_{\reg{Y},i}} \otimes \Pi_{r,\reg{T}}
\bigg\|_{\op}
\label{eq:EL_ER_4}
\end{align}
Thus, we have
\begin{align*}
\eqref{eq:EL_ER_4} 
= & \max_{0 \le \ell, r \le t} 
\sum_{\substack{i \in [\ell+1] \\ j \in [r]}} 
\frac{1}{N^2}
\bigg\|
\Pi^{\sf EPR}_{\reg{A},\reg{S}^{(\ell+1)}_{\reg{Y},i}} 
\otimes 
\bigg(
\Pi_{r,\reg{T}} \cdot 
U_{\reg{T}^{(r)}_{\reg{X},j}}^\dagger \cdot
\Pi_{\ell+1,\reg{T}^{(r)}_{\reg{X},j}, \reg{S} \setminus \reg{S}^{(\ell+1)}_{\reg{Y},i}} \cdot
U_{\reg{T}^{(r)}_{\reg{X},j}} \cdot
\Pi_{r,\reg{T}}
\bigg)
\bigg\|_{\op} \\
\leq & \max_{0 \le \ell, r \le t} \frac{(\ell+1)r}{N^2} 
= \frac{(t+1)t}{N^2}.
\end{align*}
This completes the proof.
\end{proof}

\begin{lemma}[\Cref{lem:FLdagger_U_FR:zero}, restated]   
For any integer $t \ge 0$ and any unitary $U$ acting non-trivially on the register $\reg{A}$,
\begin{align*}
\|F^{L,\dagger} U F^R \Pi_{\leq t}\|_{\op} \leq 3\sqrt{t(t+2)/N}
\quad \text{and} \quad
\|F^{R,\dagger} U F^L \Pi_{\leq t}\|_{\op} \leq 3\sqrt{t(t+2)/N}\, .
\end{align*}
\end{lemma}

\begin{proof}
It immediately follows from~\Cref{lem:FL_VL_FR_VR_close,fact:operator_E,lem:ELdagger_U_ER:zero} 
together with the triangle inequality.
\end{proof}

\begin{lemma}[\Cref{lem:punc_S}, restated]
Let $\set{\cP_\tau}_\tau$ be collection of sets where the index $\tau$ ranges over $(y \in [N], L_1 \in \RIdist, R_1 \in \RDdist, L_2 \in \RIdist, R_2 \in \RDdist)$ and $\cP_\tau \subseteq [N]^3 \times [N]^{|L_2|} \times [N]^{|R_2|}$. Define the operator 
\begin{align*}
\Ospru^{\bullet}
& \colon \ket{y}_{\reg{A}} \ket{L_1}_{\reg{S}_1} \ket{R_1}_{\reg{T}_1} \ket{L_2}_{\reg{S}_2} \ket{R_2}_{\reg{T}_2} \\
& \mapsto 
\ket{y}_{\reg{A}} \;
\frac{1}{\sqrt{N^{|L_2|+|R_2|+3}}}
\sum_{   \substack{
    (\bfk,\bfz) \in \Sspru\left(\substack{L_1,L_2 \\ R_1,R_2}\right) \colon \\
    (\bfk,\bfz) \notin \cP_{y,L_1,R_1,L_2,R_2}
}   }
\ket{L^{(k_1,k_3)}_1 \cup L_2^{(k_2,\vec{z}_L)}}_{\reg{S}} 
\ket{R^{(k_1,k_3)}_1 \cup R_2^{(k_2,\vec{z}_R)}}_{\reg{T}}
\ket{k}_{\reg{K}}.
\end{align*}
If there exists $\delta \ge 0$ such that for any $\tau$,
\begin{align*}
    \frac{\big|\cP_{y,L_1,R_1,L_2,R_2} \cap \Sspru\left(\substack{L_1,L_2 \\ R_1,R_2}\right)\big|}{N^{|L_2|+|R_2|+3}} \le \delta,
\end{align*}
then
\begin{align*}
    \|\Ospru^{\bullet} - \Ospru\|_{\op} = \sqrt{\delta}.
\end{align*}
\end{lemma}

\begin{proof}
For any normalized state $\ket{\psi} = \sum_{y,L_1,R_1,L_2,R_2} \alpha_{y,L_1,R_1,L_2,R_2} \ket{y}_{\reg{A}} \ket{L_1}_{\reg{S}_1} \ket{R_1}_{\reg{T}_1} \ket{L_2}_{\reg{S}_2} \ket{R_2}_{\reg{T}_2}$, we have
\begin{align*}
\|(\Ospru^{\bullet} - \Ospru) \ket{\psi}\|_2
= \bigg\|
\sum_{   \substack{
    y,L_1,R_1,L_2,R_2 \\
    (\bfk,\bfz) \in \Sspru\left(\substack{L_1,L_2 \\ R_1,R_2}\right) \colon \\
    (\bfk,\bfz) \in \cP_{y,L_1,R_1,L_2,R_2}
}   }
\frac{\alpha_{y,L_1,R_1,L_2,R_2}}{\sqrt{N^{|L_2|+|R_2|+3}}}
\ket{y}_{\reg{A}}
\ket{L^{(k_1,k_3)}_1 \cup L_2^{(k_2,\vec{z}_L)}}_{\reg{S}} 
\ket{R^{(k_1,k_3)}_1 \cup R_2^{(k_2,\vec{z}_R)}}_{\reg{T}}
\ket{\bfk}_{\reg{K}}
\bigg\|_2.
\end{align*}
By~\Cref{lem:good_tuples_satisfy_conditions}, we can instead apply $\Osprus$ to the state and bound its squared norm as follows:
\begin{align*}
& \bigg\|
\sum_{   \substack{
    y,L_1,R_1,L_2,R_2 \\
    (\bfk,\bfz) \in \Sspru\left(\substack{L_1,L_2 \\ R_1,R_2}\right) \colon \\
    (\bfk,\bfz) \in \cP_{y,L_1,R_1,L_2,R_2}
}   }
\frac{\alpha_{y,L_1,R_1,L_2,R_2}}{\sqrt{N^{|L_2|+|R_2|+3}}}
\ket{y}_{\reg{A}}
\ket{L_1}_{\reg{S}_1} \ket{R_1}_{\reg{T}_1} 
\ket{L_2}_{\reg{S}_2} \ket{R_2}_{\reg{T}_2}
\ket{\vec{z}_L}_{\reg{Z_L}} \ket{\vec{z}_R}_{\reg{Z_R}}
\ket{\bfk}_{\reg{K}}
\bigg\|^2_2 \\
& = \sum_{y,L_1,R_1,L_2,R_2}
\sum_{   \substack{
    (\bfk,\bfz) \in \Sspru\left(\substack{L_1,L_2 \\ R_1,R_2}\right) \colon \\
    (\bfk,\bfz) \in \cP_{y,L_1,R_1,L_2,R_2}
}   }
\bigg|\frac{\alpha_{y,L_1,R_1,L_2,R_2}}{\sqrt{N^{|L_2|+|R_2|+3}}}\bigg|^2 \\
& = \sum_{y,L_1,R_1,L_2,R_2}
|\alpha_{y,L_1,R_1,L_2,R_2}|^2 \cdot \frac{\big|\cP_{y,L_1,R_1,L_2,R_2} \cap \Sspru\left(\substack{L_1,L_2 \\ R_1,R_2}\right)\big|}{N^{|L_2|+|R_2|+3}} \\
& \le \max_{y,L_1,R_1,L_2,R_2} \frac{\big|\cP_{y,L_1,R_1,L_2,R_2} \cap \Sspru\left(\substack{L_1,L_2 \\ R_1,R_2}\right)\big|}{N^{|L_2|+|R_2|+3}} \\
& \le \delta. \tag*{\qedhere}
\end{align*}
\end{proof}

\section{Missing Proofs in~\Cref{sec:spru}}
\label{sec:missing_defs_proofs}

\subsection{Proof of~\Cref{lem:good_tuples_satisfy_conditions}}

\begin{lemma}[\Cref{lem:good_tuples_satisfy_conditions}, restated]
For any $L_1, L_2 \in \Rinj$, $R_1, R_2 \in \RDdist$, every tuple in $\Sspru\left(\substack{L_1,L_2 \\ R_1,R_2}\right)$ satisfies all conditions in~\Cref{lem:conditions_robust_decodability} and is therefore robustly decodable.
\end{lemma}

\begin{proof}[Proof of~\Cref{lem:good_tuples_satisfy_conditions}]
Recall~\Cref{def:good_tuples}. For convenience, we color the conditions according to the properties they enforce in the following definition: \bcolor{distinctness}, \rcolor{disjointness}, and \gcolor{no extra $k_2$-correlated pairs}.
\begin{enumerate}
    \item $k_1 \notin \rcolor{\Big( \Dom(L_1) \oplus \Dom(L_2) \Big)} \cup \rcolor{\Big( \Dom(R_1) \oplus \Dom(R_2) \Big)}$
    \item \(
        k_2 \notin 
        \gcolor{ \Bigg( \bigg( \Big( \Dom(L_1) \oplus k_1 \Big) \cup \Dom(L_2) \bigg) \oplus \bigg( \Big( \Im(L_1) \oplus k_3 \Big) \cup \Im(L_2) \bigg) \Bigg) } \\
        \cup 
        \gcolor{ \Bigg( \bigg( \Big( \Dom(R_1) \oplus k_1 \Big) \cup \Dom(R_2) \bigg) \oplus \bigg( \Big( \Im(R_1) \oplus k_3 \Big) \cup \Im(R_2) \bigg) \Bigg) }
    \)
    \item $k_3 \notin \rcolor{ \Big( \Im(L_1) \oplus \Im(L_2) \Big) } \cup \rcolor{ \Big( \Im(R_1) \oplus \Im(R_2) \Big) }$
    \item $\vec{z}_L\in \bcolor{ [N]^{|L_2|}_{\dist} }$ and $\vec{z}_R\in \bcolor{ [N]^{|R_2|}_{\dist} }$
    \item $\set{\vec{z}_L}$ and $\rcolor{ \Big( \Im(L_1) \oplus k_3 \Big) } \cup \bcolor{ \Im(L_2) } \cup \gcolor{ \Bigg( \bigg( \Big( \Dom(L_1) \oplus k_1 \Big) \cup \Dom(L_2) \bigg) \oplus k_2 \Bigg) }$ are disjoint.
    \item $\set{\vec{z}_R}$ and $\rcolor{ \Big( \Im(R_1) \oplus k_3 \Big) } \cup \bcolor{ \Im(R_2) } \cup \gcolor{ \Bigg( \bigg( \Big( \Dom(R_1) \oplus k_1 \Big) \cup \Dom(R_2) \bigg) \oplus k_2 \Bigg) }$ are disjoint.
\end{enumerate}
Each condition can be verified directly.
\end{proof}

\end{document}